\documentclass[]{article}
\usepackage[latin1]{inputenc}
\usepackage[english]{babel}
\usepackage{amsmath}
\usepackage{amsfonts}
\usepackage{amssymb}
\usepackage{graphicx}
\usepackage{amsthm}
\usepackage{stmaryrd}
\usepackage{proof}
\usepackage{enumitem}
\usepackage{adjustbox}
\usepackage{tikz}
\usepackage{lscape}
\usepackage{appendix}

\theoremstyle{definition}
\newtheorem{theorem}{Theorem}[section]
\theoremstyle{definition}
\newtheorem{lemma}[theorem]{Lemma}
\theoremstyle{definition}

\theoremstyle{definition}
\newtheorem{proposition}[theorem]{Proposition}
\theoremstyle{definition}

\theoremstyle{definition}
\newtheorem{definition}{Definition}[section]
\theoremstyle{definition}
\newtheorem{example}{Example}[section]
\theoremstyle{definition}
\newtheorem{remark}{Remark}[section]
\theoremstyle{definition}
\newtheorem{notation}{Notation}[section]
\theoremstyle{definition}
\newtheorem{note}{Notation}[definition]

\usepackage[affil-it]{authblk}

\newcommand{\VV}{\mathbb{V}}

\newcommand{\Nes}{\mathsf{Nes}}
\newcommand{\C}{\mathsf{C}}
\newcommand{\emp}{\mathsf{0}}
\newcommand{\N}{\mathsf{N}}
\newcommand{\T}{\mathsf{T}}
\newcommand{\W}{\mathsf{W}}

\newcommand{\PCL}{\mathbb{PCL}}

\newcommand{\NN}{\mathbb{N}}
\newcommand{\TT}{\mathbb{T}}
\newcommand{\WW}{\mathbb{W}}
\newcommand{\CC}{\mathbb{C}}
\newcommand{\AAA}{\mathbb{A}}
\newcommand{\UU}{\mathbb{U}}

\newcommand{\LWk}{\mathsf{Wk_L}}

\newcommand{\Cut}{\mathsf{cut}}
\newcommand{\init}{\mathsf{init}}

\newcommand{\fu}{\Vdash^\forall}
\newcommand{\fe}{\Vdash^\exists}

\newcommand{\barra}[1]{\Vdash_{#1}}
\newcommand{\lab}[1]{\mathbf{G3#1}}
\newcommand{\Lfe}{\mathsf{L\fe}}
\newcommand{\Rfe}{\mathsf{R\fe}}
\newcommand{\Lfu}{\mathsf{L\fu}}
\newcommand{\Rfu}{\mathsf{R\fu}}
\newcommand{\Rcon}{\mathsf{R>}}
\newcommand{\Lcon}{\mathsf{L>}}
\newcommand{\Rbar}{\mathsf{R|}}
\newcommand{\Lbar}{\mathsf{L|}}
\newcommand{\Ref}{\mathsf{Ref}}
\newcommand{\Trans}{\mathsf{Tr}}
\newcommand{\Lc}{\mathsf{L\subseteq}}
\newcommand{\Single}{\mathsf{Single}}
\newcommand{\Repl}[1]{\mathsf{Repl_{#1}}}
\newcommand{\Unif}[1]{\mathsf{U_{#1}}}
\newcommand{\Abs}[1]{\mathsf{A_{#1}}}

\newcommand{\Rright}{\mathsf{R\rightarrow}}
\newcommand{\Lright}{\mathsf{L\rightarrow}}

\newcommand{\Rand}{\mathsf{R\wedge}}
\newcommand{\Land}{\mathsf{L\wedge}}
\newcommand{\Ror}{\mathsf{R\lor}}
\newcommand{\Lor}{\mathsf{L\lor}}

\newcommand{\Mon}{\mathsf{Mon\forall}}

\newcommand{\g}{\prec_g}
\newcommand{\ww}{\prec}
\newcommand{\wclos}{\prec^*}

\newcommand{\Ctr}{\mathsf{Ctr}}
\newcommand{\Ctrl}{\mathsf{Ctr_L}}
\newcommand{\Ctrr}{\mathsf{Ctr_R}}
\newcommand{\Wkl}{\mathsf{Wk_L}}
\newcommand{\Wkr}{\mathsf{Wk_R}}

\newcommand{\Wk}{\mathsf{Wk}}

\newcommand{\Lbot}{\bot_{\mathsf{L}}}

\newcommand{\fpcl}{\mathcal{L}}

\newcommand{\val}{\vDash_{\mathcal{N}}}

\newcommand{\Maxc}{\mathsf{MAX}}
\newcommand{\cons}[2]{{#1}^{#2}}
\newcommand{\clesser}[1]{\leqslant_{#1}}
\newcommand{\worldsC}{\mathcal{W}}
\newcommand{\valC}{\mathcal{V}}
\newcommand{\spheres}[2]{\nu^{#1}_{#2}}
\newcommand{\neigh}[1]{\mathcal{N}({#1})}

\newcommand{\canonicalex}[1]{\mathfrak{M}_{#1}}
\newcommand{\modelex}[1]{\mathcal{M}_{#1}}
\newcommand{\modelNeigh}{\mathcal{M}_{N}}

\newcommand{\worldsCW}{\mathcal{W}^{\mathsf{cw}}}

\newcommand{\neighFunctW}{\mathcal{N}^{\mathsf{cw}}}
\newcommand{\neighW}[1]{\mathcal{N}^{\mathsf{cw}}(\mathsf{#1})}

\newcommand{\axMod}{(MOD)}
\newcommand{\univ}[1]{\mathsf{Univ}({#1})}

\newcommand{\worldsCU}{\mathcal{W}^{\mathsf{u}}}
\newcommand{\valCU}{\mathcal{V}^{\mathsf{u}}}
\newcommand{\neighFunctU}{\mathcal{N}^{\mathsf{u}}}
\newcommand{\neighU}[1]{\mathcal{N}^{\mathsf{u}}(\mathsf{#1})}

\newcommand{\grafo}{\mathcal{G}_\mathcal{B}}

\newcommand{\axiom}[1]{\mathcal{H}_{\mathsf{#1}}}
\newcommand{\PP}{\mathbb{P}}
\newcommand{\derivable}[1]{\vdash_{\mathsf{#1}}}
\newcommand{\valid}[1]{\vDash_{\mathsf{#1}}}

\newcommand{\spheresCW}[2]{\mu^{#1}_{#2}}

\newcommand{\LconT}{\mathsf{L}>^\star}
\newcommand{\size}[1]{|#1|}

\title{Uniform labelled calculi for preferential conditional logics based on neighbourhood semantics\footnote{Submitted for publication to the Journal of Logic and Computation. The article will be revised after referees reports.
This work was partially supported by the Project TICAMORE ANR-16-CE91-0002-
		01 and by WWTF project MA 16-28.}
}
\author{Marianna Girlando$^1$, Sara Negri$^2$, Nicola Olivetti$^3$}
\date{}

\affil{ 
	\small{
$^1$ Inria Saclay - Ile-de-France \& LIX, Ecole Polytechnique, France.\\

$^2$ Departement of Philosophy, University of Helsinki, Finland.\\

$^3$ Aix Marseille Univ, Université de Toulon, CNRS, LIS, Marseille, France.
}
}

\begin{document}
	
	\maketitle

\begin{abstract}
	The  preferential conditional logic $ \PCL $, introduced by Burgess, and
	its extensions are studied. First, a natural semantics
	based on neighbourhood models, which generalise Lewis' sphere models for
	counterfactual logics, is proposed.   Soundness and completeness of $ \PCL $ and its extensions
	with respect to this class of models are proved directly.
	Labelled sequent calculi for all logics of the family are then introduced.
	The calculi are modular and have standard proof-theoretical properties,
	the most important of which is  admissibility of cut, that entails a syntactic proof of  completeness of the calculi. 
	By adopting a general strategy, root-first proof search terminates, thereby providing a decision procedure for $ \PCL $ and its extensions. 
	Finally, the semantic completeness of the calculi is established:
	from a finite branch in a  failed proof attempt it is possible  to extract a finite
	countermodel of the root sequent. The latter result gives a
	constructive proof of the finite model property of all the logics considered.
\end{abstract}

\section{Introduction}

Conditional logics have been studied from a philosophical viewpoint since the 60's,  with seminal works by, among other, Lewis, Nute, Stalnaker, Chellas, Pollock  and Burgess.\footnote{Cf. \cite{lewis1973}, \cite{stalnaker1968theory},  \cite{stalnaker1970semantic}, \cite{chellas1975basic}, \cite{pollock1981refined}, \cite{burgess1981quick}, \cite{veltman1985logic}.}  In all cases, the aim is to represent a kind of hypothetical implication $A > B$ different from classical material implication, but also from other non-classical implications, such as the intuitionistic one.  

There are mainly two  kinds of interpretations of  a conditional  $A > B$. The first is  hypothetical/counterfactual: ``If A were the case then B would be the case'', while the second is prototypical: ``Typically (normally) if A then B'', or ``B holds in most  normal/typical cases in which A 
holds''. Applications of conditional logics to computer science, more specifically to artificial intelligence and knowledge representation,  have followed these two interpretations. The hypothetical/counterfactual interpretation has lead to the study of the relation of conditional logics with the notion of \emph{belief change}, which has lead to the crucial issue of the Ramsey Test. The prototypical interpretation  has found an interest in the formalisation of default and non-monotonic reasoning (the well-known KLM systems) and has  some relation with probabilistic reasoning. The range of conditional logics is actually more extensive, comprising also deontic and causal interpretations. 

All interpretations of the conditional operator agree on the rejection of some properties of material implication, along with properties of other non-classical implications, such as the intuitionistic one. These undesirable properties are strengthening,  $A > B$ \emph{implies}  $(A\land C) > B$;  transitivity, $A> B$ \emph{and} $B > C$ \emph{imply} $A > C$, and contraposition, $A > B$ \emph{implies} $\neg B > \neg A$. 

The semantics of conditional logics is defined in terms of various kinds of possible-world models, most of them comprising a notion of preference, comparative similarity or choice among worlds. Intuitively, a conditional $A > B$ is true at a world $x$ if $B$ is true 
in all the worlds  most normal/similar/close to $x$ in which $A$ is true.
In contrast with the situation in
standard modal logic, there is 
no unique semantics for conditional logics.

\

In this paper we consider the  conditional logic $\PCL$  (Preferential Conditional Logic), one of the fundamental systems of conditional logics.   An axiomatization of $\PCL$  (and the respective completeness proof) has been originally presented in the seminal work by  Burgess in \cite{burgess1981quick}, where the system is called S, and then by Veltman \cite{veltman1985logic}. 
Logic $\PCL$  generalises Lewis' basic logic of counterfactuals, and its flat fragment corresponds to the preferential logic P of non-monotonic reasoning proposed by Kraus, Lehmann and Magidor \cite{kraus1990nonmonotonic}. 

The logic takes its name, $\PCL$,  from its original semantics, defined in terms of \emph{preferential models}. In these models,  every world $x$  is associated with a set of accessible worlds $W_x$ and a \emph{preference} relation $\leq_x $ on this set; the intuition is that this relation assesses the relative normality/similarity of pairs of worlds with respect to $x$. Roughly speaking,   a conditional $A > B$ is forced by  $x$ if $B$ is true in all accessible worlds (that is, worlds  in $W_x$) where $A$ holds and that  are  most ``normal" with respect to $x$, where their normality is assessed  by the  relation $\leq_x$\footnote{According to some interpretations, normality means minimality with respect to $\leq_x$.}. 


In this paper we present an alternative semantics for $\PCL$    based on  \emph{neighbourhood models}.   Neighbourhood semantics has been successfully employed to analyse non-normal modal logics \cite{chellas1975basic}, as their semantics cannot be defined in terms of ordinary relational Kripke models. In neighbourhood models, every world $x$ is  equipped with a set of neighbourhoods $N(x)$ and each $\alpha\in N(x)$  is a non-empty set of worlds.  The general intuition is that each neighbourhood $\alpha \in N(x)$  represents a state of information/knowledge/affair to be taken into account in evaluating the truth of modal formulas at world $x$.  In the conditional context, neighbourhood inclusion can be understood as follows:
if  $\alpha, \beta \in N(x)$ and $\beta\subseteq \alpha$, then worlds in $\beta$ are at least as plausible/normal as worlds in $\alpha$.

It turns out that neighbourhood models provide a very natural semantics for $\PCL$. This semantics abstracts away from the details of the preference relations and, moreover,  the definition of the conditional can be seen as a simple modification of the \emph{strict implication} operator, avoiding the unwanted properties of strengthening, transitivity and contraposition. 
The strict implication demands that each $\alpha\in N(x)$  ``validates" the implication  $A \to B$. The truth condition for the conditional only requires that,  for all   $\alpha\in N(x)$ containing an $A$ -world, there is a smaller neighbourhood $\beta\subseteq \alpha$  non-vacuously validating  the  implication  $A \to B$, where non-vacuously means that $\beta$ must contain an $A$-world.
No further properties or structure of neighbourhood models are needed. 

The use of neighbourhood models for analysing conditional logics is not a novelty: Lewis' sphere models for counterfactual logics belong to this approach. However, the crucial property of sphere models is that neighbourhoods (e.g. spheres) are \emph{nested}: given $\alpha, \beta \in N(x)$, either $\alpha\subseteq \beta$ or $\beta\subseteq \alpha$. This property entails that worlds belonging to  $\bigcup N(x)$ can be always be compared according to their level of normality\footnote{In models where minimal spheres always exist, the nesting property is equivalent to the existence of a ranking function  $r_x$  defined for every world $ x $. The function $r_x(y)$ evaluates the level of normality of each world $y\in W_x$ with respect to $x$.}. This assumption is controversial in some contexts such as belief revision \cite{girard2007onions} and non-monotonic reasoning. 
The logic $\PCL$ is more general: its neighbourhood models do not assume nesting of neighbourhoods, whence worlds in $\bigcup N(x)$ are not necessarily comparable with respect to their level of normality. 

Although $\PCL$ is the basic system we consider in this paper, stronger systems can be obtained by assuming properties of neighbourhood models: normality, total reflexivity, weak centering, centering, uniformity and absoluteness. These conditions are analogous to the ones considered by Lewis for sphere models, and give rise to a total of 15 preferential systems.

\

The Hilbert axiomatization of $\PCL$ is given by adding to 
the smallest conditional logic CK  three axioms, namely, (ID), (CM) and (OR). The family of preferential logics is obtained by adding axioms in correspondence with the semantic properties mentioned above.

In sharp contrast with the simplicity of its Hilbert axiomatization, the proof theory of $\PCL$ and its extensions is largely unexplored. To the best of our knowledge, the only existing proof systems for $\PCL$ can be found in  \cite{giordano2009tableau,schroder2010optimal} and, more recently, in \cite{nalon2018resolution,girlando2019uniform}. All of them are based on preferential semantics, and the last two cover only logic $\PCL$ and none of the extensions\footnote{For a more detailed discussion on the literature, refer to section 8.}.

Building on the neighbourhood semantics, we define labelled sequent calculi for $\PCL$ and its extensions\footnote{Some results of this work have been preliminarily presented in \cite{negri2015analytic}.}. The calculi make use of both world and neighbourhood labels to encode the relevant features of the semantics into the syntax. All  calculi are \emph{standard}, meaning that each connective is handled exactly  by  dual left and right rules, justified through a clear meaning explanation. 
As a special feature,  a new operator, $\mid$, is introduced for translating the meaning explanation of the conditional operator into sequent rules. 
Moreover, the calculi are \emph{modular}, to the extent that logical rules are the same for all systems, while relational rules for neighbourhood and world labels are added to define calculi for extensions.
We  do not consider explicitly the family of Lewis' logics, for which several internal and labelled calculi  exist. Nonetheless the present framework can be adapted to cover these systems as well. 

In addition to 
simplicity and modularity, 
the calculi 
have
strong proof theoretical properties, such as height-preserving invertibility and admissibility of  contraction and cut. 

We show that the calculi are terminating under the adoption of a uniform proof search strategy, obtaining thereby a decision procedure for (almost) all logics of the $\PCL$ family. However, since the logics in this family belong to different complexity classes \cite{halpern1994complexity}, the uniform strategy will be  unavoidably far from optimal.  


We also prove semantic completeness of the calculus: from a failed proof of a formula it is possible to extract a \emph{finite} neighbourhood countermodel, built from a branch of the attempted proof. 
This result provides a constructive proof of the finite model property for each logic of the  $\PCL$ family  with respect to the neighbourhood semantics. 

\

The paper is organised as follows: In Section 2, the family of $\PCL$ logics  and the neighbourhood semantics is introduced. Section 3 shows completeness of $\PCL$ and its extensions with respect to the neighbourhood semantics. In Section 4, we introduce  labelled sequent calculi for family of preferential logics. In Section 5 we prove the main syntactic properties of the calculi, including admissibility of cut, thereby obtaining a syntactic proof of the their completeness. 
In Section 6, a decision procedure for the logics is presented. In Section 7, we present a proof of semantic completeness for the calculi, by extracting a countermodel form failed proof search. Finally, Section 8 discusses some related work.

\section{Preferential logics and neighbourhood semantics}
In this section we introduce the family of preferential conditional logics. 

\begin{definition} 
	The set of well formed formulas of $ \PCL $ and its extensions is defined by means of the following grammar, for $ p \in Atm $ propositional variable and $ A, B \in \fpcl $: 
	$$ \fpcl :: = p \mid \bot \mid A \wedge B \mid A\lor B  \mid A \rightarrow B \mid A > B .$$
\end{definition}

Preferential conditional logic $ \PCL $ is the basic system of the family; extensions of $ \PCL $ are obtained by adding to the basic system the axioms for \textit{normality}, \textit{total reflexivity}, \textit{weak centering}, \textit{centering}, \textit{uniformity} and \textit{absoluteness}. The resulting 15 logics are represented in the lattice of Figure \ref{fig:logics}.

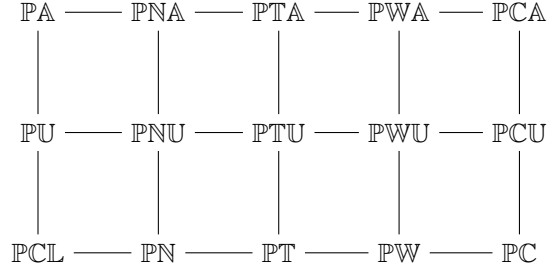
\begin{figure}[h]
	\begin{center}
		\begin{tikzpicture}[scale=0.8]
		\node (x) at  (-2,-2){};
		\node (y) at  (2,-2){};
		\node (z) at  (2,2){};
		\node (k) at  (-2,2){};
		
		\node (PCL) at  (-2,-2)  { $\mathbb{PCL} $};
		\node (PU) at  (-2,0)  { $\mathbb{PU} $};
		\node (PA) at  (-2,2)  { $\mathbb{PA} $};
		\node (PN) at  (0,-2)  { $\mathbb{PN} $};
		\node (PNU) at  (0,0)  { $\mathbb{PNU} $};
		\node (PNA) at  (0,2)  { $\mathbb{PNA} $};
		\node (PT) at  (2,-2)  { $\mathbb{PT} $};
		\node (PTU) at  (2,0)  { $\mathbb{PTU} $};
		\node (PTA) at  (2,2)  { $\mathbb{PTA} $};
		\node (PW) at  (4,-2)  { $\mathbb{PW} $};
		\node (PWU) at  (4,0)  { $\mathbb{PWU} $};
		\node (PWA) at  (4,2)  { $\mathbb{PWA} $};
		\node (PC) at  (6,-2)  { $\mathbb{PC} $};
		\node (PCU) at  (6,0)  { $\mathbb{PCU} $};
		\node (PCA) at  (6,2)  { $\mathbb{PCA} $};
		
		\draw (PCL) --(PN);
		\draw (PN) --(PT);
		\draw (PT) --(PW);
		\draw (PW) --(PC);
		\draw (PU) --(PNU);
		\draw (PNU) --(PTU);
		\draw (PTU) --(PWU);
		\draw (PWU) --(PCU);
		\draw (PA) --(PNA);
		\draw (PNA) --(PTA);
		\draw (PTA) --(PWA);
		\draw (PWA) --(PCA);
		
		\draw (PCL) --(PU);
		\draw (PN) --(PNU);
		\draw (PT) --(PTU);
		\draw (PW) --(PWU);
		\draw (PC) --(PCU);
		\draw (PU) --(PA);
		\draw (PNU) --(PNA);
		\draw (PTU) --(PTA);
		\draw (PWU) --(PWA);
		\draw (PCU) --(PCA);
		\end{tikzpicture}
	\end{center}
	\caption{The preferential family of conditional logics}
	\label{fig:logics}
\end{figure}

The axiomatic presentation of $ \PCL $ and its extensions is given in Figure \ref{axiom_table}. Propositional axioms and rules are standard. Given a logic $ \mathbb{K} $ of the preferential family, we denote its axiom system as $ \axiom{K} $, and derivability of a formula $ F $ in the axiom system as $ \derivable{K} F $.
\begin{figure}[h!]
	\begin{adjustbox}{max width = \textwidth}
		\begin{tabular}{l l}
			\hline
			&\\
			(RCEA) {\Large $\frac{A \leftrightarrow B}{(A> C)\leftrightarrow(B> C)}$} &
			(RCK)   {\Large $\frac{A  \rightarrow B}{(C > A)\rightarrow (C> B)}$}\\[1.5ex]
			(ID)  $A > A$  &
			(R-And)  $ (A > B) \land (A> C) \rightarrow   ( A > (B \land C))$ \\[1.5ex]
			(CM)  $ (A > B) \land (A > C) \rightarrow ((A\land B) > C)$ &
			(OR) $(A > C) \land (B > C) \rightarrow ((A\lor B) > C)$   \\[1ex]
			(N)  $ \lnot (\top > \bot) $ &
			(T) $ A \rightarrow  \lnot (A > \bot) $  \\[1ex] 	
			(W)  $ (A>B) \rightarrow (A\rightarrow B) $&
			(C) $ (A\wedge B) \rightarrow (A>B) $   \\[1ex]
			(U$_1$)  $  (\lnot A > \bot ) \rightarrow \lnot ( \lnot A > \bot ) > \bot $ &
			(U$_2$)  $ \lnot (A > \bot ) \rightarrow( (A > \bot ) > \bot)  $  \\[1ex]
			(A$_1$)  $ (A> B) \rightarrow ( C>(A>B) ) $ &
			(A$_2$)  $ \lnot(A> B) \rightarrow ( C> \lnot (A>B) )$  \\[1ex]
			\hline
		\end{tabular}
	\end{adjustbox}
	
	\vspace{0.3cm}
	
	\begin{adjustbox}{max width = \textwidth}
		\begin{tabular}{c}
			$ \axiom{PCL} =  $ \{(RCEA), (RCK), (ID), (R-And), (CM), (OR)\};\\[0.3ex]
			$ \axiom{PN}= \axiom{PCL} $ + (N); ~  $ \axiom{PT}= \axiom{PN} $ + (T); ~ $ \axiom{PW}= \axiom{PT} $ + (W); ~ $ \axiom{PC}= \axiom{PW} $ + (C);\\[0.3ex]
			$ \axiom{PU}= \axiom{PCL} $ + (U$ _1 $)+(U$ _2 $); ~	$ \axiom{PNU}= \axiom{PU} $ + (N); ~  $ \axiom{PTU}= \axiom{PNU} $ + (T); \\[0.3ex]
			$ \axiom{PWU}= \axiom{PTU} $ + (W); ~ $ \axiom{PCU}= \axiom{PWU} $ + (C); \\[0.3ex]
			$ \axiom{PA}= \axiom{PCL} $ + (A$ _1 $)+(A$ _2 $); ~	$ \axiom{PNA}= \axiom{PA} $ + (N); ~  $ \axiom{PTA}= \axiom{PNA} $ + (T); \\[0.3ex]
			$ \axiom{PWA}= \axiom{PTA} $ + (W); ~ $ \axiom{PCA}= \axiom{PWA} $ + (C). \\[0.3ex]
		\end{tabular}
	\end{adjustbox}
	
	\caption{Axiomatization of $ \PCL $}
	\label{axiom_table}
\end{figure}

The following proposition contains some theorems of $\PCL$ that will be (tacitly) used in the following. The first four are well-known axioms, respectively called (RT), (MOD), (DT), and (CSO)  in the literature. Axiom (DT) is equivalent to (OR), and from (DT) axiom (RT) is derivable. Axiom (CSO) is equivalent to (CM)+(RT). The proof of the last three axioms is given in \cite{kraus1990nonmonotonic}. 

\begin{proposition} \label{lemma:derivable_ax}
	The following formulas are derivable in  $\PCL $:
	\begin{enumerate} 
		\item (RT) ~ $ (A > B) \land ((A\land B) > C) \rightarrow (A > C) $;
		\item (MOD) ~ $(A > \bot) \to (B > \lnot A)$; 
		\item (DT) ~ $((A \land B) > C) \to (A > (B \to C))$;
		\item  (CSO) ~ $(A > B) \land (B > A)) \to ((A > C) \to (B> C))$;
		\item $ ((A\lor B) > A) \wedge ((B\lor C) >B) \rightarrow ((A\lor C) >A) $;
		\item $ ((A\lor B)> A) \wedge ((B\lor C)>B ) \rightarrow A> (C\rightarrow B)) $;
		\item $ ((A\lor B)> A) \wedge (B> C) \rightarrow ( A> (B\rightarrow C))  $.
	\end{enumerate}
\end{proposition}

The semantics of $ \PCL $ is usually defined in terms of preferential models, as explained in the Introduction. Here we define an alternative semantics in terms of neighbourhood models. 

\begin{definition}\label{def:neighbourhood_models}
	A \textit{neighbourhood model} is a structure $ \modelNeigh = \langle W, N, \llbracket \  \rrbracket  \rangle $
	where:
	\begin{itemize}
		\item $ W $ is a non empty set of elements, the \textit{possible worlds};
		\item $ N:  W \rightarrow \mathcal{P}(\mathcal{P}(W)  ) $ is the \textit{neighbourhood function}, which associates to each $ x \in W $ a set $ N(x)\subseteq \mathcal{P}(W) $, called a \textit{system of neighbourhood};
		\item $\llbracket \  \rrbracket: \mbox{\it Atm} \rightarrow \mathcal{P}(W) $ is the propositional evaluation.		
	\end{itemize}
	The elements of $ N(x) $ are called \textit{neighbourhoods}, and are denoted by lowercase Greek letters.
	For all $ x \in W $, we assume the neighbourhood function to satisfy the property of  \textit{non-emptiness}: For each $ \alpha \in N(x) $, $ \alpha $ is non-empty.	
\end{definition}


\begin{notation}
	The symbol $ \Vdash $ is used to denote the forcing (or truth) of a formula at a world of a model: $ x \Vdash B $ means that $B$ is true at $x$. Given a neighbourhood $\alpha$, 
	we use $ \alpha \fe B $  as a shorthand for \textit{there exists  $ y \in \alpha $ such that $ y \Vdash B $},
	and $ \alpha \fu B $ as a shorthand for  \textit{for all  $ y \in \alpha $ it holds that $ y \Vdash B $}.
\end{notation}

\noindent  Before giving its formal definition, we give an intuitive motivation of the truth condition for the conditional operator in neighbourhood semantics.  Suppose we want to define a conditional operator more fine-grained than the material implication, and suitable for an hypothetical, non-monotonic, or plausible interpretation. As a first attempt, we can  define a kind of strict implication, in analogy to the corresponding notion  in normal modal logic:
\begin{quote}
	(1) $x \Vdash A > B$ \textit{iff for all $\alpha \in N(x)$ it holds $\alpha \fu A \to B$}.
\end{quote}
However, this definition is not suitable for the conditional operator, as it would satisfy the unwanted properties of strengthening (or monotonicity), transitivity, and contraposition.
An equivalent, slightly redundant, formulation of  (1) consists in a restriction to neighbourhoods that contain $A$-worlds: 
\begin{quote}
	(2) $x \Vdash A > B$ \textit{iff for all $\alpha \in N(x)$, if $\alpha \fe A$ then $\alpha \fu A \to B$}. 
\end{quote}
Thus, for every $\alpha \in N(x)$, if $\alpha$ contains an $A$-world, we require  that $\alpha \fu A \to B$. The latter condition is too strong: in the intended interpretation, and in particular in the non-monotonic reading, the  conditional should tolerate exceptions. Thus, instead of requiring $A \to B$ to be verified by the \emph{whole} $\alpha$, we only demand the formula to be verified by a sub-neighbourhood $\beta$ of $\alpha$. 
\begin{quote}
	(3) $x \Vdash A > B$ \textit{iff for all $\alpha \in N(x)$, if $\alpha \fe A$ then there exists a $\beta \in N(x)$, with $\beta\subseteq \alpha$ such that  $\beta \fu A \to B$.}
\end{quote}
Here, however, there is still a problem: the condition on $\beta$ could be vacuously satisfied by choosing a $\beta$ that does not contain any $A$-world (at least whenever $A\not=\top$). To rule out this case, we modify (3) as follows:
\begin{quote}
	(4) $x \Vdash A > B$ \textit{iff for all $\alpha \in N(x)$, if $\alpha \fe A$ then there exists $\beta \in N(x)$, with $\beta\subseteq \alpha$ such that  $\beta \fe A$ and $\beta \fu A \to B$.}
\end{quote}
Definition (4) is the truth definition of conditional
adequate to formalize the logics of the preferential family.

\begin{definition}
	For any formula  $F\in  \fpcl$ and $x\in W$,  truth of a formula in a model, in symbols $ x \Vdash F $, is defined as follows. For atoms $p$, $ x \Vdash p $ if $x\in \llbracket p\rrbracket$; truth conditions for Boolean combinations are the standard ones; the truth condition for the  conditional operator is (4).
	
	We say that a formula $F$ is valid in $ \modelNeigh$  if for all $ x\in W$, $ x \Vdash F $. We say that a formula if valid in the class of all neighbourhood models (resp. in a class of models $\cal K$)  
	if for all neighbourhood models $\modelNeigh$ (resp. in $\cal K$) it holds that $F$ is valid in $\modelNeigh$; this will be denoted by  $\val F$ (resp. $\valid{K} F$).
\end{definition}

\begin{definition}\label{def:extensions_neigh}
	Extensions of the class of neighbourhood models are defined as follows: 
	\begin{itemize}[noitemsep]
		\item \textit{Normality}: For all $ x \in W$ it holds that $ N(x) \neq \emptyset $;
		\item \textit{Total reflexivity}: For all $ x \in W$ there exists $\alpha \in N(x) $ such that $x \in \alpha$;
		\item \textit{Weak centering}: For all $ x \in W$ and $\alpha \in N(x) $, $ x \in \alpha $;
		\item \textit{Centering}: For all $ x \in W $ and $ \alpha \in N(x) $, $x \in \alpha$ and  $\{x\} \in N(x)  $; 
		\item \textit{Uniformity}\footnote{The property of uniformity as we have defined it is sometimes called \emph{local uniformity}, to distinguish it from the following property, called \emph{uniformity}: for all $ x, y \in W $,  $ \bigcup \, N(x) = \bigcup \, N(y)$. However, the set of valid formulas in the class of models satisfying uniformity and local uniformity is the same. A similar remark applies to the property of absoluteness.  }: For all $ x \in W $ it holds that if $ y \in \alpha $ and $ \alpha \in N(x) $, then $ \bigcup \, N(x) = \bigcup \, N(y)$. 
		\item \textit{Absoluteness}:  For all $ x \in W$ it holds that if $ y \in \alpha $ and $ \alpha \in N(x) $, then $  N(x) =  N(y)$.
	\end{itemize}
	The extensions are respectively denoted by $\modelex{N} $, $ \modelex{T} $, $ \modelex{W} $, $ \modelex{C} $, $ \modelex{U} $ and $ \modelex{A} $. As happens with axioms, semantic conditions can be combined, yielding 15 classes of models: so $ \modelex{NT} $ is a neighbourhood model with normality and total reflexivity, $ \modelex{WA} $ is a neighbourhood model with weak centering and absoluteness, and so on.    
\end{definition}

Not all the extensions of the above table are proper conditional logics. We observe that  
\begin{itemize}
	\item[1.] $\mathbb{PCA} $ collapses to classical logic;  
	\item[2.] $\mathbb{PWA} $ collapses to {\bf S5}. 
\end{itemize}
We provide a proof of the above through the semantics, obtaining a collapse of models. This implies the collapse of logical systems, once completeness has been proved.

For 1, we prove that $N(x)=\{\{x\}\}$. Let $y\in\alpha$ and $\alpha\in N(x)$. By absoluteness, $N(x)=N(y)$. By centering, $\{x\}\in N(x)$ and 
$\{y\}\in N(y)$, so that $\{y\}\in N(x)$ and $x\in \{y\}$, whence $x=y$. It follows that there is only one possible world, and the forcing condition of the conditional collapses to the one of material implication.

For 2, we prove that $N(x)=\{S\}$, where $S$ is any set of worlds to which $x$ belongs. Let $\alpha, \beta\in N(x)$. We show that 
$\alpha=\beta$. Let $y\in\alpha$; then, by absoluteness $N(x)=N(y)$, so $\beta\in N(y)$, and by centering $y\in \beta$. We conclude $\alpha\subseteq \beta$. The other inclusion is proved in the same way. Moreover, from the fact that for any $y\in S$, $N(y)=\{S\}$ it follows that all the possible words are equivalent: thus, the forcing condition of a conditional $ A>B $ reduces to the truth condition of the strict implication $ \Box (A \rightarrow B) $.

\

\noindent By adding to $ \axiom{PCL} $ the axiom
	\begin{quotation}
		(CV) ~ $ ((A>C) \wedge \lnot (A>\lnot B)) \rightarrow ((A\wedge B)>C) $
	\end{quotation}
	we obtain logic $ \VV $, which is the basic system of Lewis' counterfactual logic. By adding the axiom to the other preferential logics, we get the family of counterfactual logics, $ \VV $ and extensions, introduced in \cite{lewis1973}. Lewis defined the semantics of counterfactual logics in terms of sphere models; and sphere models for $ \VV $ can be obtained by adding to neighbourhood models the following condition:
	\begin{quotation}
		\textit{Nesting}: For all $ \alpha,\beta \in N(x) $, either $ \alpha \subseteq \beta $ or $ \beta \subseteq \alpha $.	
	\end{quotation} 
	Thus, the family of Lewis' logics is by all means an extension of the preferential systems, and the proof theoretic and model theoretic methods exposed in the following sections can be (more or less modularly) extended to cover Lewis' logics. 

\section{Soundness and completeness of neighbourhood models}

We now  prove soundness and completeness of the classes of models with respect to the axioms of $ \PCL $ and its extensions. 

\subsection{Soundness}

\begin{theorem}[Soundness]
	\label{ax_soundness}
	For $ F \in \fpcl $, $ \axiom{P} $ axiom system for a preferential logic $ \mathbb{P} $ and $ \mathcal{P} $ the corresponding class of neighbourhood models, it holds that if $\derivable{P} F$, then $ \vDash_{\mathcal{P}} F $.
\end{theorem}

\begin{proof}
	The proof consists in showing that the axioms are valid, and that the inference rules preserve validity. By means of example, we prove soundness of axioms (CM), (OR) and (U$_1$).
	
	\noindent (CM) \, $ ((A>B) \wedge (A>C)) \rightarrow ((A\wedge B)>C) $.  Consider an arbitrary neighbourhood model $ \modelNeigh$ and an  arbitrary  world $x$, and suppose that $x$ forces the antecedent of the implication.  We show that $x$ forces the succedent. The assumption means that:
	\begin{itemize}[noitemsep]
		\item [1.] $ \modelNeigh, x \Vdash A>B $, i.e., if there exists $ \alpha \in N(x) $ such that $ \alpha \fe A $, then there exists 
		$ \beta \subseteq \alpha  $ such that $ \beta \fe A $ and $ \beta \fu A\rightarrow B $;
		\item [2.] $\modelNeigh, x \Vdash A>C $, i.e., if there exists $ \alpha \in N(x) $ such that $ \alpha \fe A $, then there exists  
		$ \gamma \subseteq \alpha  $ such that $ \gamma \fe A $ and $ \gamma \fu A\rightarrow C $.
	\end{itemize} 
	Suppose that there is $\alpha\in N(x) $ such that $ \alpha \fe A\wedge B $; in particular, $ \alpha \fe A$ so by 1 we have that there is $ \beta \subseteq \alpha  $ such that $ \beta \fe A $ and $ \beta \fu A\rightarrow B $. By 2 from $ \beta \fe A $, we have that there is 
	$\gamma \subseteq \beta  $ such that $ \gamma \fe A $ and $ \gamma \fu A\rightarrow C $. Since $\gamma\subseteq \beta  $ and $\gamma\fe A$, by $\beta \fu A\rightarrow B $ we get $\gamma\fe A\wedge B $. From $ \gamma \fu A\rightarrow C $, a fortiori we have $ \gamma \fu A\wedge B\rightarrow C $, so we have proved that $x\Vdash A\wedge B> C$. 
	
	\noindent (OR) \, $( (A>C) \wedge (B>C)) \rightarrow ((A\lor B) >C) $. Suppose there is a neighbourhood model which satisfies the antecedent, i.e.
	\begin{itemize}[noitemsep]
		\item [1.] $ \modelNeigh, x \Vdash A>B $ , i.e., if there exists $ \alpha \in N(x) $ such that $ \alpha \fe A $, then there exists $ \beta \subseteq \alpha  $ such that $ \beta \fe A $ and $ \beta \fu A\rightarrow B $;
		\item [2.] $ \modelNeigh, x \Vdash B>C $, , i.e., if there exists $ \alpha'\in N(x) $ such that $ \alpha' \fe B $, then there exists $ \gamma \subseteq \alpha'  $ such that $ \gamma \fe B $ and $ \gamma \fu B\rightarrow C $.
	\end{itemize}
	Our claim is $x\Vdash( A\lor B)> C$. Assume there is $\alpha'' \in N(x) $ such that $ \alpha \fe A\lor B $. Then either  $ \alpha \fe A $ or $ \alpha \fe B $. In the first case we use 1 and obtain that there is $ \beta \subseteq \alpha'' $ such that $ \beta \fe A $ and $ \beta \fu A\rightarrow C $. Then from 2 (with $\beta$ in place of $\alpha'$) we obtain that there is $\gamma \subseteq \beta$ such that 
	$\gamma\Vdash B$ (and a fortiori  $\gamma\Vdash A\lor B$) and $ \gamma \fu B\rightarrow C $. Since $\gamma\subseteq \beta\subseteq \alpha$, by 1 we have $\gamma\fu A\rightarrow C$, and a fortiori $\gamma\fu (A\lor B)\rightarrow C$. The second case is dealt with in a similar way, so we conclude $x\Vdash (A\lor B)> C$.
	
	\noindent (U$ _1 $) \, $ (\lnot A>\bot) \rightarrow (\lnot (\lnot A >\bot) >\bot )$. Suppose there is a neighbourhood model with local uniformity that verifies the antecedent, i.e. 
	\begin{itemize}[noitemsep]
		\item [1.] $ \modelex{U}, x \Vdash \lnot A> \bot $, , i.e., if there exists $ \alpha \in N(x) $ such that $ \alpha \fe \lnot A $, then there exists $ \beta \subseteq \alpha  $ such that $ \beta \fe \lnot A $ and $ \beta \fu \lnot A\rightarrow \bot $.
	\end{itemize}
	We claim that if there exists $ \alpha' \in N(x) $ such that $ \alpha' \fe \lnot (\lnot A>\bot )$, then there is $\gamma\subseteq \alpha'  $ such that $ \gamma\fe \lnot (\lnot A>\bot )$ and  $ \gamma\fu \lnot \lnot (\lnot A>\bot )$. The latter two give a contradiction, so we need to prove that the existence of the above $\alpha'$ leads to a contradiction.
	
	Assume  $ \alpha' \fe \lnot (\lnot A>\bot )$, i.e. there is $y\in  \alpha'$ such that $y \nVdash\lnot A >\bot$. Then there is $\delta\in N((y)$ such that $\delta\fe\lnot A$. Since  $ \bigcup N(x) = \bigcup N(y) $ by the condition of uniformity, there is $\alpha''\in N(x)$ such that $z\in \alpha''$ and $\alpha''\fe \lnot A$. By 1, there is $ \beta \subseteq \alpha''  $ such that $ \beta \fe \lnot A $ and $ \beta \fu \lnot A\rightarrow \bot $, so we have the desired contradiction.	
\end{proof}

\subsection{Completeness of $ \PCL $}

\noindent We here prove the completeness of of $\PCL$ with respect to the neighbourhood semantics (extensions are treated in Subsection 3.3).

Generally speaking, proving completeness for the axiom systems of $ \PCL $ and its extensions seems to be quite an arduous task. Burgess \cite{burgess1981quick} was the first to provide a completeness proof for $ \PCL $, using preferential models. His proof in the mentioned paper, condensed in a few pages, is quite intricate and  not so easy to grasp. In his thesis, Veltman \cite{veltman1985logic} gave a proof of strong completeness of $\PCL$ with respect to preferential semantics. This result is far from elementary.  In  \cite{halpern1994complexity} Halpern and  Friedman sketched a completeness proof for $ \PCL $, claiming the proof to be similar to Burgess' proof. Moreover, they state that the proof can cover extensions of $ \PCL $, but the proof for extensions is postponed to a full paper which never appeared. 

More recently, in \cite{giordano2009tableau}, the completeness of the axiomatization of $ \PCL $ and its extensions is proved with respect to classes of preferential models, assuming the Limit assumption.

For Lewis' sphere models, a direct completeness result was given by Lewis in \cite{lewis1973}: he proved that the axioms of $ \VV $ and extensions are sound and complete with respect to sphere models. However, the proof heavily relies on the connective of comparative plausibility, which is definable in $ \VV $ but not in $\PCL$.

To the best of our knowledge, no completeness result is known for the axioms of $ \PCL $ and its extensions with respect to neighbourhood models. 
The proofs in the rest of this section cover $\PCL$ and all its extensions, except those ones congaing weak centering (and not containing centering). The proofs make use of 
some notions and lemmas from \cite{giordano2009tableau}.


\

We follow the standard strategy: in order to prove the completeness of an axiom system $\axiom{K}$ with respect to a class of models $\modelex{K}$, we define a model $\canonicalex{P}$ and we prove that:
\begin{enumerate}
	\item $\canonicalex{K}$ is \emph{canonical}, meaning that for any formula $F\in \fpcl$, $\derivable{K} F$ if and only if  $ F $ is valid in $ \canonicalex{K} $;
	\item $\canonicalex{K} \in \modelex{K}$.
\end{enumerate}
From these two facts the completeness of $\axiom{K}$ with respect the class $\modelex{K}$ immediately follows. For $\PCL$ the class $\modelex{K}$ will be the class of all neighbourhood models; for extensions, $\modelex{K}$ will be the class of models extended with the properties detailed in Definition \ref{def:extensions_neigh}.

As usual, the model is be built by considering \emph{maximal consistent sets} of formulas. We start by recalling standard definitions and properties. The notion of (in-)consistency and subsequent definitions and lemmas on maximal consistent sets are relative to some axiom system $\axiom{K}$. 

\begin{definition}\label{def:maxc}
	Given a set of formulas $ S \in \fpcl $, we say that $ S $ is \emph{inconsistent} if it has a finite subset $ \{B_1, \dots , B_n \} \subseteq S $ such that $ \derivable{K} (B_1 \wedge \dots \wedge B_n)\rightarrow \bot $. We say that $ S $ is \emph{consistent} if it is not inconsistent. We say that $ S $ is \emph{maximal consistent} if $ S $ is consistent and for any formula $ A \notin S $, $ S \cup \{A\} $ is inconsistent.  We denote by $X,Y,Z,\dots$ the maximal consistent sets and by $\Maxc$ the set of all maximal consistent sets over a fixed language. 
\end{definition}
We assume all standard properties of $\Maxc$ sets, in particular the following:

\begin{lemma}\label{lindem} ~
	\begin{itemize}
		\item [$a)$] For $ S $ set of formulas, $S$ is consistent if and only if there exists $Z\in \Maxc$  such that $S\subseteq Z$.
		\item [$b)$] For $ A $ formula, $\valid{K} A$ if and only if for all $ Z\in \Maxc$, $A\in Z$.
	\end{itemize}
\end{lemma}
\begin{proof}
	The direction \emph{only if} of $a)$ is  the standard \emph{Lindembaum lemma}, proved by means of an inductive construction. Property $b)$ is a  sub-case of $a)$, obtained by taking $S = \{\lnot A\}$, by controposition and completeness of all $Z\in \Maxc$ (either $A\in Z$ or $\neg A\in Z$). 
\end{proof}

We will (shortly) define the worlds of the canonical model $ \canonicalex{K} $ as the set $\{(X, A) \mid X\in \Maxc \ \mbox{and} \ A\in \fpcl \ \mbox{and} \ A\in X \}$. 
Thanks to Lemma \ref{lindem}, in order to prove that the  $ \canonicalex{K} $ is indeed \emph{canonical}, we will only have to show that for any formula $F\in\fpcl$ and for any world $(X, A)$, it holds that:
\begin{quotation} 
	$(\textit{Truth Lemma})$ ~ $F\in X$ if and only if  $(X, A) \Vdash F$.
\end{quotation}
It easy to see that canonicity of $ \canonicalex{K} $ follows: $A$ is valid in $ \canonicalex{K} $ if and only if for all $ Z\in \Maxc$, $A\in Z$ (by the Truth Lemma and definition of the worlds),  if and only if  $\derivable{K} A$ (by Lemma \ref{lindem}). 

Before providing the canonical model construction, we introduce some additional definitions and lemmas.

\begin{definition}
	Let $ X \in \Maxc $. The set of \emph{conditional consequences} of a formula $B \in \fpcl$ is defined as:
	$ \cons{X}{B} = \{C \in \fpcl \mid B>C \in X \} $.
\end{definition}


\begin{lemma}\label{lemma:maximal_consequences}
	The following hold: 
	\begin{enumerate}
		\item $ B \in \cons{X}{B} $;
		\item If $\cons{X}{B} \subseteq Y $ and $B>C \in X $, then $ C \in Y$; 
		\item $B>C \in X$ iff for all $ Y$, $\cons{X}{B}\subseteq Y $ implies $C \in Y$.
	\end{enumerate}
	
\end{lemma}

\begin{proof}
	We prove only  direction $\Leftarrow$ of statement 3. By hypothesis, there is no $Z \in \Maxc $ such that $ \cons{X}{B} \cup \{\lnot C \}\subseteq Z$. By lemma  \ref{lindem},  $\cons{X}{B} \cup \{\lnot C \}$ is inconsistent, and there must be some $ D_1, \dots, D_n \in \cons{X}{B}$ such that $\derivable{K}  (D_1 \wedge \dots \wedge D_n) \rightarrow C$. Thus, by (RCK) and (R-And), $\derivable{K} ((B>D_1) \wedge \dots \wedge (B>D_n)) \rightarrow (B>C)$. Since $ (B>D_1), \dots , (B>D_n) \in X$, also $B>C \in X $.   
\end{proof}

\begin{definition}\label{def:lesser_equal}
	Let $X \in \Maxc$, $A, B \in \fpcl $. Define $A \clesser{X} B $ if  $ (A \lor B) > A \in X$.
\end{definition}

\begin{proposition} \label{prop:ref_tr}
	The relation $ \clesser{X} $ is reflexive and transitive. 
\end{proposition}

\begin{proof}
	Reflexivity follows from axiom (ID) and (OR).
	Transitivity immediately follows from 1 of Lemma \ref{lemma:derivable_ax}. 
\end{proof}

\begin{proposition} [From \cite{giordano2009tableau}] \label{prop:relations}
	If $ A \clesser{X} B $, $ \cons{X}{A} \subseteq Y$ and $B \in Y $, then $ \cons{X}{B} \subseteq Y $.
\end{proposition}

\begin{proof}
	Let  $B>C \in X $ (thus, $ C \in \cons{X}{B}$). Our goal is to show that $ C \in Y $. 
	By hypothesis, we know that $ (A\lor B)>A \in X $. From Axiom 6 of Proposition \ref{lemma:derivable_ax} it follows that $A>(B\rightarrow C) \in X$. Thus, $B\rightarrow C \in \cons{X}{A} $ and, by hypothesis $B\rightarrow C \in Y $ and $B \in Y $. Thus, $ C \in Y$. 
\end{proof}

\begin{proposition} \label{prop:relations_important}
	If $ A \clesser{X} B \clesser{X} C $, $ \cons{X}{A} \subseteq Y$ and $C \in Y $, then $\cons{X}{B} \subseteq Y $.
\end{proposition}

\begin{proof}
	By hypothesis, $(A \lor B) > A \in X $ and $(B\lor C)> B \in X$. By Axiom 5 of Proposition \ref{lemma:derivable_ax}, $A>(C\rightarrow B) \in X $. Thus, $ C\rightarrow B \in X^A$, and $C\rightarrow B \in Y $. Since $C \in Y $, we have $B \in Y $. Thus, we have that $A \clesser{X} B $,  $ \cons{X}{A} \subseteq Y$ and $B \in Y $. Applying Proposition \ref{prop:relations} we obtain  $\cons{X}{B} \subseteq Y $.
\end{proof}

\noindent We can now proceed with the construction of the canonical model. 

\begin{definition} \label{def:canonical_model}
	For $ p $ propositional atom, let 
	\begin{itemize}
		\item $\worldsC = \{ (X,A) \mid X \in \Maxc \textit{ and } A \in \fpcl \textit{ and } A \in X \} $;
		\item $ \valC(p) = \{ (X,A) \in \worldsC \mid p \in X \} $. 
	\end{itemize}
	For $ (X,A), (Y,B) \in \worldsC $, we define a neighbourhood as:
	$$
	\spheres{(X,A)}{(Y,B)} = \{ (Z,C) \in \worldsC \mid \cons{X}{C} \subseteq Z \textit{ and } C \clesser{X} B \textit{ and } B \notin Z \} \cup \{(Y,B) \}
	$$
	Now for any $ (X,A) \in \worldsC$, let the neighbourhood function be defined as :
	$$ 
	\neigh{(X,A)} = \{ \spheres{(X,A)}{(Y,B)} \mid \cons{X}{B} \subseteq Y \ \mbox{and} \ B\in\fpcl\} 
	$$
	Finally, let the \textit{canonical model} be defined as $ \canonicalex{N} = \langle \worldsC, \mathcal{N}, \valC\rangle $.
\end{definition}

\begin{note}
	Slightly abusing the notation, we write $  \neigh{X,A}$ instead of  $\neigh{(X,A)}$. Moreover, since in $\spheres{(X,A)}{(Y,B)}$ the $A$ is not needed, we simplify the notation to $\spheres{X}{(Y,B)}$.
\end{note}

\begin{proposition} \label{prop:verification_canonical}
	The canonical model $ \canonicalex{N} $ is a neighbourhood model.
\end{proposition}

\begin{proof}
	It suffices to verify that that non-emptiness holds; since for all $ (Y, B ) \in \worldsC $ it holds that $ (Y,B) \in \spheres{X}{(Y,B)} $, the property immediately follows. 
\end{proof}

\begin{lemma}\label{lemma:neighbourhood_nesting}
	If $ \spheres{X}{(Y,B)} \in \neigh{X,A} $ and $ (U,D) \in \spheres{X}{(Y,B)}  $, then $ \spheres{X}{(U,D)} \subseteq \spheres{X}{(Y,B)}   $. 
\end{lemma}

\begin{proof}
	We prove the non-trivial case in which $ (U,D) \neq (Y,B) $.
	Let $ (V,E) \in \spheres{X}{(U,D)}  $; we have to show that $ (V,E) \in \spheres{X}{(Y,B)}  $. Thus, we have to show that $a)$ $ \cons{X}{E} \subseteq V $, $b)$ $ E \clesser{X} B $ and $c)$ $ B \notin V $.
	Again, we consider the non-trivial case in which $ (V,E) \neq (U,D) $. 
	Since $ (V,E) \in \spheres{X}{(U,D)}  $ we have that $\cons{X}{E} \subseteq V  $ (requirement $a$ is met),  $ E \clesser{X} D $ and $ D \notin V $. Since $ (U,D) \in \spheres{X}{(Y,B)}  $ we have, among the others, that $ D\clesser{X} B $. By transitivity of $ \clesser{X} $ (Proposition \ref{prop:ref_tr}) it follows that $ E \clesser{X} B $. Thus, $b)$ is satisfied. It remains to prove that $ B \notin V $. For the sake of contradiction, suppose that $ B \in V $.  From this latter, $ E \clesser{X} D \clesser{X} B $ and $\cons{X}{E} \subseteq V  $ it follows by Proposition \ref{prop:relations_important} that $ \cons{X}{D}  \subseteq V$; thus, by Lemma \ref{lemma:maximal_consequences}, $ D \in V $, against previous assumption. Thus, requirement $c)$ is satisfied.
\end{proof}

\noindent We are now ready to prove the Truth Lemma. 

\begin{lemma}[Truth Lemma] \label{lemma:truth_lemma}
	Let $ F \in \fpcl$ and $ X \in \Maxc $. The following statements are equivalent:
	\begin{itemize}
		\item $ F \in X $;
		\item $ \canonicalex{N}, (X,A) \Vdash F $.
	\end{itemize}
\end{lemma}

\begin{proof}
	As usual, the proof proceeds by mutual induction on the complexity of the formula $F$. We show only the case of $ F \equiv G>H $, tacitly assuming that the inductive hypothesis holds on subformulas of $ F $, that is for $G$ (and similarly for $H$): and any world $(U, B)$: $ G \in U $ iff $ \canonicalex{N}, (U,B) \Vdash G $. 
	Thus, we have to prove the equivalence of the following statements:
	\begin{enumerate}
		\item $ G>H \in X $;
		\item For all $ \alpha \in \neigh{X,A} $, if $ \alpha \fe G $ then there exists $ \beta \in \neigh{X,A} $ with $ \beta \subseteq \alpha $, $ \beta \fe G $ and $ \beta \fu G\rightarrow H $.		
	\end{enumerate}
	\noindent $ \mathbf{[1\Rightarrow 2]} $ Assume 1, and suppose that $ \alpha \in \neigh{X,A} $ and $ \alpha \fe G $, for $ \alpha = \spheres{X}{(Y,B)} $. We  must show that there exists a $ \beta \in \neigh{X,A} $ such that $\beta \subseteq \alpha$,  $\beta \fe G $ and $ \beta \fu G\rightarrow H $.
	
	We distinguish two cases, depending on whether $ B\clesser{X} G $ holds or not. Suppose it holds; then, we show that we can take $ \beta = \alpha = \spheres{X}{(Y,B)} $. Given the hypothesis we only have  to prove that $ \alpha \fu G\rightarrow H $. To this aim let $ (U,D) \in \spheres{X}{(Y,B)} $ and  $ G\in U $. From $ (U,D) \in \spheres{X}{(Y,B)} $ it follows that $ \cons{X}{D} \subseteq U $ and $ D \clesser{X} B $. Since $ B\clesser{X} G $, by transitivity of $\leq_X$ we obtain   $ D \clesser{X} G $. Therefore we have:  $ G \in U $,  $ \cons{X}{D} \subseteq U $ and  $ B\clesser{X} G $,so that by Proposition \ref{prop:relations} we obtain $ \cons{X}{G} \subseteq U $. Since  $ G>H \in X $ we have $ H \in \cons{X}{G} $, and finally $ H\in U $. 
	
	Now suppose that $ B\clesser{X} G $ does not hold. Therefore $ \lnot ((B\lor G) > B) \in X$. Thus, $ \cons{X}{B\lor G} \cup \{ \lnot B\}$ is consistent, so that (by lemma \ref{lindem}) there exists some $ Z \in \Maxc $ such that $ \cons{X}{B\lor G} \cup \{ \lnot B\} \subseteq Z $ (whence $G\in Z$). Let us consider the world $(Z, B\lor G) $. Note that by construction $ \cons{X}{B\lor G} \subseteq Z $, and obviously $ (B\lor G) \clesser{X} B $ and $ B \notin Z$,  By Definition \ref{def:canonical_model}  $ (Z,B\lor G) \in \spheres{X}{(Y,B)}$. We show that we can take the required $ \beta = \spheres{X}{(Z,B\lor G)} $:  since  $\cons{X}{B\lor G} \subseteq Z $, we have $\spheres{X}{(Z,B\lor G)} \in \neigh{X,A} $; since $ (Z,B\lor G) \in \spheres{X}{Y,B}$, by lemma \ref{lemma:neighbourhood_nesting} we have $\spheres{X}{(Z,B\lor G)} \subseteq \spheres{X}{(Y,B)}$; since $G\in Z$, we immediately have $\spheres{X}{(Z,B\lor G)} \fe G$. We still have to prove that $ \spheres{X}{(Z,B\lor G)} \fu G\rightarrow H $. 
	To this purpose suppose $ (U,D) \in \spheres{X}{(Z,B\lor G)} $ and $G\in U$: since  we have  $D \leq_X B\lor G \leq_X G$, as before, by Proposition \ref{prop:relations}, we obtain $ \cons{X}{G} \subseteq Z$ and we can  conclude   $ H\in Z $.

	\
	
	\noindent $ \mathbf{[2\Rightarrow 1]} $ Assume 2. We show that for all $ Z \in \Maxc $, if $ \cons{X}{G} \subseteq Z $, then $ H \in Z $. 
	By Lemma \ref{lemma:maximal_consequences}, this is equivalent to $ G>H \in X $.
	
	To this aim, suppose that $ \cons{X}{G} \subseteq Z $, for some $ Z $. Then, $ (Z,G) \in \worldsC $. Let us consider the neighbourhood $  \spheres{X}{(Z,G)} = \alpha$: by construction this world belongs to $ \neigh{X,A} $ and thus, by hypothesis, $ \spheres{X}{(Z,G)} \fe G $. By hypothesis 2., there exists some neighbourhood $ \beta \in \neigh{X,A}$ such that $ \beta \subseteq \alpha $, $ \beta \fe G $ and $ \beta \fu G\rightarrow H $. It easy to see that it must be  $ \beta = \alpha =  \spheres{X}{(Z,G)} $, since by Definition \ref{def:canonical_model} the only world that satisfies $ G $ in the neighbourhood $  \spheres{X}{(Z,G)} $ is $ (Z,G) $ itself ($\forall  (U,D) \in \spheres{X}{(Z,G)}$ if $(U,D) \not= (Z,G)$ then $G \notin U$). 
	Thus, from $ \spheres{X}{(Z,G)} \fu G\rightarrow H $,  $ (Z,G) \in  \spheres{X}{(Z,G)} $ and $ G \in Z $ it immediately follows that $ H \in Z $.
\end{proof}

\noindent By the previous lemma we immediately obtain:

\begin{theorem}[Completeness]
	\label{theorem:completeness}
	For $ F \in \fpcl$, if $ \valid{N} $ then $ \derivable{PCL} F $.
\end{theorem}


\subsection{Completeness for  extensions of $\PCL$}

\noindent Our aim is to extend the completeness proof to the whole family of all preferential logics. We are able to extend the proof to  all extensions of $ \PCL $, except for the systems containing \textit{weak centering} (and not containing centering).
To obtain a proof for a logic featuring more than one semantic condition, it suffices to combine the proof strategies for each case. 

Unless otherwise specified, all notions refer to the canonical model for $ \PCL $ defined in the previous section. In some cases, the canonical model needs to be modified to account for specific conditions.
The following proposition (whose proof is obvious) will used for the cases of absoluteness and uniformity.
\begin{proposition}
	For every $(X,A) \in \worldsC$, it holds:
	$$\bigcup \neigh{X,A} = \{(Z,C)\in \worldsC \mid Z^C \subseteq X\}.$$
\end{proposition}

\subsubsection*{Normality}
We show that in presence of Axiom (N), the canonical model $ \canonicalex{N} $ satisfies the condition of normality:
\begin{quotation}
	For all $(X, A)\in \worldsC$, it holds that $\neigh{X,A}\neq \emptyset$.
\end{quotation}
By Axiom (N), we have that for all $(X, A)\in \worldsC$, it holds that $\lnot (\top > \bot)\in X$. Thus, $X^{\top}$ is consistent and by Lemma \ref{lindem} there is $Z\in \Maxc$ such that $X^{\top} \subseteq Z$. As a consequence, $(Z,\top)\in \worldsC$, and $\spheres{X}{(Z,\top)} \in \neigh{X,A}$, whence  $\neigh{X,A}\neq \emptyset$.

\subsubsection*{Absoluteness}
We show that in presence of Axioms $(A_1),(A_2)$, the canonical model $ \canonicalex{N} $ satisfies the condition of local absoluteness: 
\begin{quotation}
	If $(Z,C) \in \bigcup \neigh{X,A}$, then $ \neigh{X,A} = \neigh{Z,C}$.
\end{quotation}
We first prove that $a)$ for any formula $B\in \fpcl$, $ \cons{X}{B} = \cons{Z}{B} $. To this aim, let $  G \in \cons{X}{B} $; then $ B>G \in X $. By Axiom (A$ _1 $), $ C>(B>G) \in X$, and $ B>G \in \cons{X}{C}$. Since $(Z,C) \in \bigcup \neigh{X,A}$, it holds that $ \cons{X}{C} \subseteq Z$; from this follows that $ B>G \in Z $, and thus $ G \in \cons{Z}{B} $. Conversely, suppose $ G \notin \cons{X}{B} $. Then $ \lnot (B>G) \in X $; by (A$ _2 $) $ C > \lnot (B>G) \in X $, and $ \lnot (B>G) \in \cons{X}{C} $. Again, since  $ \cons{X}{C} \subseteq Z$ we have $ \lnot(B>G) \in Z $, and thus $ G \notin \cons{Z}{B} $.

Observe that $b)$ for any formulas $D,E\in \fpcl$, it holds $D \clesser{X} E$ if and only if  $D \clesser{Z} E$. In fact, from $D \clesser{Z} E$ follows that  $(D \lor E)> D \in Z$, and by proceeding similarly as in $a)$ we obtain that $(D \lor E)> D\in X$ if and only if $(D \lor E)> D \in Z$.

From $a)$ it immediately follows that for any $(Y,B)$,  $\cons{X}{B} \subseteq Y$ if and only if $\cons{Z}{B} \subseteq Y$. Then, by $b)$ we have $ \spheres{X}{(Y,B)} = \spheres{Z}{(Y,B)}$, whence by $a)$ we obtain $ \neigh{X,A} = \neigh{Z,C}$.


\subsubsection*{Total Reflexivity}
In this case we need to modify the construction of the canonical model.

\begin{definition}\label{def:universe}
	The \emph{universe} of $ (X,A) $ is the set:
	$$ \univ{X,A} = \{(Y,B) \in \worldsC \mid \textit{ for all } G \in \fpcl, ~ G>\bot \in X\textit{ implies } \lnot G \in Y \} .$$
	The canonical model $ \canonicalex{U} = \langle \worldsCU, \neighFunctU, \valCU \rangle $ is defined by stipulating $ \worldsCU = \worldsC $, $ \valCU = \valC $, and
	$$
	\neighU{X,A} = \neigh{X,A} \cup \{ \univ{X,A} \}.
	$$
	where $ \neigh{X,A}$ is the same as in Definition \ref{def:canonical_model}. 
\end{definition}


\begin{lemma}\label{prop:univ}
	For any $ (X,A), (Y, B) \in \worldsC $, it holds that $\spheres{X}{(Y,B)} \subseteq \univ{X,A} $.
\end{lemma}
\begin{proof}
	Assume that some $(Z,C) \in  \spheres{X}{(Y,B)} $. We have to prove that for all $ G \in \fpcl $, if $ G>\bot \in X $ then $ \lnot G \in Z $, and this immediately follows from \axMod \ and $ \cons{X}{C} \subseteq Z $. 
\end{proof}

We show that in presence of axiom $(T)$, the canonical model $ \canonicalex{U} $ satisfies the condition of total reflexivity, that is:
\begin{quotation}
	If $(X,A)\in \worldsCU$, there exists $\alpha\in \neighU{X,A}$ such that $(X,A)\in\alpha$.
\end{quotation}
It is immediate to verify that the condition holds: because of Axiom (T), we have that $(X,A)\in \univ{X,A}$.

Since we have modified the definition of the canonical model, we have to verify that the Truth Lemma still holds. 
To this aim,  we  need to add one case in the direction  [\textbf{1} $\Rightarrow$ \textbf{2}]  of the proof, that is, if $ G>H \in X $, then $ \canonicalex{U}, (X,A) \Vdash G>H $. Assume that $ G>H \in X $ and that for some  $ \alpha \in \neighU{X,A} $ it holds $ \alpha \fe A $. If $ \alpha \in \neigh{X,A} $ the proof proceeds as in Lemma \ref{lemma:truth_lemma}. Let now $ \alpha = \univ{X,A} $ and suppose for some $ (Z,C) \in \univ{X,A} $ it holds that $ (Z,C) \Vdash G $, whence $ G \in Z $. We show that there must exist an $ (U,D) \in \bigcup \neigh{X,A} $ such that $ (U,D) \Vdash G $. If this were not the case, we would get that for all $ (U,D) \in \bigcup \neigh{X,A} $, $ G\notin U $. But this entails that $ \cons{X}{G} $ is inconsistent; and thus $ G> \bot \in X $, against the hypothesis that $ (Z,C) \Vdash G $ and $ (Z,C) \in \univ{X,A} $.  Thus there is a $ (U,D) \in \bigcup \neigh{X,A} $ such that $ G \in U $. We take $ \alpha' = \spheres{X}{(U,D)} $. Observe that $\alpha' \subseteq  \alpha = \univ{X,A}$. We can proceed as in proof of Lemma \ref{lemma:truth_lemma} by finding for a $\beta \in \neigh{X,A}$ with $\beta \subseteq \alpha'$ fulfilling the required conditions.

\subsubsection*{Uniformity}
We take the same model construction as for total reflexivity, that is the model $\canonicalex{U}$. 
Thus, we only need that in presence of axioms $(U_1)$ and $ (U_2) $ $ \canonicalex{U} $ satisfies the condition of local uniformity, that is, for any $ (X,A), (Y,B) \in \worldsCU $: 
\begin{quotation}
	If $ (Y,B) \in \bigcup \neighU{X,A} $, then $ \bigcup \neighU{X,A} = \bigcup \neighU{Y,B}  $.
\end{quotation}
To this aim, first observe that
$$ \bigcup \neighU{X,A} = \univ{X,A} $$
Suppose now $ (Y,B) \in \bigcup \neighU{X,A}= \univ{X,A} $. We show that $ G>\bot \in X $ if and only if $ G>\bot \in Y $.
Let  $  G>\bot \in X $. Then by axiom (U$ _1 $) it follows that $ \lnot (G>\bot )>\bot \in X $. Since $ (Y,B) \in\univ{X,A}$   we  have $ \lnot \lnot (G>\bot ) \in Y $, that is  $ G>\bot \in Y $. Conversely, suppose that $  G>\bot \notin X $, i.e., $ \lnot (G>\bot ) \in X $. By axiom (U$ _2 $) we have that $ (G>\bot) >\bot \in X $, and since $ (Y,B) \in\univ{X,A}$, we get $ \lnot (G>\bot) \in Y $, whence $ G>\bot \notin Y $. \\
From the fact that $ G>\bot \in X $ if and only if $ G>\bot \in Y $ we obtain that for all $ (Z,C) \in \worldsCU $, $ (Z,C ) \in \univ{X,A} $ if and only if $ (Z,C) \in \univ{Y,B} $, which means 
$ \bigcup \neighU{X,A} = \bigcup \neighU{Y,B}  $.

\subsubsection*{Centering}
We modify the canonical model construction as follows. 
\begin{definition}
	For $ (X,A), (Y,B) \in \worldsCW $, let:
	$$
	\spheresCW{(X,A)}{(Y,B)} = \spheres{X}{(Y,B)} \cup \{ (X,A)\}.
	$$
	Observe that here the formula $A$ in $(X,A)$ is relevant. 
	Then, for any $ (X,A) \in \worldsC$, $ \neighW{X,A} = \{ \spheresCW{(X,A)}{(Y,B)} \mid \cons{X}{B} \subseteq Y \} $. The set of worlds $ \worldsC $ and the evaluation function $ \valC $ do not change, and the canonical model is $ \canonicalex{C} = \langle \worldsC, \neighFunctW, \valC \rangle$.
\end{definition}

We now show that in presence of axioms $(W)$ and $ (C) $, the canonical model $ \canonicalex{U} $ satisfies the condition of centering: 
\begin{itemize}
	\item[$a$)] For every world $ (X,A)$ and every $\alpha \in  \neighW{X,A}$, $ (X,A)\in   \alpha$;
	\item [$b) $] $\{(X,A)\}\in  \neighW{X,A}$.
\end{itemize}
Condition $a)$ holds by definition. As for $b)$, first observe that for any $(X,A)$ it holds by (W) that $\cons{X}{A} \subseteq X$, so that $\spheresCW{(X,A)}{(X,A)}\in  \neighW{X,A}$. We now show that $\spheresCW{(X,A)}{(X,A)} = \{(X,A)\} $. To this aim, we prove that there is no world $ (Y,B)  \in   \spheresCW{(X,A)}{(X,A)}$ such that $ (Y,B) \neq (X,A) $. For the sake of contradiction, suppose such a world exists. It follows that $A\not\in Y$ and  $ B \clesser{X} A $,  which means that $ (A\lor B) > A \in X $. Thus, by axiom (W), $ (A \lor B) \rightarrow B \in X $. Since by definition $ A \in X $, we have $ B \in X $. By axiom (C) it follows that also $B>A \in X  $. Thus, $ A \in \cons{X}{B} $; and since $ \cons{X}{B} \subseteq Y $ we have $ A \in Y $, which contradicts with the assumption $  A \notin Y $. 

\

Since we have modified the canonical model, we have to verify that the Truth Lemma continues to hold.
For the direction $\mathbf{[1\Rightarrow 2]}$, suppose that $ G>H \in X $ and that for $ \alpha \in \neighW{X,A} $ it holds that $ \alpha \fe G $. We can proceed as in the proof of Lemma \ref{lemma:truth_lemma}, finding a suitable $\beta\in \neighW{X,A} $. The fact that $(X,A)$ belongs to every neighbourhood in $\neighW{X,A}$, and also to $\beta$, does not compromise the assertion that $ \beta \fu G\rightarrow H $, since from the hypothesis $G>H\in X$ follows by (W) that $G\to H\in X$. 

For the direction $\mathbf{[2\Rightarrow 1]}$, assume 2. We distinguish two cases: 
\begin{itemize}[noitemsep]
	\item[$ i. $] $G\not\in X$;
	\item[$ ii. $] $G\in X$.
\end{itemize}
In case $i$, we proceed as in the proof of Lemma \ref{lemma:truth_lemma}, by proving that for all $ Z \in \Maxc $, if $ \cons{X}{G} \subseteq Z $, then $ H \in Z $. To this aim, let us consider $\alpha = \spheresCW{(X,A)}{(Z,G)} = \spheres{X}{(Z,G)} \cup \{ (X,A)\}\in \neighW{X,A}$. By hypothesis, there exists a neighbourhood $ \beta \in \neighW{X,A}$ such that $ \beta \subseteq \alpha $, $ \beta \fe G $ and $ \beta \fu G\rightarrow H $. Since $G\not\in X$, it must be that $\beta = \spheresCW{(X,A)}{(Z,G)}$, whence $(Z,G)\in \beta$ follows.

In case $ii$, let us consider  $\alpha = \spheresCW{(X,A)}{(X,A)} \in \neighW{X,A}$. By hypothesis, there exists a neighbourhood $ \beta \in \neighW{X,A}$ such that $ \beta \subseteq \spheresCW{(X,A)}{(X,A)} $, $ \beta \fe G $ and $ \beta \fu G\rightarrow H $. However, since $\spheresCW{(X,A)}{(X,A)} = \{(X,A)\}$, it must be $\beta = \{(X,A)\}$. Thus, since $G \to H\in X$ and $G\in X$, we obtain $H\in X$. By axiom (C), we finally obtain $G> H\in X$.


\begin{theorem}[Completeness for extensions]
	\label{theorem:completeness_ext}
	Let $ \mathsf{ext} $ denote one of the logics: $ \PP\NN$, $\PP\TT$, $\PP\CC$, $\PP\UU$ , $\PP\NN\UU$, $\PP\TT\UU$, $\PP\CC\UU$, $\PP\NN\AAA$, $\PP\TT\AAA$, $\PP\CC\AAA $.
	For $ F \in \fpcl$, if $ F $ is valid in a class of models for $ \mathsf{ext} $, then $ \derivable{ext} F $.
\end{theorem}



\section{A family of labelled sequent calculi}

In this section we introduce labelled calculi for $ \PCL $ and its extensions. 
We call $ \lab{CL} $ the calculus for $ \PCL $. Calculi for extensions are denoted by $ \lab{CL} $ to which we add the name of the frame conditions of the corresponding logics: thus,  $ \lab{CL^N} $ is a proof system for $ \PCL\NN $, $ \lab{CL^{TU}} $ is a proof system for $ \PCL\TT\UU $. Let $ \lab{CL^*} $ denote the whole family of calculi.

The definition of the sequent calculi  $ \lab{CL^*} $ follows the well-established methodology of enriching the language of the calculus by means of labels, thus importing the semantic information of neighbourhood models into the syntactic proof system\footnote{Refer to \cite{negri2005proof} for the general methodology in Kripke models and to \cite{negri2016non-normal} for the general methodology in neighbourhood semantics.}. For this reason, it is useful to recall the the truth condition for the conditional operator in neighbourhood models:

\begin{center}
	$(*) \  x \Vdash A > B $ \, \textit{iff} \, \textit{ for all $ \alpha \in N(x)$, if 
		$  \alpha \fe A $ then there exists $ \beta\in N(x)$ such that $\beta \subseteq \alpha $,  $ \beta \fe A $, and $ \beta \fu A\rightarrow B$.}
\end{center}

\noindent  We enrich the language $ \fpcl$ as follows.  

\begin{definition} 
	\label{labelled_formulas}
	Let $ x, y,z, \dots $ be variables for worlds in a neighbourhood model, and $ a, b, c, \dots $ variables for neighbourhoods. \textit{Relational atoms} are the following expressions:
	\begin{itemize}[noitemsep]
		\item $ a \in N(x) $,  ``neighbourhood $ a $ belongs to the family of neighbourhoods associated to $ x $'';
		\item $ x \in a $, ``world $ x $ belongs to neighbourhood $ a$'';
		\item $ a \subseteq b $, ``neighbourhood $ a $ is included into neighbourhood $ b $''. 
	\end{itemize}
	
	\noindent \textit{Labelled formulas} are defined as follows. Relational atoms are labelled formulas and, for $ A \in \fpcl$, the following are labelled formulas:
	\begin{itemize}[noitemsep]
		\item $ x:A $, ``formula $ A $ is true at world $ x $'';
		\item  $ a \fe A $, ``$ A $ is true at some world of neighbourhood $ a $'';
		\item $ a \fu A $,  ``$ A $ is true at all worlds of neighbourhood $ a $'';
		\item $ x \barra{a} A|B $, ``there exists $ \beta\in N(x)$ such that $\beta \subseteq \alpha $, $ \beta \fe A $, and $ \beta \fu A\rightarrow B$''.
	\end{itemize}
	We use $ \{x\} $ to denote a neighbourhood consisting of exactly one element. 
\end{definition}

\noindent Relational atoms and labelled formulas are defined in correspondence with semantic notions. Relational atoms describe the structure of the neighbourhood model, whereas labelled formulas are defined in correspondence with the forcing relations at a world ($ x \Vdash A$) and at a neighbourhood ($a \fe A$, $ a \fu A $). 
Formula  $  x \barra{a} A|B $ introduces a semantic condition corresponding to the consequent of the right-hand side of $ (*) $. The reason for the introduction of this formula is that $ (*) $ is too rich to be expressed by a single rule. 
Thus we need to break $ (*) $ into two smaller conditions, one (the antecedent) covered by rules for formulas $ x: A>B $ and the other (the consequent) covered by $ x \barra{a} A|B $.

\begin{definition}
	Sequents of $ \lab{CL^*} $ are expressions $ \Gamma \Rightarrow \Delta $
	where $ \Gamma $ and $ \Delta $ are multisets of relational atoms and labelled formulas, and relational atoms may occur only in $ \Gamma $.
\end{definition}

\begin{figure}[h!]
	\begin{adjustbox}{max width= \textwidth}
		\begin{tabular}{|c c|}
			\hline 
			& \\
			\multicolumn{2}{|l|}{\textbf{Initial sequents}} \\ [1ex]
			$\infer[\init]{ x:p, \Gamma \Rightarrow \Delta, x:p}{} $ & $\infer[\Lbot]{ x:\bot, \Gamma \Rightarrow \Delta}{}  $ \\ [1ex]
			\multicolumn{2}{|l|}{\textbf{Rules for local forcing}} \\ [1ex]
			$ \infer[\Lfu]{ x \in a, a \Vdash^{\forall} A, \Gamma \Rightarrow \Delta }{ x:A,  x \in a, a \Vdash^{\forall} A, \Gamma \Rightarrow \Delta} $ & 
			$ \infer[\Rfu \  \mathsf{(x!)}]{\Gamma \Rightarrow \Delta, a \fu A }{ x \in a, \Gamma \Rightarrow \Delta, x:A } $  \\[1ex]
			$ \infer[\Lfe  \ \mathsf{(x !)}]{ a \Vdash^{\exists}A, \Gamma \Rightarrow \Delta }{ x \in a, x:A, \Gamma \Rightarrow \Delta} $ & 
			$ \infer[\Rfe]{  x \in a, \Gamma \Rightarrow \Delta, a \Vdash^{\exists}A }{  x \in a, \Gamma \Rightarrow \Delta, x:A, a \Vdash^{\exists}A } $  \\[1ex]
			\multicolumn{2}{|l|}{\textbf{Propositional rules}}  \\[1ex]
			$ \infer[\Land]{x:A\wedge B , \Gamma \Rightarrow \Delta}{x:A, x:B, \Gamma \Rightarrow \Delta} $ & $ \infer[\Rand]{\Gamma \Rightarrow \Delta,x:A\wedge B}{ \Gamma \Rightarrow \Delta,x:A & \Gamma \Rightarrow \Delta, x:B} $\\[1ex]
			$ \infer[\Lor]{x: A\lor B, \Gamma \Rightarrow \Delta}{x:A, \Gamma \Rightarrow \Delta & x:B, \Gamma \Rightarrow \Delta} $ &
			$ \infer[\Ror]{\Gamma \Rightarrow \Delta,x:A\lor B}{\Gamma \Rightarrow \Delta,x:A, x:B} $\\[1ex]
			$ \infer[\mathsf{L}\rightarrow]{ x: A\rightarrow B, \Gamma \Rightarrow \Delta }{\Gamma \Rightarrow \Delta, x:A & x:B, \Gamma \Rightarrow \Delta } $ & $ \infer[\mathsf{R}\rightarrow]{ \Gamma \Rightarrow \Delta, x:A\rightarrow B }{ x:A, \Gamma \Rightarrow \Delta, x:B} $\\[1ex]
			\multicolumn{2}{|l|}{\textbf{Rules for the conditional}} \\ [1ex]
			\multicolumn{2}{|c|}{ $ \infer[\Rcon \ \mathsf{(a!)}]{ \Gamma \Rightarrow \Delta , x: A>B }{ a \in N(x), a \Vdash ^{\exists}A, \Gamma \Rightarrow \Delta, x \Vdash_{a} A|B}   $  } \\[1ex]
			\multicolumn{2}{|c|}{ $ \infer[\Lcon]{a \in N(x), x: A>B, \Gamma \Rightarrow \Delta }{ a \in N(x), x: A>B, \Gamma \Rightarrow \Delta, a \Vdash^{\exists} A & x \Vdash_{a} A|B, a \in N(x), x: A>B, \Gamma \Rightarrow \Delta } $  } \\[1ex]
			\multicolumn{2}{|c|}{$ \infer[\Rbar]{ c  \in N(x), c \subseteq a, \Gamma \Rightarrow \Delta, x \Vdash_{a} A|B }{ c  \in N(x), c \subseteq a, \Gamma \Rightarrow \Delta, x \Vdash_{a} A|B, c \Vdash^{\exists}A & c  \in N(x), c \subseteq a, \Gamma \Rightarrow \Delta, x \Vdash_{a} A|B, c \Vdash^{\forall}A\rightarrow B  } $}\\[1ex]
			\multicolumn{2}{|c|}{$ \infer[\Lbar \ \mathsf{(c!)}]{ x\Vdash_a A|B, \Gamma \Rightarrow \Delta }{ c  \in N(x), c \subseteq a, c \Vdash^{\exists}A, c \Vdash^{\forall} A\rightarrow B, \Gamma \Rightarrow \Delta } $}\\ [1ex]
			\multicolumn{2}{|l|}{\textbf{Rules for inclusion}} \\ [1ex]
			\multicolumn{2}{|c|}{
				$\infer[\Ref]{\Gamma \Rightarrow \Delta}{ a \subseteq a, \Gamma \Rightarrow \Delta } $ \quad $ \infer[\Trans]{ c \subseteq  b, b \subseteq a, \Gamma \Rightarrow \Delta }{ c \subseteq a, c \subseteq b, b \subseteq a, \Gamma \Rightarrow \Delta } $ \quad
				$ \infer[\Lc]{ x \in a, a \subseteq b, \Gamma \Rightarrow \Delta }{ x\in a, a \subseteq b, x \in b, \Gamma \Rightarrow \Delta } $}\\[1ex]
			\hline
		\end{tabular}
	\end{adjustbox}
	\caption{Sequent calculus $ \lab{CL}  $} 
	\label{rules_CL}
	
\end{figure}

\begin{figure}[h!]
	\begin{adjustbox}{max width= \textwidth}
		\begin{tabular}{|c c|}
			\hline 
			& \\
			\multicolumn{2}{|l|}{\textbf{Rules for extensions}} \\[1ex]
			$ \infer[\N \ \mathsf{(a !)}]{\Gamma \Rightarrow \Delta}{a \in N(x),\Gamma \Rightarrow \Delta}$ &  $ \infer[\emp \ \mathsf{(y!) (\star)}]{a \in N(x), \Gamma \Rightarrow \Delta}{y \in a, a \in N(x), \Gamma \Rightarrow \Delta} $\\[1ex]
			$\infer[\T \ \mathsf{(a!)}]{\Gamma \Rightarrow \Delta}{ x \in a, a \in N(x),\Gamma \Rightarrow \Delta}$ &
			$ \infer[\W]{a \in N(x),\Gamma  \Rightarrow \Delta}{x \in a, a \in N(x),\Gamma  \Rightarrow \Delta} $ \\[1ex] 
			$ \infer[\Single]{ \{x\} \in N(x),\Gamma \Rightarrow \Delta}{ x \in \{x\}, \{x\} \in N(x) ,\Gamma \Rightarrow \Delta } $
			&
			$ \infer[\C]{ a \in N(x),\Gamma \Rightarrow \Delta}{ \{ x\} \in N(x), \{x\} \subseteq a, a \in N(x),\Gamma \Rightarrow \Delta }
			$\\[1ex] 
			$ \infer[\Repl{1} \ (*)]{ y \in \{x\}, At(x), \Gamma \Rightarrow \Delta }{ y \in \{x\}, At(x), At(y), \Gamma \Rightarrow \Delta  } $& 
			$ \infer[\Repl{2}\ (*)]{ y \in \{x\}, At(y), \Gamma \Rightarrow \Delta }{ y \in \{x\}, At(x), At(y), \Gamma \Rightarrow \Delta  } $\\[1ex]
			\multicolumn{2}{|c|}{$ \infer[\Unif{1} \ \mathsf{(c!)}]{a \in N(x), y \in a, b \in N(y), z \in b, \Gamma \Rightarrow \Delta}{z \in c,c \in N(x),a \in N(x), y \in a, b \in N(y), z \in b ,\Gamma \Rightarrow \Delta  } $}\\[1ex] 
			\multicolumn{2}{|c|}{$ \infer[\Unif{2}\ \mathsf{(c!)}]{a \in N(x), y \in a, b \in N(x), z \in b, \Gamma \Rightarrow \Delta}{z \in c,c \in N(y),a \in N(x), y \in a, b \in N(y), z \in b ,\Gamma \Rightarrow \Delta} $}\\[2ex]
			$ \infer[\Abs{1}]{a \in N(x), y \in a, b \in N(x), \Gamma \Rightarrow \Delta}{b \in N(y), a \in N(x), y \in a, b \in N(x), \Gamma \Rightarrow \Delta} $
			&
			$ \infer[\Abs{2}]{ a \in N(x), y \in a, b \in N(y), \Gamma \Rightarrow \Delta}{ b \in N(x), a \in N(x), y \in a, b \in N(y), \Gamma \Rightarrow \Delta }
			$ 
			\\[1ex]
			\multicolumn{2}{|l|}{\textbf{Rules obtained by closure conditions}}\\[1ex] 
			$ \infer[\Unif{1}^* \ \mathsf{(c!)}]{a \in N(x), x \in a, z \in a, \Gamma \Rightarrow \Delta}{z \in c, c \in N(x), a \in N(x), x \in a, z \in a, \Gamma \Rightarrow \Delta} $
			&
			$ \infer[\Abs{1}^*]{a \in N(x), y \in a, \Gamma \Rightarrow \Delta}{ a \in N(y), a \in N(x), y \in a, \Gamma \Rightarrow \Delta} $\\[1ex]
			\multicolumn{2}{|c|}{$ \infer[\Unif{1}^{**} \ \mathsf{(c!)}]{a \in N(x), y \in a, a \in N(y), \Gamma \Rightarrow \Delta}{y \in c, c \in N(x), a \in N(x), y \in a, a \in N(y), \Gamma \Rightarrow \Delta} $}\\[1ex]
			\multicolumn{2}{|c|}{$ \infer[\Unif{2}^{**}]{a \in N(x), y \in a, \Gamma \Rightarrow \Delta}{y \in c, c \in N(y), a \in N(x), y \in a, a \in N(y), \Gamma \Rightarrow \Delta} $}\\[1ex]
			\multicolumn{2}{|l|}{$(*) \  At(x) := x:P, x \in a, a \in N(x), x \in \{z\}$, for $ P $ atomic formula.}\\ [1ex]
			\hline
		\end{tabular}
	\end{adjustbox}
	
	\bigskip
	
	\begin{adjustbox}{max width= \textwidth}
		\begin{tabular}{c}
			$ \lab{CL^N} = \lab{CL} + \N + \emp ; \, \lab{CL^T} = \lab{CL^N} + \T ; \,  \lab{CL^W} = \lab{CL^T} + \W ; $\\[0.5ex]
			$  \lab{CL^C} = \lab{CL^W} + \C + \Single + \Repl{1} + \Repl{2} ; $\\[0.5ex]
			$ \lab{CL^U} = \lab{CL} + \Unif{1} + \Unif{2} ; 
			\,   \lab{CL^{NU/TU/WU/CU}} =  \lab{CL^{N/T/W/C}}; 
			$\\[0.5ex]
			$ \lab{CL^A} = \lab{CL} + \Abs{1} + \Abs{2} 
			; \,  \lab{CL^{NA/TA/WA/CA}} =  \lab{CL^{N/T/W/C}} + \Abs{1} + \Abs{2}. 
			$\\[0.5ex]
		\end{tabular}
	\end{adjustbox}
	\caption{Sequent calculi for extensions of $ \lab{CL}  $} 
	\label{rules_extensions}
\end{figure}

\noindent Figure \ref{rules_CL} contains the rules for  $ \PCL $, whereas Figure \ref{rules_extensions} shows the rules for extensions of $ \PCL $. We write $  \mathsf{(a!)} $ as a side condition expressing the requirement that label $ a $ should not occur in the conclusion of a rule. Propositional rules are standard. Rules for local forcing make explicit the meaning of the forcing relations $ \fu  $ and $ \fe $. Rules for the conditional are defined on the basis of the truth condition for $ > $ in neighbourhood models. 

Each rule of Figure \ref{rules_extensions} is defined in correspondence with the frame conditions on extensions of $ \PCL $. For total reflexivity and weak centering, the frame condition can be formalized by means of a single rule. 
Rule  $ \emp $ stands for the requirement of non-emptiness in the model, and it is added to capture the condition of normality, along with rule $ \N $\footnote{The rule needs not to be added to the calculus $ \lab{CL} $: the rules of this calculus always introduce non-empty neighbourhoods, and  the system can be shown to be complete with respect to the axioms of $ \PCL $ (Theorem \ref{theorme:completeness_synth}). However, the rule is needed to express the condition of normality: the new neighbourhood introduced by rule $ \N $ could be empty. }. 

Centering  requires four rules:  
Rule  $ \C $ ensure the Centering condition by introducing formulas with neighbourhood label $ \{x\} $ (the singleton). Rule $ \Single $ ensures that the singleton contains at least one element, and rules $ \Repl{1} $ and $ \Repl{2} $ that it contains at most one element: if there is another element $ y \in \{x\}$, then the properties holding for $ x $ hold also for $ y $ (i.e. $ x $ and $ y $ are the same element).  

Similarly, extensions with  uniformity and absoluteness are defined by adding multiple rules.
Rules $ \Unif{1} $ and $ \Unif{2} $ encode the semantic condition of uniformity. In order to avoid the symbol $ \bigcup $ in the sequent language, the rules translate the following two conditions which, taken together, are equivalent to uniformity.
\begin{quote}
	$ \Unif{1} $:  If there exist $ \alpha \in N(x) $ such that $ y \in \alpha $  and $ \beta \in N(y) $ such that $ z \in \beta $, then there exists $ \gamma \in N(x) $ such that $ z \in \gamma $;
	
	$ \Unif{2} $: If there exist $ \alpha \in N(x) $ such that $ y \in \alpha $  and  $ \beta \in N
	(x)
	$ such that $ z \in \beta $, then there exists $ \gamma \in N(y) $ such that $ z \in \gamma $. 
\end{quote}
As for absoluteness, rules $ \Abs{1} $ and $ \Abs{2} $ encode the information that for any $ x \in W $, given $ a \in N(x) $ and $ y \in a $, if $ \beta \in N(x) $ then $ \beta \in N(y) $ (rule $ \Abs{1} $), and  if $ \beta \in N(y) $, then $ \beta \in N(x) $ (rule $ \Abs{2} $). Thus, $ N(x)= N(y) $.

\

	The sequent calculi $ \lab{CL^*} $ can be modularly extended to cover Lewis' logics (refer to the end of Section 2). To obtain a calculus for $ \VV $, it suffices to add to $ \lab{CL} $ a structural rule corresponding to the semantic condition of nesting:
	$$\infer[\Nes]{a \in N(x), b \in N(x), \Gamma \Rightarrow \Delta}{ a \subseteq b, a \in N(x), b \in N(x), \Gamma \Rightarrow \Delta & b \subseteq a, a \in N(x), b \in N(x), \Gamma \Rightarrow \Delta}
	$$
	The rule can be added to calculi for extensions of $ \PCL $ to obtain calculi for the corresponding logics extending $ \VV $\footnote{Refer to \cite{girlando2018counterfactuals} for a simpler labelled calculus for $ \VV $, which makes use of the connective of \textit{comparative plausibility} instead of the conditional operator.}.

\

It might happen that some instances of rules of $ \lab{CL^*} $ present a duplication of the atomic formula in the conclusion: for example, an instance of $ \Unif{1} $ with $ a=b $ displays two formulas $ a \in N(x) $ in the conclusion. Since we want contraction to be height-preserving admissible, we deal with these cases by adding to the sequent calculus a new rule, in which the duplicated formulas are contracted into one. Such an operation is called applying a \textit{closure condition} to the rules (cf. \cite{negri2005proof}). Thus, rule $ \Unif{1}^* $ is the rule obtained applying the closure condition to $ \Unif{1} $ in case $ a=b $ and $ x =y $; rules $ \Unif{1}^{**} $ and $ \Unif{2}^{**} $ are obtained from $ \Unif{1} $ and $ \Unif{2} $, in case $ a=b $ and $ y=z $; and finally, $ \Abs{1}^* $ is obtained from $ \Abs{1} $ in the case $ a=b $. There is no need to define  additional rules which can be generated by the closure condition, since such rules either collapse or are subsumed by other rules of the calculus. For instance, the rule obtained applying the closure condition to $ \Unif{2} $, case $ a=b $ and $ x=y $, is the following:
$$
\infer[\Unif{2}^*]{a \in N(x), x \in a, \Gamma \Rightarrow \Delta}{z \in c, c \in N(x), a \in N(x), x \in a, \Gamma \Rightarrow \Delta}
$$
and this is the same instance we obtain applying the closure condition to $ \Unif{1}^* $. However, the  rules added by closure condition are not needed to prove completeness of the calculi; for this reason, we have not included them in the following sections (e.g. in the termination proof). 

\

To prove soundness of the rules with respect to the corresponding system of logics, we need to interpret relational atoms and labelled formulas in neighbourhood models. 
The notion of realization interprets the labels in neighbourhood frames, thus connecting the syntactic elements of the calculus with the semantic elements of the model.

\begin{definition} \label{realization}
	Let $ \mathcal{M}= \langle W, N, \llbracket \ \rrbracket \rangle $ be a neighbourhood model for $ \PCL $ or its extensions, $ \mathcal{S} $ a set of world labels and $ \mathcal{N} $ a set of neighbourhood labels. An $ \mathcal{SN} $\textit{-realization} over $ \mathcal{M} $ consists of a pair of functions $( \rho , \sigma) $ such that:
	\begin{itemize}[noitemsep]
		\item $ \rho:  \mathcal{S} \rightarrow W $ is  the function assigning to each $ x\in \mathcal{S} $ an element $ \rho (x) \in W $;
		\item $ \sigma: \mathcal{N} \rightarrow \mathcal{P}(W) $ is the function assigning to each $ a \in \mathcal{N} $ a neighbourhood $ \sigma(a) \in N(w) $, for $ w \in W $.
	\end{itemize}
	We introduce the notion of satisfiability of a formula $ \mathcal{F} $ under an $ \mathcal{SN} $-realization by cases on the form of $ \mathcal{F} $:
	\begin{itemize}[noitemsep]
		\item $ \mathcal{M} \vDash_{\rho, \sigma} a \in N(x) $ if $ \sigma(a) \in N(\rho(x) ) $;
		\item $ \mathcal{M} \vDash_{\rho, \sigma} a\subseteq b $ if $ \sigma (a) \subseteq \sigma(b) $;
		\item $ \mathcal{M} \vDash_{\rho, \sigma} y \in \{x\} $ if $ \rho(y) \in \sigma ( \{ x\}) $;
		\item $ \mathcal{M} \vDash_{\rho, \sigma} x:P $ if $ \rho (x) \Vdash P $ \footnote{This definition is extended in the standard way to formulas obtained by the classical propositional connectives.};
		\item $ \mathcal{M} \vDash_{\rho, \sigma} a \fu A $ if $ \sigma(a) \fu A  $;
		\item $ \mathcal{M} \vDash_{\rho, \sigma} a\fe A $ if $ \sigma(a) \fe A $;
		\item $ \mathcal{M} \vDash_{\rho, \sigma} x \Vdash_a A|B$ if $ \sigma(a) \in N(\rho(x)) $ and for some $ \beta \subseteq \sigma(a)$ it holds that $ \beta \Vdash^\exists A $ and $ \beta \Vdash^\forall A\rightarrow B $;
		\item $ \mathcal{M} \vDash_{\rho, \sigma} x : A>B$ if for all $ \sigma(a) \in N(\rho(x)) $, if  $ \mathcal{M} \vDash_{\rho, \sigma} a\fe A $ then $ \mathcal{M} \vDash_{\rho, \sigma} x \Vdash_a A|B$. 
	\end{itemize}

	\noindent Given a sequent $ \Gamma \Rightarrow \Delta $, let $ \mathcal{S} $, $ \mathcal{N} $ be the sets of world and neighbourhood labels occurring in $ \Gamma \cup \Delta$, and let $ (\rho, \sigma) $ be an $ \mathcal{SN} $-realization. Define $ \mathcal{M} \vDash _{\rho, \sigma} \Gamma \Rightarrow \Delta $ if either $ \mathcal{M} \nvDash_{\rho, \sigma} F $ for some $ F \in \Gamma $ or $ \mathcal{M}\vDash_{\rho, \sigma} G  $ for some $ G\in \Delta $. Define validity under all realizations by $ \mathcal{M} \vDash \Gamma \Rightarrow \Delta  $ if $ \mathcal{M}\vDash_{\rho, \sigma} \Gamma \Rightarrow \Delta $ for all $ (\rho, \sigma) $ and say that a sequent is valid in all neighbourhood models if $ \mathcal{M}\vDash_{\rho, \sigma} \Gamma \Rightarrow \Delta $ for all models $ \mathcal{M} $.
	
\end{definition}

\begin{theorem}[Soundness]
	If a sequent $ \Gamma \Rightarrow \Delta $ is derivable in $ \lab{CL^*}$, then it is valid in the corresponding class of neighbourhood models. 
\end{theorem}

\begin{proof}
	The proof is by straightforward induction on the height of the derivation, employing the notion of realization defined above. By means of example, we show soundness of the left and right rule for the conditional operator. 
	
	$ [\Lcon] $ 
	From a neighbourhood model and a realization which validates the premisses we construct a neighbourhood model which validates the conclusion. Let $ \mathcal{M} \vDash_{\rho, \sigma} a \in N(x), x:A>B, \Gamma \Rightarrow \Delta, a \fe A $ and $ \mathcal{M} \vDash_{\rho, \sigma} x \barra{a} A|B, a \in N(x), x:A>B, \Gamma \Rightarrow \Delta $. The only relevant case is the one in which $ \mathcal{M} \vDash_{\rho, \sigma} a \fe A $ and $ \mathcal{M} \nvDash_{\rho, \sigma}  x \barra{a} A|B $. From the former we have that $ \sigma(a) \fe A $; from the latter that for $ \sigma(a) \in \rho(x) $ and for all $ \beta \in \sigma(\alpha) $ it holds that either $ \beta \nVdash^\exists A $ or $  \beta \nVdash^\forall A \rightarrow B $. By definition, this means that 
	$ \mathcal{M} \nvDash_{\rho, \sigma} x: A >B$, for  $ \sigma(a) \in \rho(x) $; and thus, $ \mathcal{M} \vDash_{\rho, \sigma} a \in N(x), x:A>B, \Gamma \Rightarrow \Delta $. 
	
	$ [\Rcon] $ Suppose $ \mathcal{M} \vDash_{\rho, \sigma} a \in N(x), a \fe A, \Gamma \Rightarrow \Delta, x \barra{a} A|B $. We show that the conclusion is valid in the same model, under the same realization. There are two relevant cases: either the one in which $ \mathcal{M} \nvDash_{\rho, \sigma} a \fe A $ or the one in  which $ \mathcal{M} \vDash_{\rho, \sigma}  x \barra{a} A|B $. In the former case we have that $ \sigma(a) \nVdash^\exists A $, for $ \sigma(a) \in \rho (x) $. In the latter case, we have that for $ \sigma(a) \in \rho (x) $, there exists $ \beta \in \sigma(\alpha) $ such that $ \beta \Vdash^\exists A $ and $  \beta \Vdash^\forall A \rightarrow B $. In both cases it holds by definition that $ \mathcal{M} \vDash_{\rho, \sigma} x:A>B $; thus, $ \mathcal{M} \vDash_{\rho, \sigma} \Gamma \Rightarrow \Delta, x:A>B $.
	
\end{proof}

\section{Structural properties and syntactic completeness}

In this section we prove the main structural properties of  calculi $ \lab{CL^*} $.  
We start with some preliminary definitions and lemmas. By \textit{height} of a derivation we mean the number of nodes occurring in the longest derivation branch, minus one. 
We write $ \vdash^n \Gamma \Rightarrow \Delta $ meaning that there is a derivation of $ \Gamma \Rightarrow \Delta $ in $ \lab{CL^*} $ with height bounded by $ n $.

\begin{definition} 
	The weight of relational atoms is 0. As for the other labelled formulas, the label of formulas of the form $ x:A $ and $ x\Vdash_a A|B $ is $ x $; the label of formulas $ a\fu A $ and $ a\fe A $ is $ a $. We denote by $ l(\mathcal{F}) $ the label of a formula $ \mathcal{F} $, and by $ p(\mathcal{F}) $ the pure part of the formula, i.e., the part of the formula without the label and without the forcing relation. The \textit{weight} of a labelled formula is defined as a lexicographically ordered pair 
	$$ \langle w(p(\mathcal{F})), w(l(\mathcal{F}) ) \rangle $$ 
	where
	\begin{itemize}[noitemsep]
		\item for all world labels $ x $, $ w(x)=0 $;
		\item for all neighbourhood labels $ a $, $ w(a)=1 $;
		\item $w(p) = w (\bot) = 1$;
		\item $ w(A\circ B) = w(A) + w(B) + 1 $ for $ \circ $ conjunction, disjunction or implication;
		\item $ w(A|B)= w(A) +w(B)+2 $;
		\item $ w(A>B) = w(A)+w(B)+3 $.
	\end{itemize}
\end{definition}


\noindent The definition of substitution of labels given in \cite{negri2005proof} can be extended in an obvious way to the relational atoms and labelled formulas of $ \lab{CL^*} $. According to this definition we have, for example,  $(a\fe A)[b/a] \equiv b\fe A$, and $(x \barra{a} B|A)[y/x]\equiv y \barra{a} B|A$.
The calculus is routinely shown to enjoy the property of height preserving substitution both of world and neighbourhood labels. The proof is a straightforward extension of the same proof in \cite{negri2005proof}.
\begin{proposition}\label{subst}
	\
	\begin{enumerate}[label=(\roman*), noitemsep]
		\item If $\vdash^n \Gamma \Rightarrow \Delta$, then $\vdash^n \Gamma{[y/x]}\Rightarrow\Delta{[y/x]}$;
		\item If $\vdash^n \Gamma \Rightarrow \Delta$, then $\vdash^n \Gamma{[b/a]}\Rightarrow\Delta{[b/a]}$.
	\end{enumerate} 
\end{proposition}

The following Lemma, adapted from \cite{negri2005proof}, ensures derivability of generalized initial sequent. The proof proceeds by mutual induction on the weight of labelled formulas.

\begin{lemma}\label{generalized_initial_sequents}
	The following sequents are derivable in $ \lab{CL^*} $. 
	\begin{enumerate}[noitemsep]
		\item $ a \fe A, \Gamma \Rightarrow \Delta, a\fe A $
		\item $ a \fu A, \Gamma \Rightarrow \Delta, a \fu A $
		\item $ x\Vdash_a A|B, \Gamma \Rightarrow \Delta, x\Vdash_a A|B $
		\item $ x:A, \Gamma \Rightarrow \Delta, x:A $
	\end{enumerate}
\end{lemma}

%


To prove admissibility of the cut rule, we need admissibility of the structural rules and invertibility of all the rules. The reader can find a detailed proof of these properties in \cite{girlandothesis}. Both lemmas are proved by induction on the height of the derivation.



\begin{lemma} 
	Let $ \mathcal{F} $ be a relational atom or a labelled formula. The rules of weakening and contraction are height-preserving admissible in $ \lab{CL^*} $: 
	$$
	\infer[\Wkl]{\mathcal{F}, \Gamma \Rightarrow \Delta}{\Gamma \Rightarrow \Delta} \quad 
	\infer[\Wkr]{\Gamma \Rightarrow \Delta,\mathcal{F}}{\Gamma \Rightarrow \Delta} \quad 
	\infer[\Ctrl]{\mathcal{F},\Gamma \Rightarrow \Delta}{\mathcal{F},\mathcal{F}, \Gamma \Rightarrow \Delta} \quad 
	\infer[\Ctrr]{\Gamma \Rightarrow \Delta,\mathcal{F}}{\Gamma \Rightarrow \Delta,\mathcal{F},\mathcal{F}}	
	$$
\end{lemma}

\begin{lemma} 
	All the rules of $ \lab{CL^*} $ are  height-preserving invertible: if the conclusion of a rule is derivable with derivation height $ n $, its premiss(es) are derivable with at most the same derivation height.
\end{lemma}

\begin{theorem}[Cut-admissibility]
	\label{cut_3}
	The rule of cut is admissible in $ \lab{CL^*} $.
	$$
	\infer[\Cut]{\Gamma, \Gamma' \Rightarrow \Delta, \Delta'}{\Gamma \Rightarrow \Delta, \mathcal{F} \quad & \quad  \mathcal{F}, \Gamma' \Rightarrow \Delta'  }
	$$
\end{theorem}

\begin{proof}
	The proof is by primary induction on the weight of the cut formula and on secondary induction on the sum of the heights of the derivations of the premisses of $ \Cut $\footnote{Refer to \cite{structural} for the general methodology of proving cut-admissibility in labelled systems.}. We distinguish cases according to the rules applied to derive the premisses:
	\begin{itemize}[noitemsep]
		\item [$a)$] At least one of the premisses of $ \Cut $ is an initial sequent;
		\item [$b)$] The cut formula is not the principal formula in the derivation of at least one premiss;
		\item [$c)$] The cut formula is the principal formula of both derivations of the premisses. 
	\end{itemize}
	We only show the case of $c)$ in which the cut formula has the form $ A>B $. For the proof of propositional cases, refer to  \cite[Theorem 3.2.3]{structural}; for the proof of the other conditional cases refer to \cite{girlandothesis}.

	\
	
	\begin{adjustbox}{max width = \textwidth}
		$$
		\infer[\Cut]{ a \in N(x), \Gamma, \Gamma' \Rightarrow \Delta, \Delta' }{
			\infer[\Rcon]{\Gamma \Rightarrow \Delta, x:A>B }{ \deduce{ b \in N(x), b \fe A, \Gamma \Rightarrow \Delta, x \barra{a} A|B }{(1)} }
			&
			\infer[\Lcon]{ a \in N(x), x:A>B, \Gamma' \Rightarrow \Delta' }{ 
				\deduce{.. \Rightarrow \Delta', a \fe A }{(2)}
				&
				\deduce{x \barra{a} A|B, a \in N(x), x:A>B, \Gamma' \Rightarrow \Delta' }{(3)}
			}
		}
		$$
	\end{adjustbox}
	
	\
	
	\noindent We first apply $ \Cut $ on the premisses of $ \Lcon $. Both applications have a smaller sum of height of the premisses with respect to the original application of $ \Cut $:
	$$
	\mathcal{D}_1 = \quad \infer[\Cut]{a \in N(x), \Gamma, \Gamma' \Rightarrow \Delta, \Delta'}{ 
		\Gamma \Rightarrow \Delta, x:A>B
		&
		\deduce{ a \in N(x), x:A>B, \Gamma'  \Rightarrow \Delta', a \fe A }{(2)}
	}
	$$
	$$
	\mathcal{D}_2 = \quad \infer[\Cut]{a \in N(x), \Gamma, \Gamma' \Rightarrow \Delta, \Delta'}{ 
		\Gamma \Rightarrow \Delta, x:A>B
		&
		\deduce{x \barra{a} A|B, a \in N(x), x:A>B, \Gamma' \Rightarrow \Delta' }{(3)}
	}
	$$
	We combine the above with two occurrences of $ \Cut $, on formulas of lesser weight than the original cut formula.
	$$
	\infer[\Ctr]{a \in N(x), \Gamma, \Gamma' \Rightarrow \Delta, \Delta'}{
		\infer[\Cut]{a \in N(x)^3, \Gamma^3, \Gamma'^2 \Rightarrow \Delta^3, \Delta'^2}{
			\infer[\Cut]{a \in N(x)^2, \Gamma^2, \Gamma' \Rightarrow \Delta^2, \Delta', x \barra{a}A|B }{ \mathcal{D}_1\quad  & \deduce{
					b \in N(x), b \fe A, \Gamma \Rightarrow \Delta, x \barra{a} A|B }{(1)[b/a]}
			}
			&
			\quad  \mathcal{D}_2}
	}
	$$
\end{proof}

\noindent 
The axioms of each system of logic can be derived in the respective calculus. By admissibility of cut, the inference rules can be shown to be admissible, therefore obtaining a syntactic proof of completeness of the calculi. Details are given in the Appendix.

\begin{theorem}[Completeness via cut admissibility]
	\label{theorme:completeness_synth} 
	If a formula $ A $ is derivable in $\axiom{PCL} $ or in one of its extensions, then there is a derivation of $ \Rightarrow x: A $ in the calculus $ \lab{CL^*} $ for the corresponding logic.
\end{theorem}

\noindent We conclude the section by proving admissibility of rules $ \Repl{1} $ and $ \Repl{2}$ in their generalized form. This lemma will be used in Section \ref{sec:semantic_compl}, to prove completeness of the calculi featuring centering with respect to neighbourhood models.

\begin{lemma}
	Rules $ \Repl{1} $ and $ \Repl{2}$ generalized to all formulas of the language are admissible in $ \lab{CL^*} $.
\end{lemma}
\begin{proof} 
	Admissibility of the two rules is proven simultaneously, by induction on the weight of formulas. We only show the proof admissibility for $\Repl{1} $ (the other rule is symmetric). 
	Since contraction and cut are admissible in $ \lab{CL^*}  $, it is sufficient to show that sequent $ y \in \{x\}, A(x) \Rightarrow A(y) $ is derivable. 
	From this sequent and the premiss of $\Repl{1}  $, the conclusion of  $\Repl{1}  $ can be derived applying cut and contraction. 	
	We proceed by induction on the weight of formula $ A(x) $; there are several cases to consider. 
	
	\noindent \textbf{1.} $ A(x) \equiv x: \mathcal{F} $, $ A(y) \equiv y: \mathcal{F} $, where $ \mathcal{F} $ is a propositional formula. We consider the case $ A(x) \equiv x:B\rightarrow C $, $ A(y) \equiv y:B\rightarrow C $.
	
	\
	
		\begin{adjustbox}{max width = \textwidth}
		$$
			\infer[\Rright]{y \in \{x\}, x:B\rightarrow C  \Rightarrow  y:B\rightarrow C  }{ 
				\infer[\Lright]{ y \in \{x\}, x:B\rightarrow C, y:B  \Rightarrow  y: C }{
					\infer[\Repl{2}]{y \in \{x\}, y:B  \Rightarrow  y: C, x:B}{
						\deduce{y \in \{x\},x :B, y:B  \Rightarrow  y: C, x:B}{}
					}
					&
					\infer[\Repl{1}]{y \in \{x\}, y:B, x:C  \Rightarrow  y: C}{
						\deduce{y \in \{x\}, y:B ,x:C, y:C \Rightarrow  y: C, x:B}{}
					}
				}
			}
		$$
		\end{adjustbox}
	
	\
	
		\noindent In this case we need $ \Repl{2} $, applied to formulas of smaller weight, and the two premisses are derivable by \mbox{Lemma \ref{generalized_initial_sequents}}.

		\noindent \textbf{2.} $ A(x) \equiv  x \Vdash_{a} B|C$, $ A(y) \equiv  y \Vdash_{a} B|C$. 
		$$
		\infer[\Lbar]{y \in \{x\}, x \Vdash_{a} B|C \Rightarrow y \Vdash_{a} B|C}{
			\infer[\Repl{1}]{c \in N(x), c\subseteq a, c \Vdash^{\exists} B, c \Vdash^{\forall}B\rightarrow C, y \in \{x\} \Rightarrow   y \Vdash_{a} B|C }{
				\infer[\Rbar]{c \in N(y), c \in N(x), c\subseteq a, c \Vdash^{\exists} B, c \Vdash^{\forall}B\rightarrow C, y \in \{x\} \Rightarrow   y \Vdash_{a} B|C}{
					(1) \qquad
					&
					\qquad (2)
		}  } } 
		$$
		Where $ (1)$ is sequent $ c \in N(y), c \in N(x), c\subseteq a, c \Vdash^{\exists} B, c \Vdash^{\forall}B\rightarrow C, y \in \{x\} \Rightarrow   y \Vdash_{a} B|C, c \Vdash^{\exists} A $, and $ (2)$ is sequent $ c \in N(y), c \in N(x), c\subseteq a, c \Vdash^{\exists} B, c \Vdash^{\forall} B\rightarrow C, y \in \{x\} \Rightarrow   y \Vdash_{a} B|C, c \Vdash^{\forall} B\rightarrow C$.
		Rule $\Repl{1}  $ is applied to the atomic formula $ c \in N(x) $, which has smaller weight than  $ A(x) $. The lower premiss is derivable by \mbox{Lemma \ref{generalized_initial_sequents}}, the upper one by steps of $\Lfe $, $\Lfu$, $\Lc$, and \mbox{Lemma \ref{generalized_initial_sequents}}.

		
		\noindent \textbf{3.} $ A(x) \equiv x: B>C $, $ A(y) \equiv y:B> C $.
		
		$$
		\infer[\Rcon]{y \in \{x\}, x: B>C \Rightarrow y:B> C}{
			\infer[\Repl{2}]{ a \in N(y), a \Vdash^{\exists} B, y \in \{x\}, x: B>C \Rightarrow y \Vdash_{a} B|C}{
				\infer[\Lcon]{ a \in N(x), a \in N(y), a \Vdash^{\exists} B, y \in \{x\}, x: B>C \Rightarrow y \Vdash_{a} B|C}{  
					&
					x \Vdash_{a} B|C,  a \in N(x), a \in N(y), a \Vdash^{\exists} B, y \in \{x\}, x: B>C \Rightarrow y \Vdash_{a} B|C  } } }
		$$
		Rule $ \Repl{2} $ is applied to formula $ a \in N(y) $, of smaller weight. The leftmost premiss is the sequent $ a \in N(x), a \in N(y), a \Vdash^{\exists} B, y \in \{x\}, x: B>C \Rightarrow y \Vdash_{a} B|C, a \Vdash^{\exists} A 
		$, derivable by Case 1. 	
\end{proof}


\section{Decision procedure}

As they are, the calculi $ \lab{CL^*} $ are not terminating. Simple cases of loops are due to the repetition of the principal formula in the premiss of a rule; more complex cases of loop are generated by the interplay of world and neighbourhood labels. Our aim in this section is to provide a termination strategy for the calculi, thus defining a decision procedure for the logic.  

Here follows some examples of loops which might occur in root-first proof search.

\begin{example}\label{ex:loop_immediate}
	Loop generated by repeated applications of rule $ \Lfu $ to  $ a\fu C $.
	$$
		\infer[\Lfe]{a \fe A, a\fu C, \Gamma \Rightarrow \Delta}{
			\infer[\Lfu]{x\in a, x: A, a \fu C, \Gamma \Rightarrow \Delta}{
				\infer[\Lfu]{x\in a, x: A, x:C,a\fu C, \Gamma \Rightarrow \Delta}{
					\deduce{x\in a, x: A, x:C, x:C,a\fu C,\Gamma \Rightarrow \Delta}{\vdots}
				}
			}
		}
	$$
\end{example}

\begin{example}\label{ex:simple_loop}
	Loop generated by repeated applications of  $\Lcon $ and $ \Lbar $, with one conditional formula in the antecedent (only the left premiss of $ \Lcon $ is shown). 
	
	\
	
	\begin{adjustbox}{max width = \textwidth}
	$$
	\infer[\Lcon]{a \in N(x), x: A> B \Rightarrow \Delta}{
				\infer[\Lbar]{x \barra{a} A|B, a\in N(x),  x: A> B \Rightarrow \Delta}{
					\infer[\Lcon]{b\in N(x), b\subseteq a, a\in N(x),b\fe A, b\fu A\rightarrow B,  x: A> B \Rightarrow \Delta }{
						\infer[\Lbar]{x \barra{b} A|B,b\in N(x), b\subseteq a,  a\in N(x), b\fe A, b\fu A\rightarrow B,  x: A> B \Rightarrow \Delta}{
							\deduce{ c \in N(x), c \subseteq b, b\in N(x), b\subseteq a, a\in N(x),  c \fe A, c \fu A\rightarrow B,  \dots,  x: A> B \Rightarrow \Delta }{\vdots }
						}
		}}}
	$$
	\end{adjustbox}
\end{example}

\begin{example}\label{ex:complex_loop}
	Loop generated by repeated applications of rules $ \Lcon $ and $ \Lbar $, with two conditional formulas in the antecedent. Let $ \Omega = x: A> B, x: C> D  $. We write only the leftmost premiss of $ \Lcon $; next to $ \Lcon $ is written the number of applications of the rule. 
	
	\begin{adjustbox}{max width= \textwidth}
	$$
	\infer[\Lcon ~ \mbox{(2)}]{a\in N(x), \Omega \Rightarrow \Delta }{
		\infer[\Lbar]{x \barra{a} A|B, x\barra{a} C|D,\Omega \Rightarrow \Delta}{
			\infer[\Lbar]{b \subseteq a, b\in N(x),  b\fe A, b\fu A \rightarrow B,x\barra{a} C|D,  \Omega \Rightarrow \Delta}{
					\infer[\Lcon~ \mbox{(4)}]{c \in N(x), c \subseteq a, c \fe C, c \fu C\rightarrow D, \dots, \Omega \Rightarrow \Delta}{
						\infer[\Lbar] { x\barra{b} A|B, x\barra{b} C|D, x\barra{c} A|B, x\barra{c} C|D,\dots, \Omega\Rightarrow \Delta}{
							\infer[\Lbar]{ d \in N(x), d\subseteq b, d \fe A, d \fu A\rightarrow B, x\barra{b} C|D, x\barra{c} A|B, x\barra{c} C|D,\dots, \Omega\Rightarrow \Delta}{
								\infer[\Lbar]{e \in N(x), e\subseteq b, e \fe C, e \fu C\rightarrow D, x\barra{c} A|B, x\barra{c} C|D,\dots, \Omega\Rightarrow \Delta}{
									\infer[\Lbar]{f \in N(x), f\subseteq c, f \fe A, f \fu A\rightarrow B, x\barra{c} C|D,\dots, \Omega\Rightarrow \Delta}{
										\infer[\Lcon ~ \mbox{(4)}]{g \in N(x), g\subseteq c, g \fe C, g \fu C\rightarrow D, \dots, \Omega\Rightarrow \Delta}{ 
											\deduce{ x\barra{d} A|B, x\barra{d} C|D,x\barra{e} A|B, x\barra{e} C|D, x\barra{f} A|B, x\barra{f} C|D,x\barra{g} A|B, x\barra{g} C|D, \cdots, \Omega \Rightarrow \Delta }{\vdots}
										}
									}	
								}
							}
						} 
					}
				}
			}
		}				
		$$
	\end{adjustbox}				
					
\end{example}

We start by proving termination for $ \lab{CL} $, and then extend the proof strategy to sequent calculi for the extensions of $\PCL$. 
We recall that all logics of the $\PCL$ family are decidable and their complexity is known.

\begin{remark}
	\label{remark:complexity}
The complexity of the family of preferential conditional logics is studied in \cite{halpern1994complexity}, where it is shown that: for systems \emph{without} uniformity and absoluteness, the decision procedure is $ \mbox{PSPACE} $-complete. For logics with uniformity, the decision problem is $ \mbox{EXPTIME} $-complete. Finally, for systems with absoluteness, the decision problem is $ \mbox{NP} $-complete. 
\end{remark}

\subsection{Decidability for $ \lab{CL} $}

In this section we define a proof search strategy which blocks rules applications leading to non-terminating branches.
We first want to prevent applications of a rule  R  to a sequent that already contains the formulas introduced by R. This is done by defining saturation conditions for each rule. 

\begin{definition}\label{def:saturation_conditions_PCL}
		Let $ \mathcal{D} $ be a derivation in $ \lab{CL} $, and  $ \mathcal{B} = S_0, S_1, \dots$ a derivation branch, with $ S_i $ sequent $\Gamma_i \Rightarrow \Delta_i $, for $ i = 1, 2, \dots $ and $ S_0 $ sequent $ \, \Rightarrow x:A_0 $. Let $ \downarrow \Gamma_{k} / \downarrow \Delta_{k}$ denote the union of the antecedents/succedents occurring in the branch from $ S_0 $ up to $ S_k  $. 
		
		We say that a sequent $\Gamma \Rightarrow \Delta$ \emph{satisfies the saturation condition w.r.t. a rule} R if, whenever $\Gamma \Rightarrow \Delta$ contains the principal formulas  in the conclusion of R, then it  also contains the formulas introduced by \emph{one} of the premisses of R. The saturation conditions are listed in Figure \ref{fig:saturation_conditions_PCL}.
		
		We say that $\Gamma \Rightarrow \Delta$ is \emph{saturated} if there is no formula $ x:p $ occurring in $ \Gamma \cap \Delta $, there is no formula $ x: \bot $ occurring in $ \Gamma $, $\Gamma \Rightarrow \Delta$ satisfies \emph{all} saturation conditions listed in the upper part of Figure \ref{fig:saturation_conditions_PCL}.
\end{definition}

\begin{figure}[t]
	\begin{adjustbox}{max width = \textwidth}
	\begin{tabular}{l l}
	\hline
	& \\
	$ \Land $ & If $ x:A \wedge B $ is in $\downarrow \Gamma $, then $ x: A $ and $ x:B $ are in $\downarrow \Gamma $\\
	$ \Rand $ & If $ x:A\wedge B $ is in $\downarrow \Delta $, then $ x:A $ or $ x:B $ is in $\downarrow \Delta $\\
	$ \Lor$ & If $ x:A \lor B $ is in $\downarrow \Gamma $, then $ x: A $ or $ x:B $ is in $\downarrow \Gamma $\\
	$ \Ror $ & If $ x:A\lor B $ is in $\downarrow \Delta $, then $ x:A $ and $ x:B $ are in $\downarrow \Delta $\\
	$ \Lright $ &  If $ x:A \rightarrow B $ is in $\downarrow \Gamma $, then $ x: B $ is in $\downarrow \Gamma $ or $ x:A $ is in $\downarrow \Delta $\\
	$ \Rright $ & If $ x:A\rightarrow B $ is in $ \Delta $, then $ x:A $ is in $\downarrow \Gamma $ and $ x:B $ is in $\downarrow \Delta $\\[1ex]
	
	$ \Ref $ & If $ a $ is in $ \Gamma\cup \Delta $, $ \Delta $ then $ a \subseteq a $ is in $ \Gamma $\\
	$ \Trans $ & If $ a \subseteq b $ and $ b\subseteq c $ are in $ \Gamma $, then $ a \subseteq c $ is in $ \Gamma $\\
	$ \Lc $ & If $ x \in a$ and $ a \subseteq b $ are in $ \Gamma $, then $ x \in b $ is in $ \Gamma $\\[1ex]
	
	$ \Lfu $ & If $ x \in a $ and $ a \Vdash^{\forall} A $ are in $  \Gamma $, then $ x:A $ is in $\downarrow \Gamma $\\
	$ \Rfu $ & If $ a \Vdash^{\forall} A $ is in $ \downarrow \Delta $ then, for some $x$, $ x \in a $ is in $ \Gamma $ and $ x:A $ in $ \downarrow \Delta $\\
	$ \Lfe $ & If $ a \Vdash^{\exists}A $ is in $ \downarrow \Gamma $ then, for some $ x $, $ x \in a $ is in $ \Gamma $ and $ x:A $ is in $ \downarrow \Gamma $\\
	$ \Rfe $ & If $ x \in a $ is in $ \Gamma $ and $ a \Vdash^{\exists}A $ is in $ \Delta $, then $ x:A $ is in $ \downarrow \Delta $\\[1ex]
	
	$\Rcon$ & If $ x:A>B $ is in $ \downarrow \Delta $ then, for some $ a $, $ a \in N(x) $ is in $ \Gamma $, $ a \Vdash ^{\exists} A $  is in $ \downarrow \Gamma $\\
	& and $ x \Vdash_a A|B $ is in $  \Delta $\\
	$\Lcon$ & If $ a \in N(x) $ and $ x: A>B$ are in $ \Gamma $, then $ a \Vdash ^{\exists} A $ is in $ \downarrow \Delta $ or $x \Vdash_{a} B|A $ is\\
	& in $ \downarrow \Gamma $\\
	$ \Rbar $ & If $ c \in N(x) $ and $ c\subseteq a $ are in $ \Gamma $ and $ x \Vdash_{a} B|A $ is in $ \Delta $, then $ c \Vdash ^{\exists} A $ is in $ \Delta $ \\
	& or $ c \Vdash^{\forall} A\rightarrow B $ is in $ \downarrow \Delta $\\
	$ \Lbar  $ & If  $ x \Vdash_{a} B|A $ is in $ \downarrow \Gamma $ then, for some $ c $, $ c \in N(x) $ and $ c \subseteq a $ are in $ \Gamma $, $ c \Vdash ^{\exists} A $ \\
	& is in $ \downarrow \Gamma $ and $ c \Vdash^{\forall} A\rightarrow B $ is in $ \Gamma $\\[1ex]
	\hline
	& \\
		$\LconT$ & If $ a \in N(x) $ and $ x: A>B$ occur in $ \Gamma $, then $ a \Vdash ^{\exists} A $ is in $ \downarrow \Delta $ or $ a\fe A $ and   \\
	& $x \Vdash_{a} B|A $ are  in $ \downarrow \Gamma $\\
	$ \Mon$  & If $b\subseteq  a $ and $ a \Vdash^{\forall} A $ are in $  \Gamma $, then $ b \Vdash^{\forall} A $ is in $\Gamma $\\[1ex]
	\hline 	
%
\end{tabular}
\end{adjustbox}
\caption{Saturation conditions associated to $ \lab{CL} $ rules}
\label{fig:saturation_conditions_PCL}
\end{figure}

In Example \ref{ex:loop_immediate}, the second bottom-up application of $ \Lfu $ is blocked by the saturation condition associated to $ \Lfu $, since formula $ x:A $ already occurs in some antecedent of the derivation branch. 
In order to block the other cases of loop, we need to modify the rules of $ \lab{CL} $, and define a proof search strategy which governs the application of rules in root-first proof search.

\begin{definition}
	We modify the rule $ \Lcon $ as rule $\LconT $, and introduce rule $ \Mon $ in $ \lab{CL} $.
	
	\
	
	\begin{adjustbox}{max width = \textwidth}
			\begin{tabular}{c}
				$ \infer[\LconT]{a \in N(x),x:A>B, \Gamma \Rightarrow \Delta }{ a \in N(x),x:A>B, \Gamma \Rightarrow \Delta, a \fe A & a \fe A, x \Vdash_a A|B, a \in N(x),x:A>B, \Gamma \Rightarrow \Delta} $\\[1ex]
				$ \infer[\Mon]{b \subseteq a, a \fu A, \Gamma \Rightarrow \Delta}{b \subseteq a, b \fu A, a \fu A, \Gamma \Rightarrow \Delta}$
			\end{tabular} 
	\end{adjustbox}		
\end{definition}


\begin{lemma}
	In $ \lab{CL} $ it holds that:
	\begin{enumerate}
		\item Rule $ \Mon $ is admissible; 
		\item Rules $ \Lcon $ and $ \LconT $ are equivalent.
	\end{enumerate}
\end{lemma}

\begin{proof}
	The proof of 1 is immediate, by induction on the height of the derivation.
	To prove that $ \LconT $ is admissible if we have $ \Lcon $, apply weakening to the right premiss of $ \Lcon $ and then apply $ \LconT $ to obtain the conclusion of $ \Lcon $. To prove that $ \Lcon $ is admissible if we have $ \LconT $, we need admissibility of $ \Cut $. Let (1) and (2) denote the left and right premiss of $ \LconT $. The conclusion of $ \LconT $ is derived as follows:
	$$
	\infer[\Lcon]{a \in N(x), x:A>B, \Gamma \Rightarrow \Delta}{
		(1) &
		\infer[\Cut]{ x\Vdash_a A|B, a \in N(x), x:A>B, \Gamma \Rightarrow \Delta }{  
			\infer[\LWk]{ x\Vdash_a A|B, a \in N(x), x:A>B, \Gamma \Rightarrow \Delta, a \fe A }{ (1)}
			&
			(2)}
	}
\vspace{-0.4cm}
	$$	
\end{proof}

\begin{definition}
	The saturation conditions for $ \LconT $ and $ \Mon $ are defined in the lower part of Figure \ref{fig:saturation_conditions_PCL}. The list of saturation conditions needed for the termination proof is given by the conditions listed in the upper part of Figure \ref{fig:saturation_conditions_PCL}, in which the condition $ \Lcon $ is replaced by $ \LconT $, and the saturation condition for $ \Mon $ is added.
\end{definition}

We shall provide a decision procedure for sequent calculus $ \lab{CL} $ modified with rules $ \Mon $ and $ \LconT $.
We now define the proof search strategy.

\begin{definition} 
	\label{def:Strategy}
	When constructing root-first a derivation tree for a sequent $ \Rightarrow x_{0}:A $, apply the following \textit{proof search strategy}:
	\begin{enumerate}
		\item Apply rules which introduce a new label (\emph{dynamic} rules) only if rules which do not introduce a new label (\emph{static} rules) are not applicable; as an exception, apply $ \Rcon $ before $ \LconT $.
		\item If a sequent satisfies a saturation condition R, do not apply to that sequent the rule  R corresponding to the saturation condition.
	\end{enumerate}
	
\end{definition}

Observe that if the strategy is applied, world labels in root-first proof search are processed one after the other, according to the order in which they are generated.

\begin{example}\label{ex:solution}
	In Example \ref{ex:simple_loop}, the loop is stopped because of the proof search strategy and the saturation condition for $ \Lbar $, which blocks the uppermost application of the rule to the formula $ x \barra{b} A|B $. 
	The proof strategy requires static rules to be introduced \emph{before} dynamic rules. Thus, the static rule $ \Ref $ is applied \emph{before} the uppermost occurrence of the dynamic rule $ \Lbar $, introducing in the derivation formula $ b\subseteq b $\footnote{By a similar argument, also $ a\subseteq a $ and a number of other formulas should occur in the derivation before the uppermost application of $ \Lbar $; but they are not relevant here.}.
	The saturation condition for $ \Lbar $ applied to $ x \barra{b} A|B $ is met if there is some label $ d $ such that formulas $ d \subseteq b $, $ d \in N(x) $, $ d \fe A $ and $ d \fu A\rightarrow B $ already occur in the antecedent of a sequent occurring lower in the branch. Thus, if we take $ d $ to be $ b $ itself, the saturation condition is met and the uppermost occurrence of $ \Lbar $ cannot be applied.
	
	To see how the loop in Example \ref{ex:complex_loop} is stopped, we re-write the derivation according to the proof search strategy, highlighting the formulas to which rule $ \Lbar $ cannot be applied.
	Observe that here rule $ \LconT $ and $ \Mon $ become relevant.
	The same conventions as in Example \ref{ex:complex_loop} apply.

	\begin{adjustbox}{max width= \textwidth}
	$$
	\infer[\LconT\mbox{(2)}]{a\in N(x), \Omega \Rightarrow \Delta }{
		\infer[\Lbar]{a \fe A, a \fe C, x \barra{a} A|B, x\barra{a} C|D,\Omega \Rightarrow \Delta}{
			\infer[\Lbar]{b \subseteq a, b\in N(x),  b\fe A, b\fu A \rightarrow B,x\barra{a} C|D,  \Omega \Rightarrow \Delta}{
				\infer[\Ref ~\mbox{(2)}]{c \in N(x), c \subseteq a, c \fe C, c \fu C\rightarrow D, \dots, \Omega \Rightarrow \Delta}{
					\infer[\LconT \mbox{(4)}]{b \subseteq b, c \subseteq c,c \in N(x), c \subseteq a, c \fe C, c \fu C\rightarrow D, \dots, \Omega \Rightarrow \Delta}{
						\infer[\Lbar] {b \fe A, b \fe C, c \fe A, c \fe C, \mathbf{x\barra{b} A|B}, x\barra{b} C|D, x\barra{c} A|B, \mathbf{x\barra{c} C|D},\dots, \Omega\Rightarrow \Delta}{
							\infer[\Mon]{ d \in N(x), d\subseteq b, d \fe C, d \fu C\rightarrow D, x\barra{c} A|B,\dots, \Omega\Rightarrow \Delta}{
								\infer[\Lbar]{d \in N(x), d\subseteq b, d \fe C, d \fu C\rightarrow D, d \fu A\rightarrow B, x\barra{c} A|B,\dots, \Omega\Rightarrow \Delta}{
									\infer[\Mon]{e \in N(x), e\subseteq c, e \fe A, e \fu A\rightarrow B,\dots, \Omega\Rightarrow \Delta}{
										\infer[\Ref~ \mbox{(2)}]{e \in N(x), e\subseteq c, e \fe A, e \fu A\rightarrow B,  e \fu C\rightarrow D,\dots, \Omega\Rightarrow \Delta}{
									\infer[\LconT \mbox{(4)}]{d \subseteq d, e \subseteq e, e \in N(x), e\subseteq c, e \fe A, e \fu A\rightarrow B, e \fu C\rightarrow D,\dots, \Omega\Rightarrow \Delta}{
											\deduce{d\fe A, d \fe C, e \fe A, e \fe C, \mathbf{x\barra{d} A|B}, \mathbf{x\barra{d} C|D}, \mathbf{x\barra{e} A|B}, \mathbf{x\barra{e} C|D}, \cdots, \Omega \Rightarrow \Delta }{\vdots}
											}
										}
									}
								}
							}
						}
					}
				}
			}
		}
	}
	$$
	\end{adjustbox}				

\

\noindent Application of $ \Lbar $ to formula $ x \barra{b} A|B $ is blocked by the saturation condition, since $ b\subseteq b $, $ b\in N(x) $, $ b \fe A $ and $ b \fu A\rightarrow B $ all occur in the branch. Application of the rule to $ x\barra{c} C|D $ is blocked in a similar way. 
Application of $ \Lbar $ to $ x \barra{d} A|B $ is blocked, since all the formulas relevant for the saturation condition occur in the branch: $ d \subseteq d $ (introduced by $ \Ref $), $ d \in N(x) $, $ d \fe A $ (introduced by $ \LconT $) and $ d \fu A \rightarrow B $ (introduced by $ \Mon $). Application of $ \Lbar $ to the other formulas in the top sequent is blocked, and the loop is stopped.
\end{example}

Before tackling the termination proof, we define an ordering of the world labels according to their generation in the branch. The resulting tree of labels is needed to ensure that the number formulas introduced in root-first proof search is \emph{finite}.

\begin{definition} \label{def:graph_labels}
	Given a sequent $ \Gamma_k \Rightarrow \Delta_k $, let $ a $, $ b $ be neighbourhood labels and $ x $, $ y $ world labels occurring in $ \downarrow \Gamma_k \cup \downarrow\Delta_k $. We define:
	\begin{itemize}
		\item $ k(x)= min \{t \mid x \mathit{\  occurs \ in  \ } \Gamma_t\} $;
		\item $ k(a)= min \{t \mid a \mathit{\ occurs \ in \ } \Gamma_t \}$;
		\item $ x \g a $, \textquotedblleft $ x $ generates $ a $\textquotedblright ~ if for some 
		$ t\leqslant k $  and $ k(a)=t $, $ a \in N(x) $ occurs in $ \Gamma_t $;
		\item $ b \g y $, \textquotedblleft $ b $ generates $ y $\textquotedblright ~ if for some $ t\leqslant k $ and $ k(y)=t $, $ y\in b $ occurs in $ \Gamma_t $; 
		\item $ x \ww y$ if for some $ a $, $ x\g a $ and $ a \g  y$ and $ x \neq y $.
	\end{itemize}
\end{definition}

\noindent Intuitively, the relation $ x\g a $ holds between $ x $ and $ a $ if $ a\in N(x) $ is introduced at some stage in the derivation (thus, with an application of $ \Rcon $ or $ \Lbar $); similarly, the relation $ b \g y  $ holds between $ b $ and $ y $ if $ y \in b $ is introduced in the derivation (thus, applying either $ \Rfu $ or $ \Lfe $).

\begin{lemma} \label{lemma:graph_labels}
	Given a derivation branch, the following hold: 
	\begin{itemize}[noitemsep]
		\item[(a)] The relation $ \ww $ is acyclic and forms a tree  with  the world label $ x_0 $ at the root;
		\item[(b)] All labels occurring in a derivation branch also occur in the associated tree; that is, letting $ x \wclos y$ be the transitive closure of $ \ww $, if $ u $ occurs in $ \downarrow \Gamma_k $, then $ x_0 \wclos u $.
	\end{itemize}
\end{lemma}

\begin{proof}
	(a) follows from the definition of relation $ \ww $ and from the sequent calculus rules. 
	Observe that the relation of generation $ \g $ between world and neighbourhood labels is \emph{unique}: it is defined by taking into account the value $ k(a) $ or $ k(y) $, which keeps track of the derivation step at which the new label is introduced. At each derivation step, dynamic rules introduce a new label which is generated by at most one world or neighbourhood label. Take a sequent $ \Gamma_k \Rightarrow \Delta_k $, and suppose $ x \g a $: by definition, there is a $ t \leqslant k $ such that $ k(a) = t $ and $ a \in N(x) $ occurs in $ \Gamma_t $. Now suppose that $ a \in N(y) $ occurs in some $ \Gamma_{s} $, with $ t < s \leqslant k $. Since $ k(a) =t $, relation $ y \g a $ does not hold in the tree of labels. A similar reasoning holds for $ y \g a $. Thus, except for the label at the root, each label in a derivation branch has \emph{exactly} one parent according to the relation $ \g $ and, by definition, also according to $ \ww $.

	As for (b), it is easily proved by induction on $ k(u) \leqslant k $. If $ k(u)=0 $, then $ u=x_0 $ and (b) trivially holds. If $ k(u)=t>0 $, $ u $ does not occur in $ \Gamma_{t-1} $ and $ u $ occurs in $ \Gamma_t $. This means that there exist a $ v $ and a $ b $ such that $ b\in N(v)$ occurs in $ \Gamma_{t-1} $ and $ u\in b $ occurs in $ \Gamma_t $; thus, $ k(v) < k(u) $. By inductive hypothesis, $ x_0\wclos  v  $; since $ v \ww u $, also $ x_0 \wclos u $ holds. 
\end{proof}

\begin{definition} \label{def:conditional_degree}
	The \emph{size} of a formula $ A $, denoted by $ \size{A}$, is the number of symbols occurring in $ A $.
	
	The \emph{conditional degree} of a formula $ A $ corresponds to the level of nesting of the conditional operator in $ A $ and is defined as follows: 
	\begin{itemize}[noitemsep]
		\item $ d(p)= d(\bot) =0 $ for $ p $ atomic; 
		\item $ d(C \circ D)= max(d(C), d(D)) $ for $ \circ \in \{ \wedge, \lor, \rightarrow \}$; 
		\item $ d(C>D)=  max(d(C) , d(D)) +1 $. 
	\end{itemize}
	Given a sequent $ \Gamma \Rightarrow \Delta $ occurring in a derivation branch $ \mathcal{B} $, the conditional degree of a world label is the highest conditional degree among the formulas it labels:
	$$ 
	d(x) =  \textit{max}\{ d(C) \mid x:C \in \downarrow \Gamma \cup \downarrow \Delta \}.
	$$
 
\end{definition}

\noindent We now prove that the proof search strategy ensures termination.

\begin{theorem}[Termination]\label{theorem:termination}
	Root-first proof search for a $ \lab{CL} $ derivation for a sequent $ \Rightarrow x_0 :A_0$ built in accordance with the strategy terminates in a finite number of steps, with either an initial sequent or a saturated sequent.  
\end{theorem}

\begin{proof}
	To prove that root-first proof search terminates, we have to show that all the branches of a derivation starting with $  \Rightarrow x_0 :A_0 $ and built in accordance with the proof search strategy  are finite. We take an arbitrary derivation branch $ \mathcal{B} $.
	Since $ \lab{CL} $ rules do not increase the complexity of formulas when going from the conclusion to the premiss(es), the only source of non-termination in the branch is the presence of an infinite number of labels. We need to show that the tree of labels associated to $ \mathcal{B} $ is finite. Let us call $ \grafo $ the graph associated to $ \mathcal{B} $ according to Definition \ref{def:graph_labels}. This amounts to prove that: 
\begin{enumerate}
	\item Each branch of $ \grafo $ has a finite length;
	\item Each node of $ \grafo $ has a finite number of immediate successors.
\end{enumerate}
	Claim 1 is proved by induction on the conditional degree of a label $ y $ occurring in the branch. If $ d(y) = 0 $, $ y $ labels either an atomic formula or a propositional formula. In any case, no new world labels are generated from $ y $, and the branch is finite. If $ d(y)>0 $, it means that $ y $ labels some conditional formula. In this case, $ y $ generates at least one world label $ z $, meaning that
	for some neighbourhood label $ a $, $ y \g a $ and $ a \g z $. By definition, $ y \g a $ if rule $ \Rcon $ or $ \Lbar $ are applied in the derivation branch, introducing formula $ a \in N(y) $. Similarly, $ a \g z $ if formula $ z \in a $ has been introduced in the branch by application of $ \Lfe  $ or $ \Rfu $. Thus, a new world label $ z $ can be generated from a world label $ y $ by a combination of the above rules, possibly with the addition of static rules. In any case, it holds that the conditional degree of the formulas labelled with $ z $ is strictly smaller than the conditional degree of the formulas labelled with $ y $. To see this, suppose that $ y:A>B $ occurs in the consequent of some sequent in the branch. Application of $ \Rcon $ introduces a relational atom $ a \in N(y) $, and generates a formula $ y \barra{a} A|B $ in the consequent. Application of rule $ \Rbar $ introduces in the consequent either formula $ a \fe A $, to which no dynamic rules can be applied, or formula $ a \fu A\rightarrow B $. In this case, rule $ \Rfu $ can be applied, and a new world label $ z \in a $ is generated, along with formula $ z:A\rightarrow B $ in the consequent. It holds that $ d(z) < d(y)   $, and similar considerations apply for the other rules combinations. 
	It holds that $ d(A_0) $ is bounded by the size of the formula $ A_0 $ at the root. Thus, for $ n = \size{A_0} $, the maximal length of each branch of $ \grafo $ is bounded by $ O(n) $.

	\
	
	Proving claim 2 requires some care. By definition, a world label $ z $ is generated by a world label $ y $ if there is some neighbourhood label $  a$ such that $ y \g a $ and $ a \g z $, for $ k(y) = s $, $ k(a) = t $ and $ k(z)= u $ with $ s<t<u $. To prove that the number of world labels generated by some $ y $ is finite, we need to prove that:
	\begin{itemize}
		\item[$a)$] A world label $ y $ generates a finite number of neighbourhood labels;
		\item[$b)$] A neighbourhood label $ a $ generates a finite number of new world labels.
	\end{itemize} 
	As for $ a) $, observe that a new neighbourhood label can be generated by application of $ \Rcon $ or $ \Lbar $. In the former case, the rule is applied to some formula $ y:A>B $ occurring in $ \Delta_{t-1} $. Since the formula disappears from $ \Delta_t $, rule $ \Rcon $ can be applied only once. Thus, the number of new neighbourhood labels linearly depends on the size of the formula $ A_0 $ at the root of the sequent. \\
	The case in which the new neighbourhood is generated by $ \Lbar $ is more complex, since the rule may interact with rule $ \LconT $, as shown in Examples \ref{ex:simple_loop} and \ref{ex:complex_loop}. 
	To see how the loop is stopped in the general case, suppose that one neighbourhood label $ a \in N(y) $ occurs in the antecedent of some sequent in $ \mathcal{B} $, along with $ n $ conditional formulas $ y:A_1>B_1, \dots , y:A_n>B_n $. 
	After $ n $ applications of $ \LconT $, $ n $ formulas $ y \barra{a}A_1|B_1, \dots , y \barra{a} A_n|B_n $ occur in the antecedent. By $ n $ applications of $ \Lbar $, $ n $ new neighbourhood label are generated, along with the following formulas in the antecedent: 

	\

	\begin{tabular}{c }
	$b_1 \subseteq a, b_1 \in N(y), b_1 \fe A_1, b_1 \fu A_1 \rightarrow B_1$ \\
	\multicolumn{1}{c}{\vdots}\\
	$b_n \subseteq a, b_n \in N(y), b_n \fe A_n, b_n \fu A_n \rightarrow B_n$\\
	\end{tabular}
	
	\
	
	\noindent Now, rule $ \LconT $ can be applied to all the conditional formulas and all the neighbourhood just introduced. Thus, $ n\cdot n $ formulas are generated in the antecedent, along with formulas $ b_1 \subseteq b_1 , \dots, b_n \subseteq b_n $ introduced by $ \Ref $. 
	
	\
	
	\begin{tabular}{c c c}
	$b_1 \subseteq b_1, y\barra{b_1} A_1|B_1,$ & \dots &$ y\barra{b_1} A_n|B_n$\\
	\vdots &  & \vdots \\
	$b_n \subseteq b_n, y\barra{b_n} A_1|B_1,$ & \dots &$ y\barra{b_n} A_n|B_n$\\
	\end{tabular}
	
	\
	
	\noindent In principle, application of $\Lbar  $ yields $ n\cdot n $ new neighbourhood labels; however, $ n $ applications of the rule are blocked by the saturation condition associated to the rule. More precisely, $ \Lbar $ cannot be applied to formula $  y\barra{b_1} A_1|B_1 $, because formulas $ b_1 \subseteq b_1$, $b_1 \in N(x)$, $b_1 \fe A_1 $ and $ b_1 \fu A\rightarrow B $ occur lower in the branch. Similarly, the saturation condition for $ \Lbar $ blocks applications of the rule to formulas $ y\barra{b_2} A_2|B_2 $, $ y \barra{b_3} A_3|B_3 $, and so on. Thus, only $ n (n-1) $ new neighbourhood labels are generated. Let $ k = n-1 $.
	
	\
	
	\noindent
	\begin{tabular}{c c c }
		$c^1_2 \subseteq b_1,  c^1_2 \fe A_2, c^1_2 \fu A_2 \rightarrow B_2$ & \dots & $c^1_n \subseteq b_1, c^1_n \fe A_n, c^1_n \fu A_n \rightarrow B_n$ \\
		\vdots &  & \vdots \\
		$c^n_1 \subseteq b_n, c^n_1 \fe A_1, c^n_1 \fu A_1 \rightarrow B_1$ & \dots & $c^n_k \subseteq b_n, b^n_k \fe A_{k}, b^n_{k} \fu A_{k} \rightarrow B_{k}$ \\
	\end{tabular}

	\
	
	\noindent
	Before applying $ \LconT $, we exhaustively apply the static rules of $ \Ref $ and $ \Mon $, obtaining the following formulas (recall that $ k = n-1 $):
	
	\
	
	\begin{tabular}{c c c}
		$ c^1_2 \subseteq c^1_2, c^1_2 \fu A_1 \rightarrow B_1 $ & \dots & $ c^1_n \subseteq c^1_n, c^1_n \fu A_1 \rightarrow B_1 $\\
		\vdots &  & \vdots \\
		$ c^n_1 \subseteq c^n_1, c^n_1 \fu A_n \rightarrow B_n $ & \dots & $ c^n_k \subseteq c^n_k, c^n_k \fu A_n \rightarrow B_n $\\
	\end{tabular}

	\
	
	We now apply $ \LconT $, and introduce $ n \cdot n (n-1) $ formulas to which $ \Lbar $ can be applied. Let us consider the $ n $ formulas generated from application of the rule to label $ c^1_2 $. Recall that $ \LconT $ also introduces local forcing formulas.
	$$
	c^1_2 \fe A_1, \dots, c^1_2 \fe A_n, y \barra{c^{1}_{2}} A_1 |B_1, y \barra{c^{1}_{2}} A_2 |B_2, y \barra{c^{1}_{2}} A_3 |B_3, \dots , y \barra{c^{1}_{2}} A_n |B_n   
	$$
	The application of $ \Lbar $ to formula $ y \barra{c^{1}_{2}} A_2 |B_2 $ is blocked by the saturation condition: formulas $ c^1_2 \subseteq c^1_2 $, $ c^1_2 \fe A_2 $ and $ c^1_2 \fu A_2 \rightarrow B_2 $ occur in the branch. Application of the rule to $ y \barra{c^{1}_{2}} A_1 |B_1 $ is also blocked: formulas $ c^1_2 \fe A_1 $ and $ c^1_2 \fu A_1 \rightarrow B_1 $ have been introduced in the branch by $ \LconT $ and $ \Mon $ respectively. Thus, rule  $ \Lbar $ can be applied only $ n(n-1)(n-2) $ times, generating  the same number of new neighbourhood labels.
	The process continues: after the next applications of $ \LconT $ and $ \Lbar $, $ n(n-1)(n-2)(n-3) $ new labels are introduced, and so on. The number of $ \Lbar $ rule applications blocked by the saturation condition strictly increases, until all the generated neighbourhood labels are blocked.
	
	To be more precise, count as one step in the generation process all applications of $ \LconT  $, $ \Mon $, $ \Ref $ and $ \Lbar $ to a sequent. During the $ i $-th step, rule $ \LconT $ generates a number $ n $ of formulas $ x \barra{e} G|H $ for each neighbourhood label occurring in the sequent. Then, rule $ \Lbar $ can be applied, introducing a new neighbourhood label for each application. However, out of every $ n $ formulas $ x \barra{e} G|H$, $ i-1 $ applications of $ \Lbar $ are blocked.
	$$ 
	\# \textit{ of new neighbourhood labels at the } i^{th} \textit{ step } = ~ \frac{n!}{(n-i)!} 
	$$ 
	It follows that after $ n+1 $ steps, all the generated neighbourhood labels are blocked and, as a consequence, all applications of $ \LconT $ and $ \Lbar $ are blocked. 
	In general, for each neighbourhood label $ a \in N(y) $ and $ n $ conditional formulas labelled with $ y $, we generate at most 
	$$
	\sum_{k=1}^{n-1} \frac{n!}{(n-k)!}
	$$
	new neighbourhood labels. The number of neighbourhood labels generated by $ \Lbar $ and $ \Lcon $ is bounded by $ O((n-2) \cdot n!) $, since at most $ n-2 $ terms appear in the sum of labels, and the biggest term in the sum is $ n! $. This can be approximated to $ O(n^2 \cdot n!) $.
	
	\
	
	To prove $ b) $, recall that a new world label $ z $ is generated from a neighbourhood label $ a $ if rule $ \Lfe $, $ \Rfu $, or $ \Lbar $ are applied in the derivation. Since in all these rules the principal formula disappears from the premiss, each rule can be applied at most once to each suitable formula, generating one world label for each application. Thus, the number of world labels generated linearly depends on the size of the formula $ A_0 $ at the root of the derivation.
	
	\
	
	Since $ \grafo $ has a finite number of nodes, the world and neighbourhood labels in the derivation are finite. Since the pure formulas are in a finite number (all subformulas of $ A_0 $), in a finite number of steps proof proof search terminates, yielding either a saturated sequent or an initial sequent. 
\end{proof}

Take $ n =\size{A_0} $. The number of labels generated from a node of $ \grafo $ is counted as follows. The number of neighbourhood labels generated by $ \Rcon $ is $ O(n) $. Since the number of conditional formulas in the derivation is bounded by $ \size{A_0} $, the number of neighbourhood labels generated by $ \Lbar $ and $ \Lcon $ is bounded by $ O(n^2\cdot n!) $. Each neighbourhood label generates one new world labels; thus, the maximal number of world labels generated from a world label is bounded by $ O(n^2 \cdot n!) $. 
To conclude, since the maximal length of each branch is bounded by $ O(n) $, the maximal number of world labels introduced in a derivation branch is bounded by $ O(n^3\cdot n!) $. 
To obtain a complexity bound for the decision procedure associated to $ \lab{CL} $, the maximal number of labels has to be combined with the number of formulas generated at each step. The exponential bound on labels, however, already shows that the complexity of the decision procedure is $\mbox{NEXPTIME}$, far from the $\mbox{PSPACE} $ bound known for $ \PCL $ (see Remark \ref{remark:complexity}).

\subsection{Decidability for extensions}

Theorem \ref{theorem:termination} can be extended to the calculi for most extensions of $\PCL$. 

We show how sequent calculi for logics with normality, total reflexivity, weak centering, centering and uniformity terminate. We do not treat extensions of $ \lab{CL} $ with the rules for absoluteness. In these logics  all $N(x)$  are the same, and there is no need to keep track of the system of neighbourhood $ N(x) $ to which a certain neighbourhood $ \alpha $ belongs. This simplification is not reflected by sequent calculus $ \lab{CLA} $, which is defined as a \emph{modular} extension of $ \lab{CL} $. Thus, proving termination of $ \lab{CLA} $ is not worth, since the simplest extension of $ \PCL $ would have the most complex decision procedure\footnote{Refer to \cite{girlandothesis} for terminating a labelled sequent calculus more suitable to treat the condition of absoluteness. The resulting decision procedure, however, is still not optimal.}.

In order to treat the extensions of $\PCL$  we define  saturation conditions for the additional rules  and prove that the tree of labels corresponding to a derivation branch is finite. 

The rules we are concerned with are $ \N $, $ \emp $, $ \T $, $ \W $, $ \C $, $ \Single $, $ \Repl{1} $, $ \Repl{2} $, $ \Unif{1} $ and $ \Unif{2} $. Proof of termination for sequent calculi displaying a combination of these rules can be obtained by combining the proof strategies exposed in this section.
We start by adding to the conditions in Figure \ref{fig:saturation_conditions_PCL} the saturation conditions for these new rules, listed in Figure \ref{fig:saturation_cond_ext}.

\begin{figure}
	\begin{adjustbox}{max width = \textwidth}
	\begin{tabular}{l l}
		\hline 
		& \\
		$ \emp $ & If  $ a \in N(x) $ is in $ \Gamma $ then $ y \in a $ is in $ \Gamma $ for some $y$\\
		$ \N $ & If $ x $ is in $ \downarrow \Gamma \cup \downarrow \Delta $ then for some $ a $, $ a \in N(x) $ is in $ \Gamma $\\
		$ \T $ & If $ x $ is in $ \downarrow \Gamma \cup \downarrow \Delta $, there is an $ a $ such that $ a \in N(x) $ and $ x \in a $ are in $ \Gamma $\\
		$ \W $ & If $ a \in N(x) $ is in $ \Gamma $ then $ x \in a  $ is in $ \Gamma $\\
		$ \C $ & If $ a \in N(x) $ is in $ \Gamma $, both $\{ x\} \in N(x) $ and  $ \{ x\} \subseteq a $ are in $ \Gamma $\\
		$ \Single $ & If $\{ x\} \in N(x) $ is in $ \Gamma $, then $ x \in \{x\} $ is in $ \Gamma $\\
		$ \Repl{1} $ & If $ y \in \{ x\} $ is in $ \Gamma $, and if some formula $ At(x) $ is in $  \Gamma  $, then $ At(y) $ is in  $\Gamma$\\
		$ \Repl{2} $ &  If $ y \in \{ x\} $ is in $ \Gamma $, and if some formula $ At(y) $ is in $  \Gamma  $, then $ At(x) $ is in  $\Gamma$\\
		$ \Unif{1} $ & If $ a \in N(x) $, $ y \in a $, $ b \in N(y) $ and $ z \in b $ are in $ \Gamma $, then for some $ c $, $ c\in N(x) $ \\
		& and $ z \in c $ are in $ \Gamma $\\
		$ \Unif{2} $ & If $ a \in N(x) $, $ y \in a $, $ b \in N(x) $ and $ z \in b $ are in $ \Gamma $, then for some $ c $, $ c\in N(y) $ \\
		& and $ z \in c $ are in $ \Gamma $\\[1ex]
		\hline
	\end{tabular}
	\end{adjustbox}
\caption{Saturation conditions for extensions}
\label{fig:saturation_cond_ext}
\end{figure}

\begin{definition}
	The proof search strategy from Definition \ref{def:Strategy} is supplemented with the following clause:
	\begin{itemize}
		\item[3.] Rule $ \emp $ can be applied to a sequent and a formula $a\in N(x)$ only if some formula $ a \fe A $ occurs in the consequent, or some formula $ a \fu A $ occurs in the antecedent.
	\end{itemize}
\end{definition}
Let us see how the proof search strategy stops the two new cases of loop generated by the rules for extensions.
Interaction of $ \emp $ and $ \N $ generate the following loop.
$$
\infer[\N]{ x: A, \Gamma \Rightarrow \Delta}{
\infer[\emp]{ a \in N(x), x: A, \Gamma \Rightarrow \Delta }{
\infer[\N]{ y \in a, a \in N(x), x: A, \Gamma \Rightarrow \Delta}{
\infer[\emp]{b \in N(y), y \in a, a \in N(x), x: A, \Gamma \Rightarrow \Delta}{
\deduce{z \in b, b \in N(y), y \in a, a \in N(x), x: A, \Gamma \Rightarrow \Delta}{
\vdots
}
 }
}
}
}
$$ 
If no formulas $ a \fe A $ occur in $ \Delta $ and no formulas  $ a \fu A $ occur in $ \Gamma $, the first application of rule $ \emp $ is blocked. Suppose $ a \fe A $ occurs in $ \Delta $. Then $ \emp $ is applied, but if restriction 3 is not met by neighbourhood $ b $, the second uppermost application of $ \emp $ is stopped. The number of formulas $ a \fe A $ in the consequent and $ a \fu A $ in the antecedent is bounded by the conditional degree of formulas at the root; thus, the loop is stopped. Intuitively, rule $ \emp $ needs to be applied only to the neighbourhood label introduced by $ \N $, to ensure that it is not empty\footnote{Refer to the derivation of axiom (N) in the Appendix.}. The neighbourhoods introduced by $ \T $, $ \Unif{1} $ or $ \Unif{2} $ already contains an element, so no other loops with $ \emp $ arise.

Applications of $ \Unif{1} $ and $ \Unif{2} $ generate the following loop, where we take $ \Omega = a \in N(x), y \in a, b \in N(y), z \in b $.
$$
\infer[\Unif{1}]{\Omega, \Gamma \Rightarrow \Delta }{
	\infer[\Unif{2}]{c \in N(x), z \in c, \Omega, \Gamma \Rightarrow \Delta}{
		\infer[\Unif{1}]{d \in N(y), z \in d, c \in N(y), z \in c, \Omega, \Gamma \Rightarrow \Delta}{
			\infer[\Unif{2}]{ e \in N(x), z \in e, d \in N(y), z \in d, c \in N(y), z \in c, \Omega, \Gamma \Rightarrow \Delta}{
				\deduce{f \in N(y), z \in f, e \in N(x), z \in e, d \in N(y), z \in d, c \in N(y), z \in c, \Omega, \Gamma \Rightarrow \Delta}{\vdots}
			}
		}
	}
}
$$
The saturation condition for $ \Unif{2} $ blocks the first bottom-up application of the rule: there is a neighbourhood label $ b $ such that $ b \in N(y) $ and $ z \in d $ are in $ \Gamma $. Similarly, a loop generated by $ \Repl{1} $ and $ \Repl{2} $ is blocked by their saturation conditions.

\

We now prove termination for the sequent calculi extending $ \lab{PCL} $, adapting the proof of termination for $ \lab{PCL} $ (Theorem \ref{theorem:termination}). Observe that Lemma \ref{lemma:graph_labels} holds for all the extensions considered: thus, the world labels occurring in a derivation branch form a tree according to the relation $ \ww $.

\begin{theorem}[Termination]\label{theorem:termination_ext}
	Root-first proof search for a sequent $ \Rightarrow x_0 :A_0$ in the sequent calculi $ \lab{CL^N} $, $ \lab{CL^T} $, $ \lab{CL^W} $, $ \lab{CL^C} $, $ \lab{CL^U} $, $ \lab{CL^{NU}} $, $ \lab{CL^{TU}} $, $ \lab{CL^{WU}} $ and $ \lab{CL^{CU}} $, built in accordance with the strategy, terminates in a finite number of steps, with either an initial sequent or a saturated sequent.  
\end{theorem}

\begin{proof}
	As in the proof of Theorem \ref{theorem:termination}, we need to check that the tree of labels $ \grafo $ associated to an arbitrary derivation branch is finite:
	\begin{enumerate}
		\item Each branch of $ \grafo $ has a finite length;
		\item Each node of $ \grafo $ has a finite number of immediate successors.
	\end{enumerate}
	As for 1, the proof remains the same as in  Theorem \ref{theorem:termination}. 
	Rule $ \emp $ introduces a new world label but, as we have seen, applications of this rule are restricted: the rule can be applied only if afterwards some rule of local forcing can be applied to the new world label. Since the number of local forcing formulas occurring in a derivation branch is bounded by the size of formula $ A_0 $, the  length of a branch in $ \grafo $ starting from a world label $ y $ is still bounded by $ O(n) $, for $ n= \size{A_0} $. 
	Rules $ \T $, $ \W $ and $ \Single $ introduce in derivation branch a world label $ x $ generated by $ x $ itself. By definition we required that $ x \neq y $ in order to have $ x \ww y $; and thus, the rules do not introduce a new node in the tree of labels. Rule $ \C $ does not introduce a new world label in the derivation.
	Replacement rules are applied only to atomic formulas, and operate exclusively on world labels which are already present in the derivation. Similarly, rules $ \Unif{1}$ and $ \Unif{2} $ do not introduce new world labels in the derivation; thus, they do not affect the length of a branch in $ \grafo $.

	\
	
	The proof of 2 remains basically the same as in Theorem \ref{theorem:termination}, meaning that the count of the number of new neighbourhood labels generated by one neighbourhood label and $ n $ formulas $ y:A_1>B_1\dots, y:A_n >B_n $ is still bounded by $ O(n^2 \cdot n!) $, for $ n = \size{A_0} $. 
	However, the number of neighbourhood labels generated from a world label increases, due to the presence of rules for extensions. Each rule $ \N $, $ \T $, $ \C $, if applicable, adds at most one neighbourhood label $ a \in N(y) $ or $ \{y\} \in N(y) $ to a world label $ y $. The number of neighbourhood labels generated by a world labels is bounded by the size of the formula $ A_0 $ at the root; in logics with normality, total reflexivity and centering this number is at most $ O(n+3) $.
	
	As for the rules of replacement, given a world label $ y $ and a formula $ z \in \{y\} $, these rules rules may introduce in the derivation formulas $ a \in N(z) $ or $ a \in N(y)  $. Thus, the number of world labels introduced from a world label $ y $ is bounded the size of formula $ A_0 $, and thus by $ O(n) $ (as before), to which we have to add the number of applications of replacement rules introducing relational atoms $ a \in N(y) $. Since replacement rules can be applied at most once to each $ z \in \{y\} $ and $ a \in N(z) $, the total number of neighbourhood labels is bounded by $ O(2n) $.  
	
	A similar reasoning holds for $ \Unif{1} $ and $ \Unif{2} $. The saturation conditions for uniformity prevent the application of both $ \Unif{1} $ and $ \Unif{2} $ to formulas $ c \in N(x) $ (or $ c \in N(y) $), and $ z \in c $ if the neighbourhood label has been generated by the rules of uniformity. Thus, only one rule of uniformity ($ \Unif{1} $ or $ \Unif{2} $) can be applied out of every 4 relational atoms $ a \in N(x)$, $y \in a$, $b\in N(y) $ (or $ b \in N(x) $) and $ z \in b $. Moreover, the rule can be applied at most once to these labels. We can thus estimate the maximal number of neighbourhood labels introduced by a world label to be bounded by $ O(2n) $.

	Following the same reasoning as for $ \lab{CL} $, we have that the maximal number of world labels generated from a world label for calculi without centering or uniformity is given by $ O(n^2 \cdot n! )$, while  for calculi with centering or uniformity the bound is $ O(2n \cdot n \cdot n! )$.
\end{proof}

To conclude, the maximal length of each branch in $ \grafo $ combined with the maximal number of nodes generated from a node yields the following maximal bounds for world labels introduced in a derivation branch:  
 $ O(n^3 \cdot n!) $ in case of calculi without centering or uniformity, and 
 $ O(2n \cdot n^2 \cdot n!) $ for calculi with centering or uniformity, always taking $ n = \size{A_0} $. In both cases, the decision procedure associated to the logic is $ \mbox{NEXPTIME} $. 

\section{Semantic completeness }
\label{sec:semantic_compl}

Completeness of a sequent calculus can be proved either with respect to the axiom system or with respect to the class of models for the logic. Theorem \ref{theorme:completeness_synth}, along with Theorem \ref{cut_3} of cut-admissibility ensures the completeness of all calculi with respect to the axiomatization of the corresponding logics. In this section we prove the semantic completeness of the calculi: we show that if a formula is valid in the class of neighbourhood models for a given logic, then it is derivable in the sequent corresponding calculus. As usual we prove the counterpositive statement: if a formula is not derivable in the sequent calculus, we can construct a countermodel (in the intended class). The model wile be extracted  from a saturated upper sequent.

Since the proof requires to build a countermodel from a saturated sequent, termination of the calculi is needed. For this reason, we prove semantic completeness of all the systems, except for those with the condition of absoluteness.

\subsection{Completeness for $  \lab{CL}  $}

\begin{theorem}\label{theorem:semantic_completeness}
	Let $ \Gamma \Rightarrow \Delta $ be a saturated upper sequent in a derivation in  $  \lab{CL} $. There exists a \textit{finite} countermodel $ \mathcal{M} $ that satisfies all formulas in $ \downarrow \Gamma $ and falsifies all formulas in $ \downarrow \Delta $.
\end{theorem}

\begin{proof}
	
	Since $ \Gamma \Rightarrow \Delta $ is saturated, it is the upper sequent of a branch $ \mathcal{B} $.  
	We construct a model $ \mathcal{M}_\mathcal{B} $ that satisfies all formulas in $ \downarrow \Gamma $ and falsifies all formulas in $ \downarrow \Delta $. The  countermodel contains the semantic informations encoded in the sequents of the derivation branch. Let 
	$$ 
	S_\mathcal{B} = \{ x \mid x \in (\downarrow \Gamma ~ \cup \downarrow \Delta) \} \qquad N_\mathcal{B} = \{ a \mid  a \in (\downarrow \Gamma~ \cup \downarrow \Delta) \} 
	$$
	Then, we associate to each $ a \in  N_\mathcal{B}  $ a neighbourhood as follows:
	$$
	\alpha_a = \{ y \in  S_\mathcal{B} \mid  y \in a \ \mathit{ belongs \ to } \ \Gamma \} 
	$$
	Thus, for each neighbourhood $ a $, $ \alpha_a \subseteq  S_\mathcal{B} $. We construct the neighbourhood model $ \mathcal{M}_\mathcal{B} = \langle W_\mathcal{B}, N_\mathcal{B}, \llbracket \ \rrbracket_\mathcal{B} \rangle$ as follows.
	\begin{itemize}[noitemsep]
		\item $ W_\mathcal{B}= S_\mathcal{B} $
		\item For any $ x \in W_\mathcal{B} $, $ N_\mathcal{B}(x)= \{ \alpha_a \ | \ a \in N(x) \ \mathit{belongs \ to} \ \downarrow \Gamma \} $
		\item For $ p $ atomic, $  \llbracket p\, \rrbracket_\mathcal{B} = \{ x \in W_\mathcal{B} \ | \ x:p \ \mathit{belongs \ to} \ \downarrow \Gamma  \}$ 
	\end{itemize}
	
	\noindent We now show that $\mathcal{M}_{\mathcal{B}} = \langle W_{\mathcal{B}}, N_{\mathcal{B}}, \llbracket \ \rrbracket_{\mathcal{B}} \rangle $ satisfies the property of non-emptiness for $ \mathbb{PCL} $ neighbourhood models:  we have to verify that every $ \alpha_a \in N(x)$ contains at least one element. 
	If $ a \in N(x) $ occurs in the sequent, it must have been introduced either by $ \Rcon $ or $ \Lbar $. If the neighbourhood label is not blocked, by the saturation conditions associated to both rules it holds that $ a\fe C $ occurs in $ \downarrow \Gamma $. Thus, by the saturation condition $ (\Rfe) $,  formula $ y\in a $ occurs in $ \Gamma $.
	

	\noindent Moreover, the model $ \mathcal{M}_\mathcal{B} $ satisfies the following property:
	\begin{quote}
		$(*)$	If $ a \subseteq b $ belongs to $ \Gamma $, then $ \alpha_{a} \subseteq \alpha_{b} $
	\end{quote}
	To verify $ (*) $, suppose $ y \in \alpha_{a} $. This means that $ y \in a $ belongs to $ \Gamma $; then, by the saturation condition $ \Lc$, also $ y \in b $ belongs to $ \Gamma $. By definition of the model we have $ y\in \alpha_{b}$, and thus that $ \alpha_{a}  \subseteq \alpha_{b} $.

	\
	
	\noindent Next, define a realization $ (\rho, \sigma) $ such that $ \rho(x) = x $ and $ \sigma(a)= \alpha_{a} $ and prove the following claims:
	\begin{itemize}[noitemsep]
		\item[] \textbf{[Claim 1]} If $ \mathcal{F} $ is in $ \downarrow \Gamma $, then $ \mathcal{M}_{\mathcal{B}} \vDash_{\rho, \sigma} \mathcal{F}$;
		\item[] \textbf{[Claim 2]} If $ \mathcal{F} $ is in $ \downarrow \Delta $, then $ \mathcal{M}_{\mathcal{B}} \nvDash_{\rho, \sigma} \mathcal{F}$;
	\end{itemize} 
	\noindent where $ \mathcal{F} $ denotes a labelled formula, i.e., $ \mathcal{F}$ is  $a \in N(x)$, $x\in a$, $a \subseteq b$, $x \Vdash^{\forall} A$, $x \Vdash^{\exists} A$, $x \Vdash_a A|B$, $x:A$,  $x: A>B   $. The two claims are proved by cases, by induction on the weight of the formula $ \mathcal{F} $.

	\textbf{[a]} If $ \mathcal{F} $ is a formula of the form $ a \in N(x) $, $ x \in a $ or $ a \subseteq b $, Claim 1 holds by definition of $ \mathcal{M}_{\mathcal{B}} $, and Claim 2 is empty. For the case of $ a \subseteq b $, employ the fact $(*)$ above.
	
	\textbf{[b]} If $ A $ is a labelled atomic formula $ x:p $, the claims hold by definition of the model; by the saturation condition associated to $ \init $ no inconsistencies arise. If $ A \equiv \bot $, the formula is not forced in any model and Claim 2  holds, while Claim 1 holds by the saturation clause associated to $ \Lbot $. If $ A $ is a conjunction, disjunction or implication, both claims hold for the corresponding saturation conditions and by inductive hypothesis on formulas on smaller weight.
	
	\textbf{[c]} If $ A \equiv a \Vdash ^{\exists} A $ is in $ \downarrow \Gamma $, then by the saturation clause associated to $ \Lfe $ for some $ x $ there are $ x \in a $, $ x:A $ are in $ \downarrow \Gamma $. By definition of the model $ \mathcal{M}_{\mathcal{B}} $, for some $ x $, $ x \in \alpha_{a} $. Then, since $ w(x:A) < w(a \Vdash^{\exists} A) $, apply the inductive hypothesis and obtain $\mathcal{M}_{\mathcal{B}} \vDash x:A $. Therefore, by definition of satisfiability, $ \mathcal{M}_{\mathcal{B}} \vDash  \alpha_{a} \Vdash ^{\exists} A $.
	
	\noindent If $  a \Vdash ^{\exists} A $ is in $ \downarrow \Delta $, then it is also in $ \Delta $. Consider an arbitrary world $ x $ in $ \alpha_{a} $. By definition of $\mathcal{M}_{\mathcal{B}}$ we have that $ x \in a $ is in $ \Gamma $; apply the saturation condition associated to $ \Rfu $ and obtain that $ x:A $ is in $ \downarrow \Delta $. By inductive hypothesis, $ \mathcal{M}_{\mathcal{B}} \nvDash x:A $; thus, since this line of reasoning holds for arbitrary $ x $, we can conclude by definition of satisfiability that $ \mathcal{M}_{\mathcal{B}} \nvDash \alpha_{a} \Vdash ^{\exists} A  $. The case in which $ A \equiv a \Vdash ^{\forall} A $ is similar.
	
	\textbf{[d]} If $ x \Vdash_{a} A|B $ is in $ \downarrow \Gamma $, then by the saturation condition associated to $ \Lbar $ for some $ c $ it holds that $ c \in N(x) $ and $ c \subseteq a $ are in $ \Gamma $, and $ a \Vdash ^{\exists} A $, $ a \Vdash^{\forall} A\rightarrow B$ are in $ \downarrow \Gamma $. By definition of the model, $ \alpha_c \subseteq \alpha_a $, and by inductive hypothesis $ \mathcal{M}_{\mathcal{B}} \vDash  \alpha_{c} \Vdash ^{\exists} A $ and $ \mathcal{M}_{\mathcal{B}} \vDash  \alpha_{c} \Vdash ^{\forall} A\rightarrow B  $. By definition, this yields $ \mathcal{M}_{\mathcal{B}} \vDash x \Vdash_{a}A|B $.

	\noindent If $ x \Vdash_{a} A|B $ is in $ \downarrow \Delta $, consider a neighbourhood $ \alpha_{c} \subseteq \alpha_a $ in $ N(x) $. Then by definition of $\mathcal{M}_{\mathcal{B}}$ we have that $ c \in N(x) $ and $ c \subseteq a $ are in $ \Gamma $; apply the saturation condition associated to $ \Rbar $ and obtain that either $ c \Vdash^{\exists} A $ or $ c \Vdash^{\forall} A \rightarrow B $ is in $ \downarrow \Delta $. By inductive hypothesis,  either $ \mathcal{M} \nvDash  \alpha_{c} \Vdash^{\exists} A  $ or $ \mathcal{M}_{\mathcal{B}} \nvDash \alpha_{c} \Vdash^{\forall} A \rightarrow B  $. In both cases, by definition $  \mathcal{M}_{\mathcal{B}} \nvDash x \Vdash_{a} A|B $.
	
	\textbf{[e]} If $ x: A>B $ is in $ \downarrow \Gamma $, then it is also in $ \Gamma $. Consider an arbitrary neighbourhood $ \alpha_{a} $ in $ N(x) $. By definition of $\mathcal{M}_{\mathcal{B}}$ we have that $ a \in N(x) $ is in $ \Gamma $; apply the saturation condition associated to $ \Lcon' $ and conclude that either $ a \Vdash ^{\exists} A $ is in $ \downarrow \Delta $, or $ x \Vdash_{a} A|B $ is in $ \downarrow \Gamma $. By inductive hypothesis, it holds that either $ \mathcal{M}_{\mathcal{B}} \nvDash \alpha_{a} \Vdash ^{\exists} A $ or $ \mathcal{M}_{\mathcal{B}} \vDash  x \Vdash_{a} A|B  $. In both cases, by definition $ \mathcal{M}_{\mathcal{B}} \vDash x :A>B $.
	
	\noindent If $ x: A>B $ is in $ \downarrow \Delta $, by the saturation condition associated to $ \Rcon $, for some $ a $ it holds that $ a\in N(x) $ is in $ \Gamma $, $ a \Vdash^{\exists} A $ is in $ \downarrow \Gamma $ and $ x \Vdash_{a} A|B $ is in $ \downarrow \Delta $. By inductive hypothesis, $ \mathcal{M}_{\mathcal{B}} \vDash \alpha_{a} \Vdash^{\exists} A  $ and $  \mathcal{M}_{\mathcal{B}} \nvDash x \Vdash_{a} A|B  $, thus, by definition, we have $ \mathcal{M}_{\mathcal{B}} \nvDash x: A>B $.
	\end{proof}

\noindent Theorem \ref{theorem:semantic_completeness} together with the soundness of $ \lab{CL} $ provides a constructive proof of the finite model property for the logics: if $A$ is satisfiable in a model (meaning that $ \lnot A $ is not valid), by   soundness of the calculi $ \lnot A $ is not provable. Thus by Theorem \ref{theorem:semantic_completeness} we build a finite countermodel of $ \lnot A $, that is a finite model in which $ A $ is satisfiable.
The same holds for the calculi for extensions of $ \PCL $, once their semantic completeness has been proved.

\subsection{Semantic completeness for extensions}

Semantic completeness for sequent calculi with normality, total reflexivity, weak centering and uniformity can be proved similarly as for $ \lab{CL} $. Extensions of the calculi with rules for centering require a modification on the countermodel construction, to account for singleton neighbourhoods.

\begin{theorem} 
	If $ A $ is valid in $ \PCL $ combined with normality, total reflexivity weak centring and uniformity, then sequent $ \Rightarrow x:A $ is derivable in $  \lab{CL} $ combined with the corresponding rules for normality, total reflexivity, weak centering and uniformity.
\end{theorem}

\begin{proof} 
	The proof proceeds as the one of Theorem \ref{theorem:semantic_completeness}. For the case of normality, a clause is added in the countermodel construction; however, Claim 1 and 2 continue to hold in the model. For the remaining cases, the countermodel construction does not change, and it only remains to verify that the countermodel $ \mathcal{M}_\mathcal{B} $ satisfies the properties of normality, total reflexivity, weak centering and uniformity, provided that the corresponding rules and saturation conditions are added to the calculus. 

\textit{Normality}: To construct a countermodel for logics featuring \textit{only} normality, the following case distinction applies, for $ Q = \forall, \exists $: 
 we need to add a clause to the definition of $ W_\mathcal{B} $, $ \alpha_a $ and for $ Q = \forall, \exists $:
\begin{itemize}[noitemsep]
	\item If $ a \in N(x) $ occurs in $ \Gamma $, and there are some formulas $ a \Vdash^Q  A $ in $ \downarrow \Gamma~ \cup \downarrow \Delta $, the countermodel $\mathcal{M}_{\mathcal{B}}$ is defined as in the case of $ \PCL $;
	\item If  $ a \in N(x) $ occurs in $ \Gamma $, but no formulas $ a \Vdash^Q  A $ occur in $ \downarrow \Gamma~ \cup \downarrow \Delta$, we set: $ W_\mathcal{B}= S_\mathcal{B} \cup \{u\}  $, for some variable $ u $ not occurring in $ \Gamma$; $ \alpha_a = \{u \}$ and $N_\mathcal{B}(u) = \{\{u\} \}  $.
\end{itemize} 
The model satisfies the condition of normality: according to the saturation condition $ \N $, for every $ x $ occurring in $  \downarrow \Gamma $, there is $ a $ such that $ a\in N(x) $ occurs in $ \Gamma $. By definition of $ \mathcal{M}_\mathcal{B} $, $ \alpha_a \in N_\mathcal{B}(x) $. 
Moreover, we have to verify that non-emptiness of the model holds also for the neighbourhood $ \alpha_a $ introduced by the rule. If there are some formulas $ a \Vdash^Q A $ occurring in $ \downarrow \Gamma \cup \downarrow \Delta $, the saturation condition associated to either $ \emp $, $ \Lfe $ or $ \Rfu $ ensures that there is at least one formula $ y \in a $ in $ \Gamma $. If there are no such formulas, the application of $ \N $ is not relevant to the derivation; following the definition, we introduce an arbitrary world $ u $ to be placed in the neighbourhood\footnote{There is no need to verify non-emptiness for stronger conditions of total reflexivity and weak centering, since the rules added to the calculus add a world belonging to the neighbourhood introduced.}.

\textit{Total reflexivity}: According to the saturation condition $ \T $, for every $ x $ occurring in $ \downarrow \Gamma\, \cup \downarrow \Delta $ also $ a \in N(x) $, $ x \in a $ occur in $ \Gamma $. By definition of $\mathcal{M}_{\mathcal{B}} $, this means that $\alpha_{a} \in N(x) $ and $ x \in \alpha_{a} $, and total reflexivity holds.

\textit{Weak centering}: Suppose $ \alpha_a \in N(x) $. We want to show that $ x \in \alpha_a $. By definition, if $ \alpha_a \in N(x) $ then $ a \in N(x) $ occurs in $ \Gamma $. By the saturation condition associated to $ \W $, it holds that also $ x \in a $ occurs in $ \Gamma $; thus, by definition of the model $ x \in \alpha_a $.

\textit{Uniformity}: Suppose $ y \in \bigcup N(x) $, which means that $ y \in \alpha_a $ and $ \alpha_a \in N(x) $. By definition, $ a \in N(x) $ and $ y \in a $ occur in $ \Gamma $. We have to show that $ \bigcup N(x) = \bigcup N(y) $, that is:
\begin{quote}
	$ z \in \bigcup N(x) $ iff $ z \in \bigcup N(y) $
\end{quote}
Assume $ z \in \bigcup N(x) $. This means that $ z \in \alpha_b $ and $ b \in N(x) $ and, by definition, $ z \in b $ and $ b \in N(x) $ occur in $ \Gamma $. By the saturation condition associated to $ \Unif{2} $, we have that for some $ c $, $ c \in N(y) $ and $ z \in c $ occur in $ \Gamma $. Thus, $ z \in \alpha_c $ and $ \alpha_c \in N(y) $, meaning that $ z \in \bigcup N(y) $. The saturation condition associated to $ \Unif{1} $ is needed to prove the other direction. 
\end{proof}

\begin{theorem} 
	If $ A $ is valid in $ \PCL $ extended with centering, then sequent $ \Rightarrow x:A $ is derivable in $  \lab{CL^C} $.
\end{theorem}

\begin{proof}
	In this case, worlds of the countermodel are not defined as the set $ S_\mathcal{B} $ of labels occurring in the branch, but as \textit{equivalence classes} $ [x] $ with respect to the relation $ y \in \{x\} $, which we will show to be an equivalence relation. Then, we require $ [x] $ to be contained in any neighbourhood of $N(x)$. For $ S_\mathcal{B} $, $ N_\mathcal{B} $ and $ \alpha_a $ as defined before, let 
\begin{quote}
	$ [x] = \{ y \in S_\mathcal{B} \mid y \in \{x\} \ \mathit{ occurs \ in }  \ \Gamma\} $;\\
	$ [x] \subseteq \alpha_a $, for $ a \in N(x)$ occurring in $ \Gamma$.
\end{quote}
We construct a model $ \mathcal{M}^c_\mathcal{B} = \langle W^c, N^c, \llbracket \ \rrbracket^c
\rangle $ as follows:
\begin{itemize}[noitemsep]
	\item $ W^c = \{ [ x] \mid  x \in S_\mathcal{B} \} $;
	\item for each $[ x ] \in W^c $, $ N^c( [x ])= \{ \alpha_a \mid a \in N(x) \ \mathit{belongs \ to} \ \downarrow \Gamma  \} $;
	\item for $ p $ atomic, $  \llbracket p\, \rrbracket^c = \{ [x] \in W^c \mid  x :p \ \mathit{belongs \ to} \ \downarrow \Gamma  \}$.
	
\end{itemize}
We first prove that $ y \in \{x\} $ is an equivalence relation. The relation is \textit{reflexive}: for each $ x $ occurring in $  \Gamma $, $ x \in \{x\} $ occurs in $ \Gamma $. This holds from the saturation conditions associated to $ \N $, $ \C $ and  $ \Single $. To prove that the relation is \textit{symmetric}, we have to show that if $ y \in \{x\} $ occurs in $ \Gamma $, then also $ x \in \{y\} $ occurs in $ \Gamma $. By reflexivity, we have that $ y \in \{y\} $. Thus, by the saturation condition associated to $ \Repl{2} $, we have that also $ x \in \{y\} $ belongs to $ \Gamma $. To prove the converse, use saturation condition associated to $ \Repl{1} $.
To show that the relation is \textit{transitive} we have to prove that if $ y \in \{x\} $ and $ x \in \{z\} $ occur in $ \Gamma $, also $ y \in \{z\} $ occurs in $ \Gamma $. By saturation conditions $ \N $ and $ \C $ and $ \Single $, we have that also $ \{ z\} \in N(z) $ occurs in $ \Gamma $. By the saturation condition associated to $ \Repl{1}$ applied to $ x \in \{z\} $, also $ \{z\} \in N(x) $ occurs in the sequent; thus, by the saturation condition associated to $ \C $ we have that both formulas $ \{ x\} \subseteq \{z\}$ and $ \{x\}\in N(z) $ occur in $ \Gamma $. Finally, by the saturation condition associated to $ \Lc $, since $ y \in \{x\} $ and $ \{x\} \subseteq \{z\} $, we have that $ y \in \{z\} $ occurs in the sequent.

Next we need to show that the definitions of $ N^c( [x ])$ and $  \llbracket p\,  \rrbracket^c$ do not depend on the chosen representative of the equivalence class in question.
\begin{quote}
	i) if $y\in [ x ] $, then $a \in N(x)$ is in  $\Gamma$ if and only if $a \in N(y)$ is in  $\Gamma$;\\
	ii) if $y\in [ x ] $, then $x:p $ is in $ \Gamma $ if and only if $ y:p $ is in $ \Gamma $.
\end{quote}
Fact $ i) $ follows from the saturation conditions associated to $\Repl{1}$ and $\Repl{2}$, applied to on the formulas $a\in N(x)$ and $ a \in N(y) $. Fact $ ii) $ follows from application of the same saturation conditions to $ x:p $ and $ y:p $.

The model $ \mathcal{M}^c_\mathcal{B} $ satisfies the property of centering. Observe that in our model $ \{ x\} $ corresponds to $ [x] $: both are defined as the set containing exactly one element, $ x $. Suppose $ \alpha_a \in N(x) $; we have to show that $ \{x\} \subseteq \alpha_a $ and $\{x\} \in N(x)  $. 
By definition of the model we have that $ [x] \subseteq \alpha_a $, and from this and $ \alpha_a \in N(x) $ it follows that $ [x] \in N([x]) $; thus, strong centering holds.

\noindent The following facts are needed in the proof Claims 1 and 2 below.
\begin{quote}
	$1)$	If $ a \subseteq b $ belongs to $ \Gamma $, then $ \alpha_{a} \subseteq \alpha_{b} $;\\
	$ 2) $ if $ [x] \in \llbracket A \rrbracket$ and $ y \in [x] $, then $ [y] \in \llbracket A \rrbracket$; \\
	$ 3) $ If $ [x] \in \llbracket A \rrbracket $, then $ x:A $ belongs to $ \downarrow \Gamma $.
\end{quote}

Fact $ 1) $ is proved in the same way as $ (*) $; the proofs of $ 2) $ and $ 3) $ are immediate from admissibility of $ \Repl{1} $ and $\Repl{2}  $ in their generalized form.

Finally, we define define a realization $ (\rho, \sigma) $ such that $ \rho(x) = [x] $ and $ \sigma(a)= \alpha_{a} $, and prove that:
\begin{itemize}[noitemsep]
	\item[] \textbf{[Claim 1]} If $ \mathcal{F} $ is in $ \downarrow \Gamma $, then $ \mathcal{M}^c_{\mathcal{B}} \vDash_{\rho, \sigma} \mathcal{F}$;
	\item[] \textbf{[Claim 2]} If $ \mathcal{F} $ is in $ \downarrow \Delta $, then $ \mathcal{M}^c_{\mathcal{B}} \nvDash_{\rho, \sigma}\mathcal{F}$.
\end{itemize} 
\noindent Again, $ \mathcal{F} $ denotes the labelled formulas of the language, including $ y \in \{x\}$, $ \{x\} \in N(x) $, $\{ x \} \in a$. The two cases are proved by distinction of cases, and by induction on the height of the derivation.
If $ \mathcal{F} $ is a relational formula that does not contain any singleton, Claim 1 holds by definition of the model, and Claim 2 is empty as in case a) of proof of the previous models. Similarly, if $ \mathcal{F} $ is either $ y \in \{x\} $, $\{x\} \in N(x)$ or $ \{x\} \subseteq a $, Claim 1 is satisfied by definition. 

The cases b)- e) of the previous proof remain unchanged; condition  2) ensures that all the elements of an equivalence class of world labels satisfy the same sets of formulas. 
\end{proof}

\section{Related work and discussion}
In this paper we have studied Preferential Conditional Logic $ \PCL $ and its extensions. We have first provided a natural semantics for this class of logics in terms of neighbourhood models. 
Neighbourhood models generalise Lewis' sphere models for counterfactual logics. We have given  a \emph{direct} proof of soundness and completeness of $\PCL $ and its extensions with respect to this class of models,  with the exception of $ \PP\WW $ and its extensions with (U) and (A).  
We have then presented labelled sequent calculi for all logics of the family. The calculi are modular and have standard proof-theoretical properties, the most important one being cut admissibility, by which completeness of the calculi easily follows. 
We have then tackled  termination of the calculi, with the aim of obtaining a decision procedure for each logic.  For all systems, except for the ones containing absoluteness, we have shown that by adopting a suitable strategy, it holds that every derivation either succeeds or ends by producing a finite tree.  With respect to the known complexity of the logics, the decisions procedures are not optimal, and further work is needed to obtain optimal procedures out of the labelled calculi. The last result we have shown is semantic completeness for the calculi, again with the exception of cases of absoluteness.

\

Concerning the semantics, a few works have considered  neighbourhood models for $\PCL$ or closely related logics. The relation between neighbourhood models and preferential models has been considered in \cite{negri2015sequent,girlandothesis} and is based on  a well-known duality between partial orders and so-called Alexandrov topologies. According to this result, neighbourhood models are build by associating to each world a topology in which the neighbourhoods are  the \emph{open} sets. For conditional logics this duality is studied in detailed in \cite{marti2013topological}. However, the topological semantics of \cite{marti2013topological}   imposes closure under arbitrary unions and non-empty intersections on the neighbourhoods. These conditions are  not required by the logic, and we have not assumed them in the definition of neighbourhood models. 

A kind of neighbourhood semantics, called Broccoli Semantics, has been considered in \cite{girard2007onions}. In this article, it is shown that the logic BL characterised by the Broccoli Semantics coincides with $\PCL$. Completeness of BL is obtained by Burgess' result. 

Yet another kind of neighbourhood semantics bearing some similarity to  ours is the Premise Semantics, considered in the seminal work by Veltman \cite{veltman1985logic}. Premise semantics is shown equivalent to preferential semantics (called ``ordered semantics''). Premise models are neighbourhood models which do not require any additional properties, as in our definition. However, the definition of the conditional is different from ours, as it considers arbitrary intersections of neighbourhoods. Then, the result of \emph{strong} completeness is proved indirectly by resorting to preferential semantics (whence generalising Burgess' result).

In this respect, the direct completeness result with respect to the neighbourhood semantics contained in this work is new. In future work we wish to complete it with the missing cases. 

\

Concerning proof systems, very few calculi are known for $\PCL$ and its extensions. Labelled  tableaux calculi for $\PCL$ and its extensions (including all the ones considered here, and Lewis' logics) are proposed in \cite{giordano2009tableau}.  The calculi are based on preferential semantics with the Limit assumption, and are defined by extending the language by  pseudo-modalities indexed on worlds. 
The tableaux calculi cover all logics considered in this work, but they are inherently different from the ones we introduced, due to the presence of Limit assumption. As a difference with the present work, termination is obtained by relatively complex blocking conditions. 
As a side note, the neighbourhood semantics could be  reformulated by assuming the Limit assumption as follows. Given a formula $A$, consider the set of neighbourhoods $\alpha\in N(x)$ minimal with respect to set inclusion such that $\alpha\fe A$, and define a conditional $A > B$ to be forced at $x$ if each neighbourhood $\alpha$ in this set forces universally $A \to B$. Corresponding calculi could possibly be  developed based on this semantics. 

Labelled sequent calculi based on preferential semantics for $ \PCL $ and its extensions, including counterfactual logics,  are presented in \cite{girlando2019uniform}. In this case, the semantics is defined \emph{without} the limit assumption. However, while there is a proof of termination for all systems of calculi, complexity issues are not analysed in detail.

An unlabelled sequent calculus for $ \PCL $ yielding an optimal $ \mbox{PSPACE}$ decision procedure is presented in \cite{schroder2010optimal}:
the calculus is obtained by closing one step rules by all possible cuts and by adding a specific rule for $\PCL$. The resulting system is undoubtedly significant, but the rules have a highly combinatorial nature and are overly complicated. In particular, a non-trivial calculation (although the algorithm is polynomial) is needed to obtain one backward instance of the (S)-rule for a given sequent.

Recently, a resolution calculus for $\PCL$ has been proposed in \cite{nalon2018resolution}. The calculus does not make use of labels, nor of any additional structure; it relies however on a non-trivial pre-processing of formulas (including renaming of subformulas and addition of propositional constants) in order to transform a formula into a suitable set of clauses to which the resolution rules can be applied. 

As a difference with Lewis' logics\footnote{Refer to \cite{girlando2016standard, girlando2017hypersequent} for recently proposed non-labelled calculi for Lewis' logics.},  it  is remarkable that today, 40 years since preferential logics has been introduced, no \emph{standard} unlabelled sequent calculi for $ \PCL $ or its extensions have been found, where by a standard calculus we mean a proof system with a fixed finite number of rules, each with a fixed finite number of premisses. 

Regarding labelled sequent calculi for preferential logics, from a computational viewpoint the main issue, to be explored in future work, is whether the calculi can be refined in order to achieve optimal complexity. This may lead to a redefinition the semantics itself, in order to obtain sharper labelled rules, or to a modification of calculus structure in itself.

\bibliography{one.bib}
\bibliographystyle{plain}

\appendix
\section*{Appendix}
	\begin{proof}[Proof of Theorem \ref{theorme:completeness_synth}]
		We have to show that the inference rules of $ \axiom{PCL} $ are admissible, and that the axioms of $ \axiom{PCL} $ and its extensions are derivable. By means of example, we show admissibility of (RCEA) and (RCK) in $ \lab{CL} $, and derivability of axioms (CM), (N), (T) and (U$_1  $) in $ \lab{CLT} $. More derivation examples can be found in \cite{girlandothesis}. 
		
		For (RCEA), suppose $ \vdash A\leftrightarrow B $. Thus, $ \vdash x:A\rightarrow B $, whence $ \Rightarrow x:A\rightarrow B  $, and  $ \vdash x:B\rightarrow A $, whence $ \Rightarrow x:B\rightarrow A $. We derive three sequents by application of $ \Cut $ and other rules, and show how to combine them into a derivation of $ \Rightarrow x:(A>C) \rightarrow (B>C) $, i.e., $ \vdash (A>C) \rightarrow (B>C) $. The other direction of the implication is similar.
		
		From $ \Rightarrow x:B\rightarrow A $ obtain by substitution $\Rightarrow y:B\rightarrow A  $.  Sequent $ y:B, y:B\rightarrow A \Rightarrow y:A $ is derivable.
		$$
		\infer[\Lfe]{(1) \quad  a \in N(x), a \fe B, x:A>C \Rightarrow x\barra{a} B|C, a \fe A }{
			\infer[\Rfe]{a \in N(x), y \in a, y: B, x:A>C \Rightarrow x\barra{a} B|C, a \fe A }{
				\infer[\Wk]{ a \in N(x), y \in a, y: B, x:A>C \Rightarrow x\barra{a} B|C, a \fe A, y :A }{ 
					\infer[\Cut]{y:A \Rightarrow y:B}{
						\infer[\Wk]{y:A \Rightarrow y:B\rightarrow A, y:B}{
							\Rightarrow y:B\rightarrow A,
						}
						&
						y:B, y:B\rightarrow A \Rightarrow y:A
					} 
				}
			}
		}
		$$
		From $ \Rightarrow x:A\rightarrow B  $ obtain by substitution $ \Rightarrow y:A\rightarrow B  $. Sequent $ y:A, y:A \rightarrow B \Rightarrow y:B $ is derivable. 
		
		\
		
		\begin{adjustbox}{max width = \textwidth}
			$$
			\infer[\Lfe]{(2) \quad c\subseteq a, a \in N(x), c \in N(x), a \fe B, c \fe A, c \fu A\rightarrow C, x:A>C \Rightarrow x\barra{a} B|C }{
				\infer[\Rfe]{ c\subseteq a, a \in N(x), c \in N(x), y \in c, a \fe B, y:A, c \fu A\rightarrow C,, x:A>C \Rightarrow x\barra{a} B|C  }{
					\infer[\Wk]{c\subseteq a, a \in N(x), c \in N(x), y \in c, a \fe B, y:A, c \fu A\rightarrow C, y : A\rightarrow C,, x:A>C \Rightarrow x\barra{a} B|C}{
						\infer[\Cut]{y:A \Rightarrow y:B}{
							\infer[\Wk]{y:A \Rightarrow y:A\rightarrow B, y:B}{
								\Rightarrow y:A\rightarrow B
							}
							&
							y:A, y:A \rightarrow B \Rightarrow y:B
						}
					}
				}
			}
			$$
		\end{adjustbox}
		
		\
		
		\noindent From  $ \Rightarrow x:B\rightarrow A $ obtain by substitution   $ \Rightarrow z:B\rightarrow A $, for some variable $ z $ different from $ y $. Sequent $z:C \Rightarrow z:C$ is derivable, as well as sequent $ z:B, z:B\rightarrow A \Rightarrow z:A $. 
		
		\
		
		\begin{adjustbox}{max width = \textwidth}
			$$
			\infer[\Rfu]{(3) \ c \subseteq a, a \in N(x), c \in N(x), c \fe A, c \fu A\rightarrow C, a \fe B, x:A>C \Rightarrow x \barra{a} B|C, c \fu B\rightarrow C }{
				\infer[\Lfu]{z \in c,  c \subseteq a, a \in N(x), c \in N(x), c \fe A, c \fu A\rightarrow C,a \fe B, x:A>C \Rightarrow x \barra{a} B|C, z:B\rightarrow C }{
					\infer[\Wk]{ z \in c,  c \subseteq a, a \in N(x), c \in N(x), c \fe A, c \fu A\rightarrow C, z: A\rightarrow C, a \fe B, x:A>C \Rightarrow x \barra{a} B|C,z:B\rightarrow C }{
						\infer[\Rright]{z:A\rightarrow C  \Rightarrow z:B\rightarrow C}{
							\infer[\Lright]{ z:A\rightarrow C, z:B \Rightarrow z:C }{
								z:C \Rightarrow z:C
								&
								\infer[\Cut]{z:B \Rightarrow z:A}{
									\infer[\Wk]{z:B \Rightarrow z:B\rightarrow A, z:A}{\Rightarrow z:B\rightarrow A}
									&
									z:B, z:B\rightarrow A \Rightarrow z:A 
								} 
							}
						}
					}
				}
			}
			$$
		\end{adjustbox}
		
		\

		\noindent To conclude, we combine the above derivations into the following:
		
		\

		\begin{adjustbox}{max width = \textwidth}
			$$
			\infer[\Rright]{\Rightarrow x: (A>C)\rightarrow (B>C) }{
				\infer[\Rcon]{x: A>C\Rightarrow x:B>C}{
					\infer[\Lcon]{a \in N(x), a \fe B, x:A>C \Rightarrow x\barra{a} B|C}{
						(1)
						&
						\infer[\Lbar]{x \barra{a} A|C, a \fe B, x:A>C \Rightarrow x\barra{a}B|C}{
							\infer[\Rbar]{c \subseteq a, a \in N(x), c \in N(x), c \fe A, c \fu A\rightarrow C, a \fe B, x:A>C \Rightarrow x\barra{a}B|C }{
								(2)\qquad 
								&
								\qquad (3)
							}
						}
					}
				}
			}
			$$
		\end{adjustbox}
		
		\
		
		\noindent For (RCK), suppose $ \vdash A\rightarrow B $. Thus, $ \Rightarrow x:A\rightarrow B $. We show how to derive sequent $ \Rightarrow x:(C>A) \rightarrow (C>B) $. From $ \Rightarrow x:A\rightarrow B $ obtain by substitution $ \Rightarrow y:A\rightarrow B $. Then, from this sequent and the derivable sequent $ y:C\Rightarrow y:C $ obtain sequent (1):
		
		\
		
		\begin{adjustbox}{max width = \textwidth}
			$$
			\infer[\Rfu]{(1) \quad b \subseteq a, a \in N(x), b \in N(x), b \fe C, b \fe C\rightarrow A,a \fe C, x:C>A \Rightarrow x \barra{a} C|B, b \fu C\rightarrow B}{
				\infer[\Lfu]{y \in b, b \subseteq a, a \in N(x), b \in N(x), b \fe C, b \fe C\rightarrow A,a \fe C, x:C>A \Rightarrow x \barra{a} C|B, y: C\rightarrow B}{
					\infer[\Wk]{y \in b, b \subseteq a, a \in N(x), b \in N(x), b \fe C, b \fe C\rightarrow A, y : C\rightarrow A, a \fe C, x:C>A \Rightarrow x \barra{a} C|B, y: C\rightarrow B}{
						\infer[\Rright]{y:C\rightarrow A \Rightarrow y:C\rightarrow B}{
							\infer[\Lright]{y:C\rightarrow A,y:C \Rightarrow y:B}{
								y:C \Rightarrow y:C
								&
								\infer[\Wk]{y:C, y:A \Rightarrow y:B}{
									y:A \Rightarrow y:B
								}
							}
						}
					}
				}
			}
			$$
		\end{adjustbox}
		
		\

		\noindent The following two sequents are derivable, by Lemma \ref{generalized_initial_sequents}:
		\begin{quote}
			(2) \ $ a \in N(x), a \fe C, x:C>A \Rightarrow x \barra{a} C|B, a \fe C  $ \\
			(3) $ b \subseteq a, a \in N(x), b \in N(x), b \fe C, b \fe C\rightarrow A,a \fe C, x:C>A \Rightarrow x \barra{a} C|B, b \fe C $\
		\end{quote}
		To conclude the proof, apply the following rules to $ y:C\rightarrow A \Rightarrow y:C\rightarrow B $:
		
		\
		
		\begin{adjustbox}{max width = \textwidth}
			$$
			\infer[\Rright]{\Rightarrow x:(C>A) \rightarrow (C>B)}{
				\infer[\Rcon]{x:C>A \Rightarrow x:C>B}{
					\infer[\Lcon]{ a \in N(x), a \fe C, x:C>A \Rightarrow x \barra{a} C|B }{
						(2)
						&
						\infer[\Lbar]{x \barra{a} C|A, a \in N(x), a \fe C, x:C>A \Rightarrow x \barra{a} C|B}{
							\infer[\Rbar]{b \subseteq a, a \in N(x), b \in N(x), b \fe C, b \fe C\rightarrow A,a \fe C, x:C>A \Rightarrow x \barra{a} C|B  }{
								(3) \qquad 
								&
								\qquad (1)
							}
						}
					}
				}
			}
			$$
		\end{adjustbox}

		\noindent Here follows the derivation of (CM), in which we have omitted three derivable left premisses: $ a \fe a \wedge B \dots \Rightarrow \dots a \fe A $, premiss of the lower occurrence of $ \Lcon $, $ b\fe A \dots \Rightarrow \dots b \fe A  $ premiss of the upper occurrence of $ \Lcon $, and $ c\subseteq a, a \fe A\wedge B \dots \Rightarrow \dots c \fe A\wedge B $, premiss of $ \Rbar $. 
		
		\
		
		\begin{adjustbox}{max width = \textwidth}
			$ 
			\infer[\Rright]{\Rightarrow x: (A > B) \land (A > C) \rightarrow ((A\land B) > C) }{
				\infer[\mathsf{L \wedge}]{x: (A > B) \land (A > C) \Rightarrow x: (A\land B) > C}{
					\infer[\Rcon]{x: A > B, x:A > C \Rightarrow x:(A\land B) > C}{
						\infer[\Lcon]{a \in N(x), a \fe A\wedge B, x: A > B, x:A > C \Rightarrow x \barra{a} A\wedge B | C  }{
							\infer[\Lbar]{x \barra{a} A|B, a \in N(x), a \fe A\wedge B, x: A > B, x:A > C \Rightarrow x \barra{a} A\wedge B | C   }{
								\infer[\Lcon]{ b \in N(x), b\subseteq a, a \in N(x), b\fe A, b\fe A\rightarrow B, a \fe A\wedge B, x: A > B, x:A > C \Rightarrow x \barra{a} A\wedge B | C }{
									\infer[\Lbar]{  x \barra{b} A|C, b \in N(x), b\subseteq a, a \in N(x),  b\fe A, b\fe A\rightarrow B , a \fe A\wedge B, x: A > B, x:A > C \Rightarrow x \barra{a} A\wedge B | C  }{
										\infer[\Trans]{c \in N(x), c \subseteq b, b \in N(x), b\subseteq a, a \in N(x), c\fe A, c\fe A\rightarrow C, b\fe A, b\fe A\rightarrow B , a \fe A\wedge B, x: A > B, x:A > C \Rightarrow x \barra{a} A\wedge B | C  }{
											\infer[\Rbar]{ c \subseteq a, c \subseteq b, c\fe A, c \fe A, c\fu A\rightarrow C, b\fe A, b\fe A\rightarrow B , a \fe A\wedge B, x: A > B, x:A > C \Rightarrow x \barra{a} A\wedge B | C }{
												\infer[\Rfu]{ c \subseteq a\dots c \subseteq b, c \fe A, c\fu A\rightarrow C, b\fe A, b\fe A\rightarrow B , a \fe A\wedge B \dots \Rightarrow \dots c\fu  A\wedge B \rightarrow C }{
													\infer[\Rright]{ y \in c , c \subseteq a, c \subseteq b, c \fe A, c\fu A\rightarrow C, b\fe A, b\fe A\rightarrow B , a \fe A\wedge B \dots \Rightarrow \dots y: A\wedge B \rightarrow C }{
														\infer[\Lfu]{y \in c , c \subseteq a, c \subseteq b, c \fe A, c\fu A\rightarrow C, b\fe A, b\fe A\rightarrow B , a \fe A\wedge B, y: A\wedge B   \dots \Rightarrow \dots y:C}{
															\infer[\Lright]{y \in c , c \subseteq a, c \subseteq b, c 	\fe A, y: A\rightarrow C, b\fe A, b\fu A\rightarrow B , a \fe A\wedge B, y: A\wedge B   \dots \Rightarrow \dots y:C}{
																\deduce{\dots y:C \Rightarrow y:C \dots}{\triangledown}
																&
																\deduce{\dots y:A\wedge B \Rightarrow y:A \dots}{\triangledown}
															}
														}
													}
												}
											}
										}
									}
								}
							}
						}
					}
				}
			}
			$
		\end{adjustbox}
		
		\

		\noindent Here follows the derivation of (N).
		
		\
			
			\begin{adjustbox}{max width = \textwidth}
			$$
			\infer[\mathsf{R\lnot}]{\Rightarrow x:\lnot (\top > \bot )}{
				\infer[\N]{x: \top>\bot\Rightarrow }{ 
					\infer[\Lcon]{a \in N(x), x: \top>\bot\Rightarrow }{
						\infer[\emp] {a \in N(x), x: \top>\bot\Rightarrow a \fe \top}{
							\infer[\Rfe]{y\in a, a \in N(x), x: \top>\bot\Rightarrow a \fe \top}{y\in a, a \in N(x), x: \top>\bot\Rightarrow a \fe \top, y :\top } 
						}
						&
						\infer[\Lbar]{ x\Vdash_a \top | \bot , y\in a, a \in N(x), x: \top>\bot\Rightarrow x:\bot}{
							\infer[\Lfu]{ c\in N(x), c \subseteq a, a \fe \top, a \fu \top \rightarrow \bot , y\in a, a \in N(x), x: \top>\bot\Rightarrow x:\bot }{ 
								\infer[\Lc]{c\in N(x), c \subseteq a, a \fe \top, a \fu \top \rightarrow \bot , y : \top \rightarrow \bot, y\in a, a \in N(x), x: \top>\bot\Rightarrow x:\bot }{ \dots \Rightarrow y:\top \qquad  & \qquad y:\bot \Rightarrow \dots } 
							}
						}
					}
				}
			}
			$$
			\end{adjustbox}
		
		\
		
		\noindent Here follows the derivation of axiom (T). The left premiss of $ \Lcon $ is the derivable sequent $ x\in a, a\in N(x), x:A, x: A>\bot \Rightarrow x: \bot, a \Vdash^{\exists} A $, not shown for reasons of space.

		\
		
		\begin{adjustbox}{max width = \textwidth}
		$$
		\infer[\Rright]{\Rightarrow x: A \rightarrow ((A > \bot ) \rightarrow \bot ) }{ 
			\infer[\Rright]{x: A \Rightarrow x: (A > \bot ) \rightarrow \bot  }{
				\infer[\T]{x: A, x:A < \bot \Rightarrow x: \bot }{ 
					\infer[\Lcon]{x\in a, a\in N(x), x:A, x: A>\bot \Rightarrow x: \bot}{ 
						\infer[\Lc]{x\Vdash_{a} A| \bot,  x \in a, a\in N(x), x:A, x:A>\bot, \Rightarrow x:\bot }{ 
							\infer[\Lfe]{b \in N(x), b\subseteq a, b\Vdash^{\exists} A, b\Vdash^{\forall} A\rightarrow \bot,  x \in a, a \in N(x), x:A>\bot \Rightarrow x:\bot }{ \infer[\Lfu]{ y \in b, y:A, b \in N(x), b\subseteq a, b\Vdash^{\forall} A\rightarrow \bot,  x \in a, a \in N(x), x:A>\bot \Rightarrow x:\bot  }{
									\infer[\Lright]{y:A\rightarrow \bot, y\in b, y:A, b \in N(x), b\subseteq a, x \in a, a \in N(x), x:A>\bot \Rightarrow x:\bot  }{
										y:A\rightarrow \bot, y\in b, y:A \dots \Rightarrow x:\bot, y:A 
										\quad 
										&
										\quad 
										y: \bot \dots \Rightarrow 
						} } } }    
		}  }  }  } 
		$$
		\end{adjustbox}
	
	\

\noindent Finally, here follows the derivation of (U$ _1 $). The derivable left premiss of $ \Lcon $, sequent $z \in c, z \in b \dots z: \lnot A \Rightarrow \dots c \fe \lnot A, z:\lnot A$, is not shown.

\

\begin{adjustbox}{max width = \textwidth}
	$$
	\infer[\Rright]{\Rightarrow x: (\lnot A > \bot ) \rightarrow (\lnot (\lnot A > \bot ) > \bot) }{
		\infer[\Rcon]{x: \lnot A > \bot \Rightarrow x:\lnot (\lnot A > \bot ) > \bot }{ 
			\infer[\Lfe]{a \in N(x), a \fe \lnot (\lnot A > \bot ), x: \lnot A > \bot \Rightarrow   x \Vdash _a \lnot (\lnot A > \bot ) > \bot }{
				\infer[\mathsf{L\lnot}]{y \in a, a \in N(x), y: \lnot (\lnot A > \bot ), x: \lnot A > \bot \Rightarrow   x \Vdash _a \lnot (\lnot A > \bot ) > \bot}{
					\infer[\Rcon]{ y \in a, a \in N(x), x: \lnot A > \bot \Rightarrow   x \Vdash _a \lnot (\lnot A > \bot ) > \bot, y: \lnot A > \bot }{
						\infer[\Lfe]{b \in N(y), y \in a, a \in N(x), b\fe \lnot A,  x: \lnot A > \bot \Rightarrow   x \Vdash _a \lnot (\lnot A > \bot ) > \bot, y\Vdash_b \lnot A | \bot}{
							\infer[\Unif{1}]{z \in b, b \in N(y), y \in a, a \in N(x), z: \lnot A,  x: \lnot A > \bot \Rightarrow   x \Vdash _a \lnot (\lnot A > \bot ) > \bot, y\Vdash_b \lnot A | \bot}{
								\infer[\Lcon]{z \in c, z \in b, b \in N(y), y \in a, a \in N(x), z: \lnot A,  x: \lnot A > \bot \Rightarrow   x \Vdash _a \lnot (\lnot A > \bot ) > \bot, y\Vdash_b \lnot A | \bot}{ 
									\infer[\Lbar]{x\Vdash_c \lnot A | \bot \dots z: \lnot A,  x: \lnot A > \bot \Rightarrow \dots}{
										\infer[\Lfe]{d \subseteq c, d \in N(x) \dots d \fe \lnot A,   d \fu \lnot A \rightarrow \bot, z: \lnot A,  x: \lnot A > \bot \Rightarrow \dots }{
											\infer[\Lfu]{k \in d, d \subseteq c \dots k: \lnot A,   d \fu \lnot A \rightarrow \bot, z: \lnot A,  x: \lnot A > \bot \Rightarrow \dots }{
												\infer[\Lright]{k \in d, d \subseteq c \dots k: \lnot A,   d \fu \lnot A \rightarrow \bot, k:\lnot A \rightarrow \bot,  z: \lnot A,  x: \lnot A > \bot \Rightarrow \dots }{\deduce{ \dots k:\lnot A \Rightarrow \dots k:\lnot A}{\triangledown} & \qquad \qquad\qquad \qquad
													\deduce{\dots k:\bot\Rightarrow \dots }{\init}
												}
											}
										}
									}
								}
							}
						}
					}
				}
			}
		}
	}
	$$
	\end{adjustbox}
	\end{proof}

\end{document}